%% file: main.tex
\def\confversion{0}
\def\ifconf{\ifnum\confversion=1}
\def\ifnotconf{\ifnum\confversion=0}

\documentclass[11 pt]{article}
\def\showauthornotes{1}
\def\showkeys{0}
\def\showdraftbox{1}

\input{macros}
\input{abbrev}

\usepackage{graphicx}

\usepackage[top=1in, bottom=1in, left=1.25in, right=1.25in]{geometry}

 \newcommand{\edits}[1]{{#1}}

\begin{document}

\title{Almost Ramanujan Expanders from Arbitrary Expanders via Operator Amplification}

\author{Fernando Granha Jeronimo\thanks{University of Illinois, Urbana-Champaign. This material is based upon work supported by the NSF grant CCF-1900460.}
                      \and
       Tushant Mittal \thanks{Stanford University.} \and
       Sourya Roy \thanks{University of Iowa.}    \and
       Avi Wigderson\thanks{Institute for Advanced Study, Princeton. This work was partially supported by NSF grant CCF-1900460.\\ Any opinions, findings and conclusions or recommendations expressed in this material are those of
                                       the author(s) and do not necessarily reflect the views of the NSF.}
}
\date{\today}

\date{}

\maketitle
%\vspace{-1.35cm}
%\draftbox
%\vspace{-0.5cm}
%% \thispagestyle{empty}

\input{abstract}

\newpage

\ifnotconf
\pagenumbering{roman}
\tableofcontents
%\newpage\todototoc\listoftodos
\clearpage
\fi
%\endgroup
 
 %\newpage

\pagenumbering{arabic}
\setcounter{page}{1}

\input{intro}

\input{prelim}

\input{expander_walks}

\input{adv_amp}

\input{applications}

\input{op_eml}

\section*{Acknowledgement}

We thank Alexander Lubotzky for stimulating and enlightening
discussions in the initial phase of this project.
We are grateful to the anonymous reviewers for their valuable suggestions that improved the quality of this paper. 

\bibliographystyle{alphaurl}
\bibliography{macros,references}

%\appendix
%\input{param_app}

\end{document}

%% file: macros.tex
\usepackage{xspace,xcolor,enumerate}
\usepackage{amsmath,amssymb,bm}
\usepackage{amsthm}
\usepackage[toc,page]{appendix}
\usepackage{thmtools}
\usepackage{thm-restate}
\usepackage{color,graphicx}
\usepackage{boxedminipage}
\usepackage{enumitem}
\usepackage{comment}
\usepackage{todonotes}
\usepackage{csquotes}
\usepackage[linesnumbered,ruled,vlined]{algorithm2e}
\ifnum\showkeys=1
\usepackage[color]{showkeys}
\fi

\definecolor{darkred}{rgb}{0.5,0,0}
\definecolor{darkgreen}{rgb}{0,0.35,0}
\definecolor{darkblue}{rgb}{0,0,0.55}

\usepackage[pdfstartview=FitH,pdfpagemode=UseNone,colorlinks,linkcolor=darkblue,filecolor=darkred,citecolor=darkgreen,urlcolor=darkred,pagebackref]{hyperref}

\allowdisplaybreaks

\usepackage[capitalise,nameinlink]{cleveref}
\usepackage[T1]{fontenc}
\usepackage{microtype}

\usepackage{mathtools,dsfont}
\usepackage{eulerpx}
\usepackage{palatino}
\usepackage[scaled=.95]{helvet}
\usepackage{eulerpx}

\DeclareMathAlphabet{\mathpazocal}{OMS}{zplm}{m}{n}
\DeclareRobustCommand*{\mathcal}[1]{\mathpazocal{#1}}
\DeclareMathAlphabet{\symsf}{OT1}{LibertinusSans-LF}{m}{n}
\ifnum\showauthornotes=1
\newcommand{\Authornote}[2]{{\sf\small\color{red}{[#1: #2]}}}
\newcommand{\Authorcomment}[2]{{\sf \small\color{gray}{[#1: #2]}}}
\newcommand{\Authorfnote}[2]{\footnote{\color{red}{#1: #2}}}
\else
\newcommand{\Authornote}[2]{}
\newcommand{\Authorcomment}[2]{}
\newcommand{\Authorfnote}[2]{}
\fi

\ifnum\showdraftbox=1

\else

\fi

\newtheorem{theorem}{Theorem}[section]
\newtheorem{lemma}[theorem]{Lemma}
\newtheorem{corollary}[theorem]{Corollary}
\newtheorem{claim}[theorem]{Claim}

\newtheorem{observation}[theorem]{Observation}
\theoremstyle{definition}
\newtheorem{definition}[theorem]{Definition}
\newtheorem{fact}[theorem]{Fact}
\newtheorem{example}[theorem]{Example}

\theoremstyle{remark}
\newtheorem{remark}[theorem]{Remark}

\newtheorem{proposition}[theorem]{Proposition}

\theoremstyle{remark}

\def\FullBox{\hbox{\vrule width 6pt height 6pt depth 0pt}}

\def\qed{\ifmmode\qquad\FullBox\else{\unskip\nobreak\hfil
\penalty50\hskip1em\null\nobreak\hfil\FullBox
\parfillskip=0pt\finalhyphendemerits=0\endgraf}\fi}

\def\qedsketch{\ifmmode\Box\else{\unskip\nobreak\hfil
\penalty50\hskip1em\null\nobreak\hfil$\Box$
\parfillskip=0pt\finalhyphendemerits=0\endgraf}\fi}

\ifnotconf
\renewenvironment{proof}{\begin{trivlist} \item {\bf Proof:~~}}
   {\qed\end{trivlist}}
\fi
\newenvironment{proofsketch}{\begin{trivlist} \item {\bf
Proof Sketch:~~}}
  {\qedsketch\end{trivlist}}
\makeatletter
\newcommand{\oset}[3][0ex]{%
  \mathrel{\mathop{#3}\limits^{
    \vbox to#1{\kern-2\ex@
    \hbox{$\scriptstyle#2$}\vss}}}}
\newcommand{\stackalign}[1]{
	\vcenter{
		\Let@ \restore@math@cr \default@tag
		\baselineskip\fontdimen10 \scriptfont\tw@
		\advance\baselineskip\fontdimen12 \scriptfont\tw@
		\lineskip\thr@@\fontdimen8 \scriptfont\thr@@
		\lineskiplimit\lineskip
		\ialign{\hfil$\m@th\scriptstyle##$&$\m@th\scriptstyle{}##$\crcr
			#1\crcr
		}
	}
}
\makeatother

%% file: abbrev.tex
\def\matr#1{\symsf{#1}}
\newcommand{\sW}{\mathrm{s}\mathcal{W}}

\DeclarePairedDelimiter\set{\lbrace}{\rbrace}

\newcommand{\bias}{\mathsf{bias}}

\let\latexcirc=\circ
\newcommand{\ccirc}{\mathbin{\mathchoice
  {\xcirc\scriptstyle}
  {\xcirc\scriptstyle}
  {\xcirc\scriptscriptstyle}
  {\xcirc\scriptscriptstyle}
}}
\newcommand{\xcirc}[1]{\vcenter{\hbox{$#1\latexcirc$}}}\let\circ\ccirc

\def\AA{\oset{\circ}{\matr  A}}
\def\BB{\oset{\circ}{\matr  B}}
\def\EE{\oset{\circ}{\matr  E}}
\def\PP{\oset{\circ}{\matr  P}}
\def\XX{\oset{\circ}{\matr  X}}
\def\JJ{\oset{\circ}{\matr  J}}

\def\to{\rightarrow}

\def\epsilon{\varepsilon}

\def\phi{\varphi}
\def\cal{\mathcal}

\newcommand{\ie}{i.\@e.\@,\xspace}
\newcommand{\eg}{e.\@g.\@,\xspace}
\newcommand{\etal}{et al.\@\xspace}
\newcommand{\cf}{{\it cf.\@,}}

\newcommand{\mper}{\,.}
\newcommand{\mcom}{\,,}

\newcommand{\reg}{\mathrm{reg}}

\newcommand{\E}{{\mathbb E}}
\newcommand{\C}{{\mathbb C}}
\newcommand{\N}{{\mathbb{N}}}
\newcommand{\Z}{{\mathbb Z}}

\newcommand{\F}{{\mathbb F}}

\newcommand{\cH}{\mathcal{H}}
\newcommand{\cD}{\mathcal{D}}
\newcommand{\cL}{\mathcal{L}}

\newcommand{\cW}{\mathcal{W}}
\newcommand{\cX}{\mathcal{X}}

\newcommand{\vsX}{\cX_\cH}
\newcommand{\vsXY}{\mathcal{XY}_\cH}
\newcommand{\abs}[1]{\ensuremath{\left\lvert #1 \right\rvert}}

\newcommand{\norm}[1]{\ensuremath{\left\lVert #1 \right\rVert}}

\newcommand{\opnorm}[1]{\norm{#1}_{\textup{op}}}

\newcommand{\ip}[2] {\ensuremath{\left\langle #1 , #2 \right\rangle}}

\newcommand{\Esymb}{\mathbb{E}}

\makeatletter

\def\Ex#1{%
    \ProbabilityRender{\Esymb}{#1}%
}

\def\ProbabilityRender#1#2{%fancy probability command
  \@ifnextchar\bgroup%
  {\renderwithdist{#1}{#2}}
   {\singlervrender{#1}{#2}}
}
\def\singlervrender#1#2{%
   \ensuremath{\mathchoice
       {{#1}\left[ #2 \right]}
       {{#1}[ #2 ]}
       {{#1}[ #2 ]}
       {{#1}[ #2 ]}
   }
}
\def\renderwithdist#1#2#3{%
   \@ifnextchar\bgroup
   {\superfancyrender{#1}{#2}{#3}}
   {\ensuremath{\mathchoice
      {\underset{#2}{#1}\left[ #3 \right]}
      {{#1}_{#2}[ #3 ]}
      {{#1}_{#2}[ #3 ]}
      {{#1}_{#2}[ #3 ]}
     }
   }
}
\def\superfancyrender#1#2#3#4#5{
   \ensuremath{\mathchoice
      {\underset{#1}{{#1}}\left#4 #3 \right#5}
      {{#1}_{#2}#4 #3 #5}
      {{#1}_{#2}#4 #3 #5}
      {{#1}_{#2}#4 #3 #5}
   }
}
\makeatother

\newcommand{\calH}{{\cal H}}

\newcommand{\poly}{{\mathrm{poly}}}
\newcommand{\polylog}{{\mathrm{polylog}}}

\newcommand{\vecn}[1]{\frac{1}{\abs{V_Y} } \vec{#1}}
\newcommand{\tr}{\mathrm{tr}}
\newcommand{\fksl}{\mathfrak{sl} }
\DeclareMathOperator{\Cay}{Cay}
\DeclareMathOperator{\spec}{Spec}
\DeclareMathOperator{\Tr}{Tr}

\newcommand{\slin}{\textrm{\textup{SL}}}
\newcommand{\expan}{\lambda}
\newcommand{\pif}{{\matr \Pi}_{\matr f} }
\newcommand{\pifw}{ \oset{\circ}{\matr \Pi}_{\matr f} }

\newcommand{\projh}{\matr P_\cH}
\newcommand{\lifth}{\matr L_\cH}
\newcommand{\permc}{J}

\newcommand{\bv}{v}
\newcommand{\bw}{{w}}
\newcommand{\bz}{{z}}

%% file: abstract.tex
We give an efficient algorithm that transforms any bounded degree
expander graph into another that achieves almost optimal (namely,
near-quadratic, $d \leq 1/\lambda^{2+o(1)}$) trade-off between (any
desired) spectral expansion $\lambda$ and degree $d$. Furthermore, the
algorithm is \emph{local}: every vertex in the new graph can compute its new neighbors
as a subset of its original neighborhood of radius
$O(\log(1/\lambda))$. The optimal quadratic trade-off is known as the
Ramanujan bound, so our construction gives almost Ramanujan expanders
from arbitrary expanders.

The locality of the transformation preserves structural properties of
the original graph, and thus has many consequences. Applied to Cayley
graphs, our transformation shows that \emph{any} expanding finite
group has almost Ramanujan expanding generators. Similarly, one can
obtain almost optimal explicit constructions of quantum expanders,
dimension expanders, monotone expanders, etc., from existing
(suboptimal) constructions of such objects. Another consequence is a
"derandomized" random walk on the original (suboptimal) expander with
almost optimal convergence rate. Our transformation also applies when
the degree is not bounded or the expansion is not constant.

We obtain our results by a generalization of Ta-Shma's technique in
his breakthrough paper [STOC 2017], used to obtain explicit almost
optimal binary codes. Specifically, our spectral amplification extends
Ta-Shma's analysis of bias amplification from scalars to matrices of
arbitrary dimension in a very natural way. Curiously, while Ta-Shma's
explicit bias amplification derandomizes a well-known probabilistic
argument (underlying the Gilbert--Varshamov bound), there seems to be
no known probabilistic (or other existential) way of achieving our
explicit {operator-valued} spectral amplification.

%% file: intro.tex
\section{Introduction}

\subsection{Background}

Expander graphs are fundamental objects in computer science and
mathematics, possessing a variety of applications in both
fields~\cite{HooryLW06,L12}. Indeed, expanders (and expansion) play a
central role in numerous algorithmic advances, cryptographic schemes,
circuit and proof complexity lower bounds, derandomization and
pseudorandom generators, error correcting codes, ... and are central to
structural results in group theory, algebra, number theory, geometry,
combinatorics. 

A central \emph{quality} measure of expansion of an infinite family of
$d$-regular graphs $\set{X_i}_{i \in \mathbb{N}}$ is the
second largest singular value of its normalized adjacency matrix,
which we denote by $\lambda(X_i) \in [0,1]$. We say that a family
$\set{X_i}_{i \in \mathbb{N}}$ is $\lambda$-expanding, for some fixed
$\lambda < 1$, if $\lambda(X_i) \le \lambda$ for every member $X_i$ of
the family. The smaller is the expansion parameter $\lambda$, the more
spectrally expanding is the family. (For simplicity, we will sometimes
discuss single graphs rather than families, and say that $X$ is a
$(d,\lambda)$-expander if it is $d$-regular and satisfies $\lambda(X)
\leq \lambda$.)

A random $d$-regular graph with $d\geq 3$ is easily shown
\cite{Pinsker73} to be $.99$-expanding with high probability, giving
rise to the existence of expanding families. The quest to explicitly\footnote{{See~\cref{def:explicit} for a formal definition.}}
construct bounded degree expanders started with Margulis' paper
\cite{Mar73}, and has been an extremely active research area in the
past half-century. Today, we have a large arsenal of constructions and
tools to establish expansion which are quite different in nature---algebraic, analytic, combinatorial, and mixtures of these (for a short
survey of this wealth see~\cite[Sec 8.7]{Wigderson18})---and we will
discuss a few of them below.

All the different constructions above
yield $d$-regular $\lambda$-expanding families with {\em some}
specific constants $d$ and $\lambda$.  Now, a large variety of
structural and algorithmic applications call for optimizing both
parameters, and understanding the best trade-off between them. One
example which is directly related to this paper is the study of random
walks on expanders, sometimes used for randomness-efficient
error-reduction of probabilistic algorithms, and also in the
construction of randomness extractors. The surprising {\em expander
  Chernoff bound} of Gillman \cite{G98} informally says that a
sequence of {\em highly correlated} $k$ vertices along a random walk
in a $(d,\lambda)$-expander, is almost as good a sampler as a sequence
of $k$ {\em independent} vertices, {with respect to empirical sums}. Saving randomness calls for
minimizing the degree $d$ while improving the quality of the sample requires minimizing the expansion parameter $\lambda$.

However, for any choice of degree $d$, the spectral expansion
$\lambda$ cannot be made arbitrarily small. The Alon--Boppana bound
\cite{Nil91} shows that $\lambda(X_i) \ge 2\sqrt{d-1}/d - {o_n}(1)$. It
intuitively says that the {\em infinite} $d$-regular tree is the best
possible spectral expander, raising the challenge of achieving it by
{\em finite graphs}. This challenge was first met by the
(independent) seminal papers of \cite{LPS88, Margulis88}; they
constructed optimal spectrally expanding families, dubbed {\em
  Ramanujan graphs}, satisfying the (Ramanujan bound) $\lambda(X_i)
\le 2\sqrt{d-1}/d$. The investigation of expanding families near or
achieving the optimal Ramanujan bound has received much
attention. However, since then, only one essentially different
construction of Ramanujan graphs was found 30 years later by
\cite{MSS15}.

%\tnote{Need to correct this}
%
%\edits{
%%The quest towards almost optimal trade-offs can be summarized as a
%%sharpening of our original major question above:
%%\begin{center}
%%  \textit{
%%  Which graphs are (almost) Ramanujan expanders?
%%  }
%%\end{center}
%\begin{center}
%  \textit{
%  Give efficient and deterministic constructions of (almost) Ramanujan expanders.
%  }
%\end{center}
%}
A study of \textit{almost Ramanujan expanders}, in which the bound above is
nearly matched, has received much attention as
well. Friedman~\cite{Friedman03} greatly strengthened Pinsker's bound
above \cite{Pinsker73}, showing that with high probability, a random
$d$-regular graph $X$ satisfies $\lambda(X) \le 2\sqrt{d-1}/d +
{o_n}(1)$. For explicit constructions, one {successful} approach
 was the {\em lifting
  method} of Bilu--Linial \cite{BL06}, which achieves $d \leq
\widetilde{O}(1/\lambda^2)$, and famously led to the (exact) Ramanujan
expanders of \cite{MSS15} mentioned above.

%\tnote{Pushed the Bilu--Linial up and the zig-zag in a separate to highlight it as it is important}

{
A different paradigm to explicitly construct almost optimal expanders, which is central for this paper, is to start with a weak expander and amplify it to a strong one via combinatorially defined graph products.} {This originated with the} {\em zig-zag product} of \cite{RVW00}. They showed that their
basic zig-zag construction achieves an explicit family of expanders
with $d \leq 1/\lambda^4$. They further derandomize the basic zig-zag
product to achieve $d \leq 1/\lambda^3$. Ben-Aroya and Ta-Shma~\cite{BT08} in their ingenious ``s-wide
zig-zag product", nearly matched the optimal quadratic
bound\footnote{We call any such bound near-optimal or almost
  Ramanujan. Of course, reducing the $o(1)$ slack in the exponent is
  clearly of much interest.}, achieving $d \leq
1/\lambda^{2+{o_\lambda}(1)}$. Their ``higher-order'' version of
zig-zag~\cite{BT08} will be central in our work. Note that these results yield specific constructions but do not provide a generic technique to amplify arbitrary expanders.

While for some applications and
structural results, {one is free to start with} {\em any} family of expanders, for
many others, the {initial} graphs are externally given to us (as e.g.\ is the
case for understanding the expansion of Cayley graphs of
groups). {So it is highly desirable to have generic algorithms for expander construction that can operate within these constraints.} Moreover, seeking different constructions and analysis tools
has led to surprising applications beyond those intended (\eg the
resolution of the Kadison--Singer conjecture by \cite{MSS14} and the
proof of $\mathrm{SL}=\mathrm{L}$ by Reingold~\cite{R05}). {Motivated by this, we study the following question:}
{
\begin{center}
  \textit{When can a family of weak expanders be amplified to a strong (almost-Ramanujan) one?}
\end{center}
}
We show that this is always possible: \emph{any} expander family can
be \emph{locally and efficiently} converted into an almost Ramanujan
family. More precisely, starting from any family of bounded degree
expanders, it is possible to obtain, for any desired target expansion
$\expan > 0$, a new family of $\expan$-expanders close to the
Ramanujan bound.

\subsection{Main Results}

%See \cref{sec:prelim} for precise definitions of relevant terminology. 
Our main result for general families of expander graphs is as follows. 
\begin{restatable}[Main I - Informal]{theorem}{MainPermAmpInformal}\label{theo:main_1_informal}
  Let $\set{X_i}_{i \in \mathbb{N}}$ be a family of $(d_0,\lambda_0)$-expanders where $\lambda_0 < 1$ is a constant.
  For any (target) $\lambda \in (0,1)$ and $X_i$, we can explicitly\footnote{See~\cref{def:explicit}} construct a $(d,\lambda)$-expander, $X_i'$,
  on the same vertex set, where $d = O(d_0/\lambda^{2+\edits{o_\lambda}(1)})$, 
  %\edits{where the $o(1)$ is as $\abs{X_i}\to\infty$.}
  %
  Moreover, the construction is local in the sense that edges in $X_i'$ correspond to short walks in $X_i$.
\end{restatable}

%\tnote{moved this paragraph as it appears too late and the connection is lost.}

The key insight is to use a result of K\"{o}nig that says that the
adjacency matrix, say $A_X$, of an arbitrary regular graph, can be written as a sum
of permutation matrices. {Thus we have, $A_X = \Ex{s\sim S}{\rho_{\mathrm{def}} (s)}$ where $S$ is a set of permutations, and $\rho_{\mathrm{def}}$ is the \textit{defining representation} that maps
a permutation to the matrix defining it. The task is to construct a set $S'$, with  $\opnorm{\Ex{s\sim S'}{\rho_{\mathrm{def}}(s)}} \leq \lambda$. We study a general version of this question with an arbitrary group representation of a finite group (See~\cref{def:rep}) . 
This is directly connected to the expansion of Cayley graphs. 
 } 
%\snote{Change looks good. Approved.}

%This gives us a set,  for every irreducible representation $\rho$ and, in particular, for the irreducible component of the defining representation
%We obtain our results by considering the seemingly more specialized
%case of Cayley expanders, which are based on group theory and
%represent a prominent way of constructing expanders. We can now move from Cayley graphs back to general graphs and answer
%	our original question.   Thus, the graph $X'$ with $A_{X'}= \Ex{s\sim S'}{\rho_{\mathrm{def}} (s)}$, is an almost-optimal expander yielding~\cref{theo:main_1_informal}.

\paragraph{Cayley Graphs}
A
Cayley graph $\Cay(G,S)$ on a finite group $G$ is specified by a
symmetric set of generators $S \subseteq G$, where vertices are
elements of $G$ and $g,h \in G$ are adjacent if and only if
$hg^{-1}$ belongs to $S$.

While many groups admit Cayley expanders, most of these are far from
the Ramanujan bound.  This is true, in particular, in the case of
non-Abelian finite simple groups which includes the symmetric group.
Breuillard and Lubotzky~\cite{BL18} ask whether it is possible to have
near-Ramanujan expanders for all families of finite simple
groups. More generally,

\begin{center}
  \textit{
  Which (family of) groups admit expanding Cayley graphs close to the Ramanujan bound?}
\end{center}

Our key result is that any group that admits a Cayley expander also
admits one that is almost Ramanujan.

\begin{theorem}[Main II]\label{theo:main_2}
  Let $G$ be a finite group and $S$ be such that $\Cay(G,S)$ is a $\expan_0$-expander,
  for some constant $\expan_0 \in (0,1)$.
  For every $\expan \in (0,1)$, there exists $S'$ such that
  \begin{enumerate}[topsep=2pt,itemsep=1pt,label=$\cdot$]
    \item $\Cay(G,S')$ is a $\expan$-expander.
    
    \item $\abs{S'} = O\left (\abs{S}/\expan^{2+{o_{\lambda}}(1)}\right )$, and
      
    \item $S'$ can be computed deterministically in $\poly(\abs{S}/\expan)$-time assuming an oracle for group operations.
  \end{enumerate}
  Furthermore, if $\Cay(G,S)$ is strongly explicit\footnote{Neighbors
    of a vertex can be computed in time polynomial in the {\em description
  length} of a vertex.}, then so is $\Cay(G,S')$.
\end{theorem}

Since expanding families of Cayley graphs are known for non-Abelian
finite simple groups~\cite[Theorem 3.1]{BL18}, this result makes
substantial progress towards the question asked therein (the $o(1)$
term needs to be removed to resolve it completely). Moreover, these
are strongly explicit (except for the Suzuki group).  Thus, our result
yields strongly explicit almost Ramanujan Cayley graphs for these
these groups, which notably includes the symmetric group!

\begin{corollary}[Explicit almost Ramanujan Cayley Expanders]\label{cor:simple_group_exp}
  For every non-Abelian finite \emph{simple}\footnote{This holds for other groups as well, as long as they have expanding generators.
    One non-simple example is the Cayley expanders of Rozenman, Shalev and Wigderson~\cite{RSW06}.}
  group $G$ and $\expan > 0$, we can explicitly
  construct almost-Ramanujan $(d,\lambda)$-Cayley multigraphs on $G$ where $d \le O(1/\expan^{2+{o_{\lambda}}(1)})$.
\end{corollary}

\begin{remark}[Connection with Codes]{A linear $\lambda_0$-balanced code over $\F_2^{n_0}$ of dimension
  $k$ is equivalent to a Cayley $\lambda_0$-expander over $G=\F_2^k$ of degree $n_0$. Let $S \subseteq G$ be the rows of a generator matrix of a
  good $\lambda_0$-balanced code (good means $k/n_0$ and $\lambda_0 < 1$ are constants).
Thus, the breakthrough construction of explicit almost optimal binary codes of Ta-Shma~\cite{Ta-Shma17} close to the Gilbert--Varshamov~\cite{G52,V57} bound can be viewed as a particular case of~\cref{theo:main_2} applied to the group $G=\F_2^k$.   %
  Applying~\cref{theo:main_2} above to $S$ with final expansion
  parameter $\lambda > 0$, we obtain a generating set $S' \subseteq G$
  of a Cayley $\lambda$-expander with degree $O(k/\lambda^{2+\edits{o_\lambda}(1)})$,
  or equivalently, we obtain a $\lambda$-balanced code of rate
  $\Theta(\lambda^{2+\edits{o_\lambda}(1)})$.}
\end{remark}

%\tnote{I shifted this paragraph below theorem as it was weirdly placed and breaking flow. }
%\snote{Looks good to me.}

%\{\rho_{\mathrm{std}}(\sigma)\mid \sigma \in S\}
%\footnote{The \emph{defining representation} - ($\rho_{\text{def}}(\sigma),\C^n$) for $\mathrm{Sym}_n$ is defined as the representation that maps a permutation to the matrix defining it.}. 
\subsection{Applications}

We will now discuss some applications of this operator amplification
technique which allows us to improve other pseudorandom objects. All
the "pseudorandom" objects below are expanders (with various {structural} properties), and for each of these, we amplify their spectral bound to almost
Ramanujan. We stress that our amplification preserves the underlying
structure and therefore, produces another object with the same properties.
Precise definitions of these objects will be given in~\cref{sec:appl}.

\paragraph*{Quantum Expanders}
Roughly speaking, a quantum expander is an operator defined by $d$
complex matrices, whose (linear) action on quantum states has a
constant spectral gap.
Quantum expanders were defined in \cite{AS04, BAST08, Hastings07b},
and Hastings~\cite{Hastings07} showed that the Ramanujan bound also
applies to them.  Existing explicit constructions are far from the
Ramanujan bound. In~\cite{Harrow07}, Harrow gave a generic
construction using expanding Cayley graphs which is explicit if the
group has a large irreducible representation and admits efficient
Quantum Fourier Transform (QFT).  Both these conditions are satisfied
by the symmetric group $\mathsf{Sym}_n$ using the generating family by
Kassabov~\cite{Kas07} and the QFT algorithm by Beals~\cite{Beals97}.

By amplifying the expansion of the generators of~\cite{Kas07}, we give
the first explicit family of almost Ramanujan quantum expanders.

\begin{restatable}[Explicit Almost Ramanujan Quantum Expanders]{corollary}{TheoQuanExp}
  For every $\lambda \in (0,1)$, there is an explicit infinite family
  of (efficient) $(O(1/\expan^{2+\edits{o_\lambda}(1)}),\lambda)$-quantum expanders.
\end{restatable}

\paragraph*{Monotone Expanders}
Monotone expanders are expanders, whose edge set can be decomposed into a constant number of {\em monotone} partial maps on $[n]$.
Bourgain and Yehudayoff~\cite{BY13} gave the only known explicit construction of monotone expanders with \emph{constant} degree. There are two natural notions of degree for a monotone expander.
  The usual vertex degree and the number of monotone maps. Our almost
  Ramanujan trade-off is with respect to the vertex degree (and the
  monotone degree is polynomial in the vertex degree).

\begin{restatable}[Almost Ramanujan Monotone Expanders]{corollary}{TheoMonExp}
  For every $\expan > 0$, there is an explicit family $\set{X_i}_{i\in \mathbb{N}}$ of (vertex)
  $d$-regular $d^{O(1)}$-monotone expanders with $d = O(1/\expan^{2+\edits{o_\lambda}(1)})$ and $\lambda(X_i) \le \expan$.
\end{restatable}

The approach is similar to that used for \cref{theo:main_1_informal}; we
express it as a sum of permutation matrices and amplify their expansion obtaining the following result. It would be
  really interesting to obtain an almost Ramanujan trade-off with
  respect to the monotone degree.

%
%\begin{remark}\todo{Should we move it to before Cor 1.6 and as main text?}  \end{remark}

\paragraph*{Dimension Expanders} Loosely speaking, dimension
expanders (over any field $\F$) are a linear algebraic extension of
expanders: a collection of $d$ linear maps on $\F^n$, which
significantly \emph{expands} (the span of) any vector space of
dimension below $n/2$. They were defined by Barak \etal in~\cite{BISW01}. Over complex numbers, any quantum expander is a
dimension expander. More generally, Dvir and Shpilka~\cite{DS09}
proved that a monotone expander directly yields a dimension expander
over every field. We give spectral almost Ramanujan expanders that
have the additional property of being dimension expanders.
{Additionally, if the starting dimension is small enough then we
achieve almost doubling of the starting dimension. See~\cref{cor:dimension} for a precise statement.}

%% \paragraph*{Product Replacement Algorithm}
%% Sampling near-uniform random elements of groups which have expanding
%% generators is easy - a random walk on the associated Cayley graph will
%% converge in logarithmic time. Remarkably, the same bound can be
%% achieved for {\em any} group (even if its diameter given a generating
%% set is linear!). This is achieved through the magic of the {\em
%%   Product Replacement Algorithm}: an auxiliary graph, whose vertices
%% are tuples of group elements (see more in [PZ]). The logarithmic
%% convergence of this algorithm was only recently proved. Again, our
%% structure preserving amplification, this auxillary graph can too be
%% made nearly-Ramanujan.

%% Such a walk is used, for
%% example, to sample random group elements~\cite{PZ01}.

%% One applications use of constructing near-optimal Ramanujan Cayley
%% graphs is to sample random group elements efficiently.  Given a Cayley
%% graph, $\Cay(G,S)$, one can consider a random walk on $G$ which starts
%% at an arbitrary vertex $g$ and at each step moves to a random neighbor
%% $g \to sg$.  Spectral expansion guarantees that walks mix quickly, \ie
%% in at most $ O_{\lambda}(\log |G|)$ steps (see \cite{HooryLW06}).  The
%% amount of randomness used in each step is $\log d$ and since the
%% degree versus expansion trade-off is now optimal, we can achieve the
%% same convergence guarantee using a smaller degree and thus the random
%% walks is more efficient in terms of randomness.

\paragraph*{Kazhdan Constant}
We can also amplify operators in \emph{infinite dimensional} Hilbert
spaces. This allows us to obtain improved (average) Kazhdan constants
of groups with \enquote{Property (T)}, which is an analog of expansion for
discrete groups. This implies better bounds for the \emph{product
  replacement algorithm}~\cite{LP00} to sample group elements (See~\cref{sec:appl} for details).

\begin{restatable}[Amplifying Average Kazhdan Constant]{corollary}{TheoKahzdan}\label{theo:kazhdan}
  Let $G$ be a finitely generated group and $S$ a finite set of generators such that the average Kazhdan constant $\overline{\mathcal{K}}(G,S)$ is
  equal to $2\cdot (1-\expan_0)$ for some constant $\expan_0 \in (0,1)$.
  For every $\expan \in (0,1)$, there is a set $S' \subseteq G$ such that
  \begin{enumerate}[topsep=4pt, itemsep=1pt]
    \item $\overline{\mathcal{K}}(G,S') \ge 2\cdot (1-\expan)$,  and thus, $\mathcal{K}(G,S') \ge 2\cdot (1-\expan)$. 
    \item $\abs{S'} = O_{\expan_0}(\abs{S}/\expan^{2+\edits{o_\lambda}(1)})$, and
    \item $S'$ can be found in time $\poly(\abs{S}/\expan)$ assuming an oracle for group operations on $G$.
  \end{enumerate}  
\end{restatable}

\iffalse
\paragraph*{Randomness-efficient Walks}
An immediate consequence of being able to achieve an almost optimum
degree versus expansion trade-off in this generic way is that we
obtain randomness-efficient random walks.  \edits{Should we explain it or delete this paragraph?} \todo{Explaining this paragraph}
\fi

\input{strategy}

\input{discussion}

\subsection{Outline}

We start in \cref{sec:prelim} by summarizing basic definitions and the
notation used throughout the paper. In \cref{sec:simple_amp}, we
generalize the simpler construction of Ta-Shma based on expander
walks. Apart from serving as a nice warm-up to the more-involved
construction, it will be used as a bootstrap for the more involved
construction based on $s$-wide replacement product which is the
subject of \cref{sec:adv_ampl}. Here, we prove the main amplification
result (a formal version of~\cref{theo:main_amp_informal} above) and
instantiate using known constructions and those obtained from
\cref{sec:simple_amp} which establishes \cref{theo:main_2}.
\cref{sec:appl} discusses the permutation amplification trick and
formally completes the proof of \cref{theo:main_1_informal}. It also
discusses the other applications in more detail.
{Finally,~\cref{sec:eml} gives an operator version of the expander
mixing lemma which improves the analysis of~\cite{CMR13}.}

%% file: strategy.tex
%To some extent . Recall that matrix Chernoff
%  bounds deteriorate with the dimension of the matrices, and we have
%  no fixed bound on their dimension here.}. Standard probabilistic
%techniques, such as the matrix Chernoff bound, have a forbidding dependence
%on the dimension of the matrices for this application.\todo{Maybe explain the bottleneck simply.}

%Switching to the bias distribution viewpoint, a subset $S\subseteq G$
%is said to be $\epsilon$-biased if it \emph{fools} all non-trivial 
%irreducible representations, \ie for every non-trivial irreducible
%\emph{representation}, $\rho$, of $G$, we have $\opnorm{\Ex{s \sim
%    S}{\rho(s)}} \leq \epsilon$.  Here, a representation of a
%group is an operator valued function, $\rho: G \to \matr
%M_\ell(\C)$, that is multiplicative, \ie for every two group elements
%$g_1,g_2$ we have $\rho(g_1g_2) = \rho(g_1)\rho(g_2)$. As mentioned
%earlier, $\Cay(G,S)$ is $\expan$-expanding if and only if $S$ is
%$\expan$-biased set. Thus, the problem of constructing optimal Cayley
%expander can be reformulated as construction of small biased
%distribution with optimal support size. In fact, we will see that the
%techniques work for general matrix value functions (not just
%representations).
\subsection{Techniques}

We first rephrase the technical result of our paper from the perspective of \enquote{bias amplification}. This makes it easy to compare to prior work, which we crucially build upon.  Let $f: S\to \matr M_\ell(\C)$ be a matrix-valued function. The quantity $\opnorm{\Ex{s \sim S}{\matr f(s)}}$ is known as the \textit{bias} of the operator-valued function $f$ with respect to $S$. The key idea in the prior works (and ours) is to establish an \enquote{bias amplification} result of the form,  
%\paragraph{Biased Sets and Expansion}

\begin{theorem}[Template Amplification Result]\label{thm:summary}
	Let $S$ be a finite set and $\expan_0 \in (0, 1)$ be a constant.  For every $\expan > 0$, there exists a deterministic polynomial time algorithm to construct
  $W \subseteq S^t$ of size $\abs{W} \leq \poly(|S|, 1/\lambda)$ such that 
  \begin{itemize}
  	\item \textbf{\textup{Scalar Amplification}} For every function
  $\matr f: S \to \C$ such that $\norm{f}_\infty \le 1$,\\ if  $\abs{\Ex{s \sim S}{\matr f(s)}} \leq  \expan_0$ then we have $\abs{\Ex{(s_1,\cdots, s_t) \sim W}{\matr f(s_t)\cdots f(s_1)}} \leq  \expan$.
  \item\textbf{\textup{Operator Amplification}} For every function
  $\matr f: S \to \matr M_\ell (\C)$ such that $\max_s \opnorm{\matr f(s)} \le 1$, if  $\opnorm{\Ex{s \sim S}{\matr f(s)}} \leq  \expan_0$ then we have $\opnorm{\Ex{(s_1,\cdots, s_t) \sim W}{\matr f(s_t)\cdots f(s_1)}} \leq  \expan$. \end{itemize}    
\end{theorem}

Such a result is closely related to expansion of Cayley graphs in the following way. For a group $G$, a set $S\subseteq G$ is called an \textit{$\epsilon$-biased set} if for every non-trivial irreducible \emph{representation}, $\rho$, of $G$, we have $\opnorm{\Ex{s \sim    S}{\rho(s)}} \leq \epsilon$. {Here, a group representation is a homomorphism from $G$ to the general linear group of operators, $\matr
M_\ell(\C)$}. The notion of $\epsilon$-biased set over $G=\Z^k_2$ was introduced in the pioneering work of Naor and Naor~\cite{NN90}. These small-bias sets have found numerous applications (\eg~\cite{ABNNR92, Vadhan12, Ta-Shma17}).

A symmetric subset $S\subseteq G$ is an $\epsilon$-biased set if and only if the $\Cay(G,S)$ is an $\epsilon$-expander. Therefore if $\Cay(G,S)$ is a given $\lambda_0$-expander, we can apply~\cref{thm:summary} for each irreducible representation $\rho$. This yields that $\Cay(G,S')$ is a $\lambda$-expander where $S' = \{s_t\cdots s_1  \mid (s_1,\cdots, s_t) \in W \}$. The degree of the new graph is $|S'| = |W|$. Note that over Abelian groups, the scalar amplification suffices as the irreducible representations of Abelian groups are $1$-dimensional, \ie scalar-valued functions called \emph{characters}. However, for general finite groups, one needs the amplification result for operator-valued functions.

\paragraph{A first approach} One can take $W =  S^t$ with $t
\approx \log_{\expan_0}(\expan)$. However, now the degree, $|S|^t =
O(1/\expan^{100})$, has also increased and the trade-off remains the
same.  Thus, we want to efficiently compute a sparse subset of $S^t$
that retains the expansion. Since we know what degree we are aiming
for, we could try to take a sparse random sample $S' \subseteq S^t$ of
size $d = O(1/\expan^{2.001})$ and hope that some form of matrix
concentration ensures that $\Cay(G,S')$ is $\expan'$-spectral expander
with $\expan' \approx \expan$. Unfortunately, it is not clear how to
show even the existence of a \emph{single} sparse subset $S'$ that
achieves the required expansion. 

%\edits{Switching to the bias distribution viewpoint, a subset $S\subseteq G$
%is said to be $\epsilon$-biased if it \emph{fools} all non-trivial 
%irreducible representations, \ie for every non-trivial irreducible
%\emph{representation}, $\rho$, of $G$, we have $\opnorm{\Ex{s \sim    S}{\rho(s)}} \leq \epsilon$.   As mentioned
%earlier, $\Cay(G,S)$ is $\expan$-expanding if and only if $S$ is
%$\expan$-biased set. Thus one may try to exploit this connection for expander construction.}
{For example, one could use matrix Chernoff {plus union bound} to select a subset $S'$ such that $\opnorm{\Ex{s\in S'}{\rho(s)}}$  is small for any irreducible representation $\rho$. However, this naive approach requires $|S'|\geq \Omega(\log \dim(\rho))$. For non-abelian groups, we have irreducible representations such that $ \dim(\rho) = \poly(|G|)$, and therefore, this cannot deduce the existence of constant-sized subsets that achieve expansion. This
  difficulty is also present in the proof of the Alon--Roichman
  theorem~\cite{AR94} and the reason why even for non-Abelian groups
  the only generic upper bound known uses $\Omega(\log(\abs{G}))$
  random generators to obtain an expander.} Nevertheless, we will show that, when applied correctly, the equivalence between expansion and small bias distribution can be used to build almost-Ramanujan expanders.

%Therefore,~\cref{thm:summary} directly yields our main result if $|W| \leq O\left(|S|/\lambda^{2+o(1)}\right)$ To see this, say we are given a $\lambda_0$-expander,
% $\Cay(G,S)$. . This yields that 

\subsubsection*{Prior Results on Bias Amplification}

\paragraph{Scalar amplification} Much of the earlier work has focused on the case of Abelian groups, especially $G = \Z_2^n$, as this has connections to other objects like error-correcting codes. 

{Rozenman and Wigderson introduced the approach of using expanders for ``scalar amplification'' via an iterated application of the expander mixing lemma (see~\cref{sec:eml} for details). Alon (in an unpublished email\footnote{We thank the anonymous reviewer for pointing out this reference.})} introduced the idea of using walks on an (auxiliary) expander graph $X$, whose vertices are identified with elements of $S$.
The set $W \subseteq S^t$ is chosen to be the collection of all walks of length $(t-1)$ on $X$. This technique gives a $\lambda$-biased set $W$, of size $|W| \leq O(|S|/\lambda^{4+o(1)})$
(\cf~\cite{Ta-Shma17}), which is quite good but still sub-optimal ($\lambda^{-4}$ instead of $\lambda^{-2}$ ).

Ta-Shma~\cite{Ta-Shma17} managed to close the gap almost optimally ($\lambda^{-2-o(1)}$) using the \emph{$s$-wide replacement product} to derandomize the above amplification. The $s$-wide
replacement product of Ben-Aroya and Ta-Shma~\cite{BT08} is a
higher-order version of the zig-zag product~\cite{RVW00}. Using the
collection of walks on the $s$-wide replacement product allows for a
much smaller collection $W \subseteq S^t$ with nearly optimal
size. This scalar technique was later applied to the more general case
of arbitrary Abelian groups by Jalan and Moshkowitz~\cite{JM21}.

\paragraph{Operator amplification}

To extend Ta-Shma's approach to non-Abelian groups, it is necessary to work with operator-valued functions, $\matr f \colon S \to \matr M_\ell(\C)$, as the irreducible representations are no longer of dimension one. To the best of our knowledge, only one general result was known for
general groups. Chen, Moore and Russell~\cite{CMR13} analyzed the
above expander walk construction using a matrix version of the
expander mixing lemma.  This gives an amplification procedure for
Cayley graphs of general groups, but the resulting degree
$O(\abs{S}/\lambda^{11})$ to achieve final expansion $\lambda$ is sub-optimal.

% \edits{ Analogous to the (folklore results of the) scalar case,
% achieving similar parameters to
%the expander walk approach.
%An alternative to the expander walk is the approach of iteratively using the \textit{expander mixing lemma}(EML). In the scalar case, it is well known that it yields similar bounds of $1/\lambda^{4+o(1)}$ ( see).  Chen,
%Moore and Russell~\cite{CMR13} prove an operator version of it and obtain a dependence factor $1/\lambda^{11}$ in the degree.
%}
\paragraph{Our Results}

We view our main
contribution as the identification of appropriate natural linear
algebraic extensions to Ta-Shma's amplification
framework~\cite{Ta-Shma17}. This gives an almost-optimal \textit{dimension independent} generalization of the scalar amplification result to operator-valued functions. Our result sharpens that of Chen, Moore and Russell~\cite{CMR13} be reducing the degree from $O(\abs{S}/\lambda^{11})$  to $O(\abs{S}/\lambda^{2 + o(1)})$ 

%Note that the definition of $\matr T_W$ extends naturally to a mapping from $\matr M_\ell(\C)^S$ to $\matr M_\ell(\C)^W$. 

\begin{theorem}[Operator Amplification (this work)]\label{theo:main_amp_informal}
  Let $S$ be a finite set and $\expan_0 \in (0, 1)$ be a constant.  For every $\expan > 0$, there exists a deterministic polynomial time algorithm to construct
  $W \subseteq S^t$ of size $\abs{W} \leq  O(\abs{S}/\expan^{2 + o(1)})$ such that for every function
  $\matr f: S \to \matr M_\ell (\C)$ with  $\opnorm{\Ex{s \sim S}{\matr f(s)}} \leq  \expan_0$ and
  $\max_s \opnorm{\matr f(s)} \le 1$, we have $\opnorm{\Ex{w \sim W}{\matr f(w)}} \leq  \expan$. 
\end{theorem}
%\edits{\paragraph{EML approach} We show that the analysis in \cite{CMR13} of the amplification via
%(iterated applications of) expander mixing lemma can be improved to get $O(\abs{S}/\lambda^{4 + o(1)})$. This matches the parameters that are known in the scalar case, which is a folklore result (see for \eg~\cite{} )
%}

%\tnote{Do we need to justify so much? Remove this?}
%\snote{No need. Remove this. We already mentione this on page 3.}
%\edits{We consider the main contribution of this work to be the broad
%applicability of the near-optimal operator amplification to \emph{any}
%family of expanders. For instance, the existence of almost Ramanujan
%expanders for all expanding groups, including the symmetric group, is
%quite surprising to us.}

%-------------- How we generalize and prove

 The key extension is a simple and
yet extremely useful change in the bias operator ($\matr \Pi_{\matr f}$)
defined by Ta-Shma which is a central object in the analysis of both
\cite{Ta-Shma17} and \cite{JM21}. In both these cases, $\matr f$ is scalar,
and they define,
\begin{align*}
  \matr \pif :\C[S] \to \C[S] \;\text{ where} \; \matr \pif \cdot s = \matr f(s) \cdot s\mper
\end{align*}

However, this approach is not readily generalizable to operators and
the view we take is that if $\matr f : S \to \matr M_\ell(\C)$,
then, $\pif$ is actually an operator on $\C^\ell \otimes \C[S]$
defined as
\begin{align*}
  \matr \pif :\C^\ell \otimes \C[S] \to \C^\ell \otimes \C[S] \;\text{ where} \; \matr \pif \,(v \otimes  s) = \matr f(s)\, v \otimes s\mper
\end{align*}

Clearly, in the Abelian case, we have $\ell = 1$ and this is
isomorphic to the setup by Ta-Shma. This generalization is very
natural and we show that not only does the older machinery gel well
with this, but the proof remains intuitive with the different spaces
neatly delineated. 

More precisely, we first establish an
\emph{operator version} of the expander walk amplification,
and then we derandomize it using (a suitable version of) the $s$-wide
replacement product. Furthermore, since the result does not depend on
the dimension, $\ell$, we can use it even for functions $\matr f: S
\to \cL(\cH)$ where $\cL(\cH)$ is the space of bounded linear
operators on an arbitrary Hilbert space, $\cH$, possibly infinite
dimensional. This is useful if the underlying group is not finite but
finitely generated by $S$.

%This extension will be so natural that it may
%almost feel that we are replacing absolute values in the original
%scalar analysis~\cite{Ta-Shma17} by operator norms. \edits{I don't like this sentence -- However,
%appropriate generalizations and care are needed in such an extension
%to operators.}
%\snote{Agree. Also the line before is kinda strong; in particular the phrase "so natural". Maybe tone it down a bit.}

%%---------------------Problem and first approach----

%\todo{Either not use representation or rearrange the bias paragraph. }
%\snote{I modifed it.}

%% file: discussion.tex
\subsection{Discussion}\label{subsec:discussion}

The results of this paper have some curious features, which we would
like to elaborate on. For most of them, we will use the following
"bare bones" description of our main spectral amplification result.
Namely, let $S$ be a finite set and $\calH$ a Hilbert space.
Let $\matr f$ be a function mapping elements of $S$ to operators on $\calH$
of unit norm, such that $\opnorm{ \E_{s \in S} [ \matr f(s) ]} \leq \lambda_0$.
For any $\lambda > 0$ take $t= c \log(1/\lambda)$ (for appropriate
$c$). We extend $\matr f$ from $S$ to $S^t$ by defining
$\matr f(s_1,\ldots,s_t) = \matr f(s_1) \cdots \matr f(s_t)$.  Clearly, $\opnorm{\E_{r
    \in S^t} [\matr f(r)]} \le (\opnorm{ \E_{s \in S} [\matr f(s) ]})^t \leq
\lambda$.
Our main result is an explicit construction of a (pseudorandom) subset
$S' \subseteq S^t$, of size only $|S'| = O(|S|/\lambda^{2+o(1)})$, with a
similar guarantee, namely $\opnorm{\E_{s' \in S'} [\matr f(s')]} \leq \lambda$.

\paragraph{Dimension Independence} Note that if the operators in
$S$ are 1-dimensional, namely scalars, then the {\em existence} of a
set $S'$ of this size (which is best possible even in this
1-dimensional case) follows directly from the Chernoff bound. Indeed,
Ta-Shma's construction~\cite{Ta-Shma17} may be viewed as derandomizing
this result, producing an explicit such $S'$.

One may try to do the same for operators in a higher dimension, say
$\ell$, by appealing to the Matrix Chernoff bounds of Ahlswede--Winter
\cite{AW02} (see also Tropp \cite{T15}). However, these concentration
inequalities pay a factor of $\ell$ in the tail bound, resulting in a
set $S'$ of size $\Omega(\log(\ell))$. As the dimension $\ell$ is
arbitrary (indeed, may be infinite), such a bound is useless.

Thus, our explicit construction has no known probabilistic (or other
existential) analog! What is curious is that our dimension-independent
analysis follows very closely that of Ta-Shma for 1-dimension, roughly
speaking, replacing scalar absolute values by the operator norm in any
dimension. We feel that it would be extremely interesting to find a
matrix concentration inequality for sampling product sets like $S^t$,
which is dimension independent.

\paragraph{Algebraic vs.\  Combinatorial Expander Constructions} Our
explicit construction of the pseudorandom set $S'$ above uses
expanders obtained from the $s$-wide zig-zag product of
\cite{BT08}. This is a combinatorial construction, a refinement of the
original zig-zag product construction of \cite{RVW00}. Nonetheless, it
has significant consequences to algebraic expander constructions which
use group theory, namely to the expansion of Cayley graphs {over the same
group.} We can take $S$ to be an expanding
generating set of a group and $\matr f$ to be some non-trivial
irreducible representation $\rho$.  Instead of using the $t$-product
of elements of $S^t$ to obtain a new amplified generating set $S'$, a
much sparser subset is chosen using the $s$-wide zig-zag construction.
The analysis of the norm amplification discussed above yields the
required expansion bound, in a way that has no dependence on the group
or the representation. The flexibility in mapping element of $S$ to
operators underlies the versatility of our spectral amplification. It
allows us to preserve some of the structure of the expanders whose
expansion are being amplified. In this case, both the starting
expander and the amplified expander are Cayley graph over the same
group.

It is interesting to note that this is a recurring phenomenon. In
\cite{ALW01}, it was discovered that the zig-zag product may be viewed
as a combinatorial generalization of the algebraic semi-direct product
of groups. This connection made possible the construction of new
expanding Cayley graphs in groups that are far from being simple, \eg
in \cite{MW04,RSW06}. It is rewarding to see again how new
combinatorial constructions, sometimes inferior in certain parameters to
some algebraic ones, yield new results in group theory.

\paragraph{Iterated Pseudorandomness} Another interesting aspect of
our result is the following. Recall that expanders are pseudorandom
objects for many purposes. One important purpose is sampling - rather
than sampling $t$ independent random elements in some set $S$, one may
sample $t$ points along a random walk on an expander on the vertex set
$S$ and a Chernoff type bound still holds (a nontrivial result of
\cite{G98})- this affords significant savings in the number of random
bits spent. For this result, any expander would do. What happens in
this paper is an iterated use of expanders as samplers as follows. We
first choose a sparse pseudorandom set of $t$-walks inside $S^t$ using
expanders walks. Then, we choose a yet sparser pseudorandom set inside
it, again using walks on an additional expander. This repeated use of
expanders improves the trade-off between the quality of spectral
amplification and the size of the final pseudorandom set to
near-optimal. The construction of this iterated selection of walks
seems critical, and (as in Ta-Shma's paper) is chosen to come from the
$s$-wide zig-zag product of two expanders~\cite{BT08}.

\paragraph{Groups and Expansion} For us, the most surprising consequence
of our results is that ``weak'' simple groups, especially the
symmetric group,\footnote{See also the groups in \cite{RSW06}, which
  are iterated wreath products of symmetric groups.} can have
near-Ramanujan generators. The question of which groups are expanding,
and just how expanding they are, is an old quest of group theory. One
dichotomy is whether {\em every} finite set of generators of the group
is expanding (these are ``strongly expanding'' groups), or if some are
and some aren't (these are ``weakly expanding'' groups). For the
symmetric group, many finite non-expanding generating sets of constant
size were long known, and Kassabov's breakthrough construction
\cite{Kas07} designed a constant size \emph{expanding} generating set.
The symmetric group is then a weakly expanding group, while, \eg
simple groups of Lie type (namely, matrix groups) are believed, and in
some cases known, to be strongly expanding. Nonetheless, our
construction works equally well for all, and we have almost Ramanujan
generators for all expanding groups.

\paragraph{Closeness to the Ramanujan Bound} As mentioned above, a family
of $d$-regular graphs is called Ramanujan if its spectral expansion
parameter $\lambda$ is at most $2\sqrt{d-1}/d$. This terminology was
introduced in the seminal work of Lubotzky, Phillips and
Sarnak~\cite{LPS88}, and it designates the optimal degree versus
expansion trade-off that a family of bounded degree expanders can
achieve. Several notions of closeness to the Ramanujan bound were
investigated, \eg $(2\sqrt{d-1}+\epsilon)/d$ (with $\epsilon > 0$
small or vanishing) in~\cite{Friedman03,MOP20,JMOPT22},
$O(1/\sqrt{d})$ in~\cite{ACKM19}, $\polylog(d)/\sqrt{d}$
in~\cite{BL06,JMOPT22} and more generally $d^{o_d(1)}/d^{1/2}$.

In this work, we obtain a bound of the form $\lambda \le
O(2^{\log(d)^c}/d^{1/2})$ for some constant $c \in (0,1)$, which we
refer to as an almost Ramanujan bound. Rephrasing in terms of the
expansion parameter, we achieve expansion $\lambda$ with degree
$O(1/\lambda^{2+\beta})$, where $\beta = O(1/\log(1/\lambda))^{c'}$
for some $c' \in (0,1)$. We stress that the nomenclatures almost
Ramanujan, near Ramanujan and etc, may vary depending on the author.
Improving the results in this work to achieve trade-offs even closer
to the Ramanujan bound (if possible) is of great interest. We suspect
that new ideas may be required to substantially improve the bound to,
say, expansion $\lambda$ versus degree
$O(\polylog(1/\lambda)/\lambda^2)$.

%\todo{Move this before and say this is what we mean by almost optimal?}

%% file: prelim.tex
\section{Preliminaries}\label{sec:prelim}

Let $X=(V,E)$ be an $n$-vertex $d$-regular { undirected} multigraph for some $d \ge
1$. We denote by $\matr A_X$ the normalized adjacency matrix of $X$,
\ie the adjacency matrix divided by $d$.

\begin{definition}[$\expan$-spectral Expander]
  Let the eigenvalues of the {symmetric} matrix $\matr A_X$, denoted as $\spec(A_X)$, be $ \set{1=\lambda_1 \ge \cdots \ge \lambda_n}$ and  define $\lambda(X) = \max\set{\abs{\lambda_2},\abs{\lambda_n}}$.
  We say that $X$ is a $\expan$-spectral expander if $\lambda(X) \le \expan$.
\end{definition}

We denote by $G$ a finite group (except in \cref{subsec:kazhdan} where
we only need it to be finitely generated). For a multiset $S \subseteq
G$, $\Cay(G,S)$ denotes a multigraph with the
vertex set being $G$ and edges $\set{(g,sg)\;\vert\; g\in G,\; s \in S
}$. {We will assume throughout that $S = S^{-1}$ and therefore, the graph is undirected.}

{
\begin{definition}[Explicit graph]\label{def:explicit}
	A family of graphs $\{X_i\}_{i\in \N}$ is said to be \textit{explicit} if the adjacency matrix of $X_i$ can be computed deterministically in $\poly(|X_i|)$-time.  Moreover, it is said to be \textit{strongly explicit} if the list of its neighbors of any vertex in $X_i$ can be computed   $\poly(\log|X_i|)$-time.
\end{definition}
}
\paragraph{Group Representations}

In order to study the expansion of a Cayley graph, we will use the
notion of a group representation\footnote{Additional background on
  representation theory of finite groups can be found
  in~\cite{serre96}.}. \textit{Weyl's unitary trick}, says that for a
large family of groups (which includes all finite groups), every
representation can be made unitary and thus, we can restrict to
studying these.

\noindent Let $\cH$ be a complex Hilbert space and denote by
$\cL(\cH)$ the algebra of bounded linear operators\footnote{For most
  applications, one can think of $\cH = \C^n$ for some $n$, and
  $\cL(\cH) = \matr M_n(\C)$, the space of $n\times n$ complex
  matrices.  However, we will need the generality in
  \cref{subsec:kazhdan}. } on $\cH$. We denote by $\textup{U}_\cH$ the
unitary group of operators acting on $\cH$.

\begin{definition}[Group Representation]\label{def:rep}
  For a group $G$, a representation is a pair $(\rho, \cH)$
  where $\rho:G\to \textup{U}_\cH$ is a group homomorphism, \ie for
  every $g_1,g_2 \in G$, we have $\rho(g_1g_2) = \rho(g_1)\,\rho(g_2)$.
  A representation is \textit{irreducible} if the only subspaces of $\cH$ that are invariant
  under the action of $\rho(G)$ are the empty space, $\set{0}$, and the entire space, $\cH$. 
\end{definition}

\noindent Every group has two special representations,  which are,
\begin{enumerate}[topsep=-0.1em,leftmargin=*]
	\item (Trivial representation ) -- $(\rho, \C)$ where for every $g$,  $\rho(g) = 1$.
	\item ((left) Regular representation ) -- $(\rho_{\reg},  \mathcal{V}_{\reg})$ where, $ \mathcal{V}_{\reg}= \C[G]$ is
              a vector space with the elements of $G$ being an orthonormal basis,  and  $\rho_{\mathrm{reg}}(g):  h \mapsto g\cdot h$.  
\end{enumerate}

\begin{fact}\label{fact:regular}
	Let $G$ be a finite group and let $\mathcal{V}_{\reg}$ be the regular representation over $\mathbb{C}$. We have
	\begin{align*}
		\mathcal{V}_{\reg} \cong  \bigoplus_{ (\rho, V_\rho) \in \textup{Irrep}(G)} \textup{dim}(\rho) \cdot \mathcal{V}_{\rho},
	\end{align*}
	where $\textup{Irrep}(G)$ denotes the set of irreducible representations of $G$.
\end{fact}

\paragraph{Expanders and Biased Distributions}

It follows from definitions that the normalized adjacency matrix of
$\Cay(G,S)$ is given by $\matr A = \Ex{s\sim S}{\rho_{\reg}(s)}$.
Moreover, the copy of the trivial representation is the space spanned
by the all-ones vector.  \cref{fact:regular} implies that this can be
block diagonalized and therefore,
\begin{align*}
  \spec(\matr A) &= \bigcup_{\rho \in \textup{Irrep}(G)} \spec(\Ex{s\sim S}{\rho(s)} ), \;\;\;\text{ and thus,  } \\
  \expan(\Cay(G,S)) &= \max_{ \stackalign{ \rho \in \textup{Irrep}(G)\\ \rho \text{ is non-trivial } } } \opnorm{ \Ex{s\sim S}{\rho(s)} }\mper
\end{align*}
For any bounded linear operator, $\matr T \colon \cH \to \cH'$, between  Hilbert spaces, we have
\[
\opnorm{\matr T} = \sup_{v \in \cH : \norm{v}= 1} \norm{\matr T v} = \sup_{v \in \cH, w \in \cH' : \norm{v}= \norm{w}=1} \abs{\ip{\matr T v}{w}} \mcom
\]
where $\norm{v} = \sqrt{\ip{v}{v}_\cH}$ and $\norm{w} = \sqrt{\ip{w}{w}_{\cH'}}$.
%(non-\edits{zero instead of non-empty})
Given this equivalence, we will find it convenient to work with the
operator norm version referred to as \emph{bias} in the
literature~\cite{CMR13}.

\begin{definition}[Biased Set on $G$]\label{def:bias_group}
  Let $\epsilon \in (0,1)$.
  We say that a multiset $S$ of elements of a group $G$ is
  $\epsilon$-biased if for every non-trivial irreducible
  representation $\rho$, we have $\opnorm{ \Ex{s\sim S}{\rho(s)} }
  \leq \epsilon$. We sometimes use the shorthand $\bias(S) \leq
  \epsilon$, where $\bias(S) = \lambda(\Cay(G,S))$.
\end{definition}

Irreducible representations of Abelian groups, called
\emph{characters}, have dimension $1$. Thus, this definition coincides
with the usual one of $\epsilon$-biased distribution \emph{fooling}
non-trivial characters~\cite{NN90,AGHP92}. These pseudorandom
distributions were introduced in the pioneering work of Naor and Naor
where several applications to derandomization were given~\cite{NN90}.

\subsection*{Notation}

Since we deal with various vector spaces and graphs, we will find it
useful to establish some convenient notation.  While we recall these
in the relevant section, the following is a summary for ready
reference.

\begin{enumerate}[leftmargin=*,label = $\cdot$]
  \item The main multigraphs we study will be $X$ and $Y$ with vertices $V_X, V_Y$ and normalized adjacency operators $\matr A_X,  \matr A_Y$.
  
  \item We denote vertices of $X,Y$ by $x,y$ and an ordered tuple of vertices by $\vec{x}= (x_0,\cdots, x_t)$.
    
  \item We use  $u,v,w$ to denote arbitrary vectors in $ \cH$ and $x,y$ for basis vectors of $\C[V_X], \C[V_Y]$ where $\C[V_X]$ is the complex vector
        space with the elements of $V_X$ being a orthonormal basis.
        
  \item The tensored vector spaces have an induced inner product. For $\vsX\coloneqq \cH\otimes \C[V_X]$, it is $\ip{v \otimes x}{w\otimes x'} = \ip{v}{w}_\cH \ip{x}{x'}$.
        Similarly,  we have one on $\vsXY\coloneqq \vsX \otimes \C[V_Y]$.
    
  \item Orthogonal decomposition: $\vsX = \vsX^\parallel \oplus \vsX^\perp$ where $\vsX^{\parallel} \coloneqq \textup{span}\set{ v \otimes \vec{1}  \mid v \in\cH}$.
        Here, $\vec{1}$ denotes the un-normalized all-ones vector.  
        Similarly,  $\vsXY = \vsXY^\parallel \oplus \vsXY^\perp$, where $\vsXY^{\parallel} \coloneqq \textup{span}\set{ z \otimes \vec{1}  \mid z \in\vsX}$.
        
  \item The operator $\smash{\AA} $ denotes the extension of operator $\matr A$ to a tensor product of spaces where it acts as identity on the other spaces.
        For example,  $\matr A_X$ acts on $\C[V_X]$ and its extension to $\vsX$ is $\AA_X = \matr I_\cH\otimes \matr A_X$. However, if we were working on $\vsXY$,
        it would be $\AA_X = \matr I_\cH\otimes \matr A_X\otimes \matr I_Y$ instead\footnote{The spaces will be self-evident and the use of the same notation
          should not be confusing.}.
        
  \item Given an operator valued function $\matr f : V_X \to \cL(\cH) $, the generalized \emph{bias operator} is defined on the basis as\footnote{An equivalent matrix definition is
        $\matr \Pi_{\matr f} \coloneqq \sum_{x\in V_X} \matr f (x)\, \otimes \matr E_{x,x}$ where $\matr E_{x,x} \in \mathbb{C}^{V_X\times V_X}$
        is the diagonal matrix with exactly one non-zero entry of value $1$ in the row and column indexed by the vertex $x$.},
    \[\pif \colon \vsX \to \vsX , \;\ v \otimes x\mapsto \matr f (x)\, v \otimes x .\]
\end{enumerate}

%% file: expander_walks.tex
\section{Operator Bias Reduction via Expander Walks}\label{sec:simple_amp}

In this section, we establish a new \emph{operator} analog of the
expander walk--based bias amplification procedure for \emph{scalars}. An analysis of this scalar
amplification was given by Ta-Shma in~\cite{Ta-Shma17}. {We prove an operator analog that can amplify from any bias~(\cref{cor:ta-shma_simp_matrix_any_bias}) which implies the main result below~(\cref{cor:ta-shma_simp_matrix}), that amplifies from constant bias.
}

\begin{theorem}[Operator Amplification via Expander Walks]\label{cor:ta-shma_simp_matrix}
  Let $X$ be a $\lambda(X)$-spectral expander, and let $\cW_t$ be the collection of walks obtained from walks of length $t$ on $X$. Then for any operator valued function $\matr f$ such that $\opnorm{\Ex{x \in V_X}{ \matr f (x)\, } } \le \expan_0$
  and $\max_{x \in V_X} \opnorm{\matr f (x)\,} \le 1$, we have
  \begin{align*}
    \opnorm{\Ex{(x_0, \ldots x_t) \in \cW_t}{\matr f (x_t) \cdots \matr f (x_0)}}  ~\le~ \left(2\lambda(X) + \expan_0 \right)^{\lfloor t/2 \rfloor}\mper
  \end{align*}
\end{theorem}

We remark that a precursor of these techniques, in the simpler setting
of Abelian groups appears in the pioneering work of Naor and Naor
introducing $\epsilon$-biased distributions over the group
$\Z_2^m$ using expanders~\cite{NN90}.

This simpler amplification of~\cref{cor:ta-shma_simp_matrix} will be
crucially used to bootstrap the almost-optimal amplification. Moreover, it yields a construction of expanding Cayley graphs {close to any desired size}, which will be
required later.

This bias reduction procedure uses walks on an auxiliary expander
graph. Here, we only use its expansion property (as opposed to later
when we rely on its structure for the $s$-wide construction). With
this it is already possible to obtain $1/\expan^{4+o(1)}$ dependence
on the final degree of an $\expan$-expander.

\begin{restatable}[]{theorem}{expwalk}\label{theo:exp_walk}
	  Let $S \subseteq G$ such that $\lambda(\Cay(G,S))=\expan_0 < 1$.
  For every $\expan \in (0,1)$ and constant $\beta \in (0,1)$, we can find $S' \subseteq G$ in time
  $\poly(\abs{S},1/\expan_0,1/\expan)$ such that $\lambda(\Cay(G,S')) \le \expan$ and $\abs{S'} = O_{\expan_0}(\abs{S}/\expan^{4+\beta})$.
\end{restatable}

%Towards this, we first formalize the connection between bias of a
%special subset of a group and the operator norm of a certain operator.
%The subset is obtained by taking random walks over an expander graph
%as mentioned above. We then proceed to bound this operator norm.
%Finally, we instantiate our construction with an explicit expander
%graph due to~\cite{Alon21}.

\paragraph*{Notation Recall} Let $S$ be any finite set and let $X$ be a
graph on the vertex set $V_x=S$ with
$\matr A_{X}$ being its normalized adjacency matrix.
Let $\calH$ be a complex Hilbert space
and $\mathcal{L}(\calH)$ be the (bounded) operators on $\calH$; an
important example will be $\mathcal{L}(\calH) = \matr M_\ell(\C)$.
For \emph{any} operator valued function, $\matr f \colon S \to \mathcal{L}(\calH)$, we define the
generalized bias operator as
\begin{align*}
  \pif: \cH\otimes \C[V_X] \mapsto \cH\otimes \C[V_X], \;\; \pif(v \otimes x) =  \matr f (x)\, v \otimes x\mper
\end{align*}
In the scalar case, since $\cH = \C$,  earlier works  \cite{Ta-Shma17,JM21} used the implicit
identification $\C \otimes \C[V_X]\cong \C[V_X]$ and defined $\pif$ as a diagonal matrix.  This
identification is no longer is suitable when $\matr f$ is operator
valued in dimension $> 1$.  However, a simple yet crucial observation is that merely
decoupling the spaces allows us to collect the terms as we proceed
along the walk.

\subsection{Operator Norm Decay}
\begin{lemma}\label{lem:connect} Let $\cW_t \subseteq V_X^{t+1}$ be the collection of all length
$t$ walks on the graph $X$ and we define $\smash{\AA_{X} =
\matr I_\cH \otimes \matr A_{X}}$. Then, we have
\[\opnorm{\Ex{(x_0, \ldots x_t) \in \cW_t}{\matr f (x_t) \cdots \matr f (x_0)}}  ~\leq~ \opnorm{\pif\left( \AA_{X} \pif\right)^t}\leq\opnorm{\left( \AA_{X} \pif\right)^{2}}^{\lfloor t/2\rfloor} \mper
\]
\end{lemma}
\begin{proof}
	\begin{equation}\label{eqn:connect}
 \pif\left( \AA_{X} \pif\right)^t\Ex{x\in V_X}{ v \otimes x }   ~=~ \Ex{(x_t, \cdots ,x_0) \in \cW_t}{\matr f (x_t) \cdots \matr f (x_0)} v \otimes x_t  \mper
\end{equation}
This can be shown easily via induction on $t$, and we refer to
\cref{lem:operator_walk} for a formal proof of a more general
statement.
%  A minor technicality is that the operators in the image of
%$\matr f$ act on $\cH$ whereas $\pif$ acts on the space $\vsX :=
%\cH\otimes \C[V_X]$. 
We use projection and lifting maps to move
between the spaces $\vsX$ and $\cH$. Define $\projh : \vsX \to \cH $
and $\lifth : \cH \to \vsX$, as,
\begin{align*}
  \projh (w \otimes x) = w,\; \; \lifth (v) = \Ex{x\in V_X}{ v \otimes x } = \frac{1}{\abs{V_X}}\, v\otimes \vec{1}  \mper
\end{align*}

From the definition, $\norm{\lifth (v)} = \norm{v}\frac{\norm{\vec{1}}}{\abs{V_X}} = \frac{\norm{v}}{\sqrt{\abs{V_X}}}$ and thus, $\smash{\opnorm{\lifth} =1/\sqrt{\abs{V_X}}}$. We can use Cauchy-Schwarz to get that
$\smash{\opnorm{\projh} = \sqrt{\abs{V_X}}}$. Now, we put this together to obtain a simple expression on the quantity we need to bound
\begin{align*}
    \opnorm{\Ex{(x_0, \ldots x_t) \in \cW_t}{\matr f (x_t) \cdots \matr f (x_0)}}       &~=~ \sup_{\norm{v}= 1} \norm{\Ex{(x_0, \cdots , x_t) \in \cW_t}{\matr f (x_t) \cdots \matr f (x_0)} v}_2\\
           &~=~ \sup_{\norm{v}= 1} \norm{\projh\,\left(\Ex{(x_0, \cdots ,s_t) \in \cW_t}{\matr f (x_t) \cdots \matr f (x_0)} v \otimes x_t \right)}_2\\
       &~=~ \sup_{\norm{v}= 1}  \norm{\projh \pif\left( \AA_{X} \pif\right)^t  \Ex{x\in V_X}{ v \otimes x}}_2\\
&~=~ \sup_{\norm{v}= 1}  \norm{\projh \pif\left( \AA_{X} \pif\right)^t  \lifth \, v  }_2\\
&~\le~ \opnorm{\pif\left( \AA_{X} \pif\right)^t}\opnorm{\projh}\opnorm{\lifth} \\
\label{eqn:product_pi}&~\le~ \opnorm{\pif\left( \AA_{X} \pif\right)^t}\mper
\end{align*}
The last inequality follows from submultiplicativity of the operator norm and the observation that $\opnorm{\pif} = \norm{f}_\infty \leq 1$.
\end{proof}

Now that we have reduced the problem to studying the operator norm, we
will study how the norm decays as we take walks.  We use the
decomposition, $\vsX = \vsX^\parallel \oplus \vsX^\perp$ where
$\vsX^{\parallel} \coloneqq \textup{span}\set{ v \otimes \vec{1} \mid
  v \in\cH}$.  The decay comes from two sources.  For $z \in
\vsX^\perp$, we get a decay by $\expan(X)$ by the definition of $X$
being an expander.  \cref{claim:parallel_vector_bias} shows that for
$z \in \vsX^\parallel$, we get a decay from $\pif$, equal to the
initial bias. We put this together in \cref{cor:ta-shma_simp_matrix}
to obtain the desired exponential decay.
\begin{claim}\label{claim:parallel_vector_bias}
  For $z \in \vsX^{\parallel}$, we have
  \begin{align*}
    \norm{\left(\pif\, z\right)^{\parallel}}_2 ~\le~ \opnorm{\Ex{x \in V_X}{ \matr f (x)\, }} \cdot \norm{z}_2\mper
  \end{align*}
\end{claim}

\begin{proof}
%  The equation trivially holds when $z = 0$, so assume $z \ne 0$ and scale it so that $\norm{z}_2=1$.
  From definition of $\vsX^{\parallel}$, we can assume that $ z = u \otimes \vec{1}$.  Computing we have,  
  \begin{align*}
    \norm{\left(\pif\left(u \otimes \vec{1} \right)\right)^{\parallel}}_2 ~&=~ \sup_{w \in \cH \colon \norm{ w\otimes \vec{1} }_2=1} \abs{\ip{ w  \otimes \vec{1}}{\pif\left( u \otimes \vec{1} \right)}}\\
    ~&=~ \sup_{w \in \cH \colon \norm{ w\otimes\vec{1}}_2=1} \abs{\ip{ w  \otimes \vec{1}}{\pif\left( u \otimes \sum_{x\in V_X} x \right)}}\\
    ~&=~ \sup_{w \in \cH \colon \norm{ w\otimes \vec{1}}_2=1} \abs{\ip{ w  \otimes \vec{1}}{\sum_{x\in V_X} \left(\matr f (x)\, u \otimes x\right)}}\\
    ~&=~ \sup_{ w  \in \cH \colon \norm{w \otimes\vec{1}}_2=1} \abs{\sum_{x \in V_X} \ip{ w }{   \matr f (x)\,  u} \ip{\vec{1}}{x} } \\
    ~&=~ \sup_{ w  \in \cH \colon \norm{w\otimes \vec{1}}_2=1} \abs{ \ip{ w }{ \abs{V_X}\left(\Ex{x \in V_X}{ \matr f (x)\, }\right)u}  }   \\
    &~\le~  \opnorm{\Ex{x \in V_X}{ \matr f (x)\, } } \abs{V_X} \norm{w}\norm{u}  = \opnorm{\Ex{x \in V_X}{ \matr f (x)\, } }\norm{z}_2 \mper
\qedhere
   \end{align*}
The last line follows as $\norm{z}_2 = \norm{u}_2 \sqrt{\abs{V_x}}$ and $\norm{w} = \frac{1}{\sqrt{\abs{V_x}}} $.
\end{proof}

We show that for every two\footnote{This is the source of loss of a
  factor of $2$ in the exponent (which leads to degree
  $O(\abs{S}/\lambda^{4+o(1)})$ rather than the desired degree of
  $O(\abs{S}/\lambda^{2+o(1)})$ we will later achieve. Note that the
  same loss occurs in the original zig-zag analysis of~\cite{RVW00},
  which was later remedied by the $s$-wise zig-zag of~\cite{BT08}.}
steps of the walk, the norm of the (associated) operator decays as
follows.

\begin{lemma}\label{lemma:two_step_analysis}
  Let $X$ be a $\lambda(X)$-spectral expander and let $\matr f$ be such that $\opnorm{\Ex{x \in V_X}{ \matr f (x)\, }} \le \expan_0$ and
  $\max_{x \in V_X} \opnorm{\matr f (x)\,} \le 1$. Then,
  \begin{align*}
    \opnorm{\left( \AA_{X} \pif  \right)^2} ~\le~ 1 - (1-\lambda(X))^2 (1-\expan_0)\mper
  \end{align*}
\end{lemma}

\begin{proof}
  Let $\matr A_J = \matr J/\abs{V(X)}$, where $\matr J$ is the $\abs{V(X)} \times \abs{V(X)}$ all ones matrix.
  We can write $\matr A_X = (1-\lambda)  \matr A_J + \lambda \matr E$, where $\lambda=\lambda(X)$ and
  $\opnorm{\matr E} \le 1$. Then
  \begin{align*}
    \opnorm{\AA_X \pif  \AA_X} ~\le~ & (1-\lambda)^2 \opnorm{\AA_J \pif \AA_J} ~+~ \lambda (1-\lambda) \opnorm{\EE \pif \AA_J} \\
    & +~ (1-\lambda)\lambda \opnorm{\AA_J \pif \EE} ~+~ \lambda^2 \opnorm{\EE \pif \EE}\mper
  \end{align*}
To analyze $\opnorm{{\AA_J \pif \AA_J}}$ , let $z \in \vsX$ be a unit vector which is decomposed as $z = z^{\parallel} + z^\perp$. We have,
   \begin{align*}
    \norm{ \left(\AA_{J}  \pif \AA_{J}  \right) \left(z^\perp + z^{\parallel}\right)}_2  &~=~   \norm{\left( \AA_{X}  \pif \AA_{X} \right) z^{\parallel}}_2 &&(\text{As } \lambda(\matr A_J) = 0)\\
    &~=~ \norm{\AA_{X} \left(\left(\pif z^{\parallel}\right)^\perp +  \left(\pif z^{\parallel}\right)^\parallel \right)}_2\\
    &~=~  \norm{\left(\pif z^{\parallel}\right)^\parallel}_2\\    
    &~\le~ \expan_0.  && \text{(By~\cref{claim:parallel_vector_bias})}
   \end{align*}
 Thus,  $   \opnorm{\AA_J \pif \AA_J} \le \expan_0 $. Recall that $\opnorm{\pif} \le 1$ since $\max_x \opnorm{\matr f (x)\,} \le 1$, and we also have $\opnorm{\matr E}, \opnorm{\matr A_J} \le 1$.
  Then,
  \begin{align*}
    \opnorm{\AA_X \pif  \AA_X} ~\le~ & (1-\lambda)^2 \expan_0 ~+~ 2\lambda (1-\lambda) ~+~ \lambda^2\\
                                ~=~ & (1-\lambda)^2 \expan_0 ~+~ 1  - (1-\lambda)^2,\\
                                ~=~ & 1 - (1-\lambda)^2 (1-\expan_0)\mcom
  \end{align*}
  concluding the proof.
\end{proof}

%\begin{lemma}\label{lemma:two_step_analysis}\todo{2 suggestions to remove this lemma}
%  Let $X$ be a $\lambda(X)$-spectral expander and let $\matr f$ be such that $\opnorm{\Ex{x \in V_X}{ \matr f (x)\, } } \le \expan_0$ and $\max_{x \in V_X} \opnorm{\matr f (x)\,} \le 1$.
%  Then,
%  \begin{align*}
%    \opnorm{\left( \AA_{X} \pif  \right)^2} ~\le~ 2\lambda(X) + \expan_0\mper
%  \end{align*}
%\end{lemma}
%
%\begin{proof}
%   Since $\opnorm{\pif} = \max_{x \in V_X} \opnorm{\matr f (x)\,} \le 1$, it is enough to bound $\smash{\opnorm{\AA_{X} \pif \AA_{X}  }}$.
%   Let $z \in \vsX$ be a unit vector which is decomposed as $z = z^{\parallel} + z^\perp$. We have
%   \begin{align*}
%    \norm{ \left(\AA_{X}  \pif \AA_{X}  \right) \left(z^\perp + z^{\parallel}\right)}_2  &~\le~ \lambda(X) +  \norm{\left( \AA_{X}  \pif \AA_{X} \right) z^{\parallel}}_2\\
%    &~\le~ \lambda(X) +  \norm{\AA_{X} \left(\left(\pif z^{\parallel}\right)^\perp +  \left(\pif z^{\parallel}\right)^\parallel \right)}_2\\
%    &~\le~ \lambda(X) + \norm{\AA_{X}\left(\pif z^{\parallel}\right)^\perp }_2 +  \norm{\left(\pif z^{\parallel}\right)^\parallel}_2\\    
%    &~\le~ 2 \lambda(X) +  \norm{\left(\pif z^{\parallel}\right)^\parallel}_2 \\
%    &~\le~ 2\lambda(X) + \expan_0.  && \text{(By~\cref{claim:parallel_vector_bias})}
%   \end{align*} \qedhere
%\end{proof}

{
The amplification guarantee of~\cref{cor:ta-shma_simp_matrix}
trivializes if $2\lambda(X) + \expan_0 \ge 1$. Nonetheless, we now
show that amplification does occur under much weaker conditions,
namely, whenever $\opnorm{\Ex{x \in V_X}{ \matr f (x)\, }} < 1$ and
the auxiliary graph $X$ has expansion $\lambda(X) < 1$. This regime of bias
amplification was instrumental in the breakthrough $\mathsf{SL}=\mathsf{L}$ result of Reingold~\cite{R05}.}

\begin{theorem}[Operator Amplification via Expander Walks (strengthening of~\cref{cor:ta-shma_simp_matrix})]\label{cor:ta-shma_simp_matrix_any_bias}
  Let $X$ be a $\lambda(X)$-spectral expander and let $\cW_t$ be the collection of walks obtained from walks of length $t$ on $X$.
  Then for any operator valued function $\matr f$ such that $\opnorm{\Ex{x \in V_X}{ \matr f (x)\, } } \le \expan_0$
  and $\max_{x \in V_X} \opnorm{\matr f (x)\,} \le 1$, we have
  \begin{align*}
    \opnorm{\Ex{(x_0, \ldots x_t) \in \cW_t}{\matr f (x_t) \cdots \matr f (x_0)}}  ~\le~ \left[1-(1-\lambda(X))^2 (1-\expan_0) \right]^{\lfloor t/2 \rfloor}\mper
  \end{align*}
\end{theorem}
\begin{proof}
	Follows by combining~\cref{lem:connect} and~\cref{lemma:two_step_analysis}.
\end{proof}
%\cref{cor:ta-shma_simp_matrix} now follows from the lemma above and
%the submultiplicativity of the operator norm.

%
%\todo{Need to move these paragraphs appropriately}
%\snote{I don't understand the point of Theorem 3.1; we should remove the entire discussion around and just do the proof of Thm 3.6}

{ This establishes that expander walks can be used to derandomize powers
of an operator, itself given by an average of bounded operators, in
the general case.
In this derandomization, we still have an
exponential norm decay, but we only ``pay additional randomness''
proportional to the degree of the auxiliary expander regardless of the
number of operators.}

%\edits{The above amplification follows from the following improved version
%of~\cref{lemma:two_step_analysis}. The proof explores the structural
%\emph{syntactic} similarity between the operator amplification and
%known zig-zag analysis~\cite{RVW02,R05,TD18}. 
%}

\subsection{Cayley Graphs and the Construction of Amplified Biased Sets}

In this section, we will prove the following,

\expwalk*

Towards this, we first formalize the connection between bias of a
special subset of a group and the operator norm of a certain operator.
The subset is obtained by taking random walks over an expander graph
as mentioned above. We then proceed to bound this operator norm.
Finally, we instantiate our construction with an explicit expander
graph due to~\cite{Alon21}.

The particular case of $S \subseteq G$ (for some group $G$) and the
function $\matr f$ being a representation $\rho$ on $\cH$
leads to the amplification of biased sets. We will construct a new
multiset $S' \subseteq G$ such if $\opnorm{\Ex{s\sim S}{\rho(s)}} \leq
\expan_0$, then we have $\opnorm{\Ex{s\sim S'}{ \rho(s)}} \leq \expan
\ll \expan_0$. Note here that the construction of $S'$ is agnostic to
$\rho$, and thus we can reduce the bias of all irreducible
representations simultaneously!
Assume that we have a graph $X$ on the vertex set $S$. For $s \in S$,
we have $\matr f (s) = \rho(s)$ in this case. Let
\begin{align*}
  S' ~=~ \{  s_t s_{t-1} \cdots s_0 \; |\; (s_0, s_1, \cdots s_t) \in \cW_t  \} \mcom
\end{align*}
which will be our new amplified biased set. Using the homomorphism
property of $\rho$, we have the following simplification
\begin{align}
 \Ex{w=(s_0, \ldots s_t) \in \cW_t}{\matr f(s_t) \cdots \matr f (s_0)} ~&=~ \Ex{(s_0, \ldots, s_t) \in \cW_t}{\rho(s_t) \cdots \rho(s_0)} ~=~ \Ex{s' \in S'}{\rho(s')}\mcom\\
\label{eqn:homo} \text{and thus, }\; \mathrm{bias}(S')  
~&\leq~ \opnorm{\pif\left( \AA_{X} \pif\right)^t}
\end{align}
where $S'$ is the new biased multiset of the construction and the
second inequality follows from the preceding calculation when $\cW_t$
is a collection of walks on $X$.

\subsubsection{Instantiating the Construction}

To construct $S'$, our construction requires an auxiliary expander
graph $X$ to perform walks on. One convenient source (among several)
is a recent construction of Alon.

\begin{theorem}[Corollary of {\cite[Thm.  1.3]{Alon21}} ]\label{thm:alon_exp}
  For every $\expan \in (0,1)$, there exists a positive integer  $m_\lambda$ {such that for every $n \in \N$,
  there is an} explicit construction of a graph $X$ on $m_\expan n$ vertices
  with degree at most $9/\expan^2$ and $\lambda(X) \leq \expan$.
\end{theorem}

\noindent We now establish the key amplification lemma.

\begin{lemma}\label{lem:simple_amplification}
  Let $S \subseteq G$ such that $\bias(S)=\expan_0 < 1$ and let  $\epsilon_0$ be a constant such that $(1+2\epsilon_0)\lambda_0 < 1$.   Then, for any $\lambda  > 0$,  we can explicitly compute $S'$
  such that $\bias(S') \leq\lambda$ and $|S'| = O_{\epsilon_0\lambda_0} \left( \frac{|S|}{\lambda^{4+4\delta(\lambda_0, \epsilon_0)}} \right)$ {where} $\delta(\lambda_0, \epsilon_0) = \frac{ \log (3 +6\epsilon_0) + \log (1/\epsilon_0) }{\log (1/\lambda_0)}$.
\end{lemma}

\begin{proof}
  Pick a constant $\epsilon_0$ such that $\lambda_1 \coloneqq
  (1+2\epsilon_0)\lambda_0 < 1$ and use~\cref{thm:alon_exp} to
  obtain an explicit $(m|S|,d,\epsilon_0 \expan_0)$-graph $X$ where $d \leq \frac{9}{\epsilon_0^2\lambda_0^2} $. Let
  $S_1$ be the multiset consisting of $m$ copies of $S$. The bias
  remains the same and now, $|V(X)| = |S_1|$.
  We construct $S'$ by multiplying elements of $t$-length walks on $X$ where
  $t = \lceil 2(1+\log_{\lambda_1} (\lambda)) \rceil$.
  The size of $S'$ is
  \begin{align*}
   |S'| = (m|S|)\cdot d^t ~&=~ m\cdot |S| \cdot \left (  \frac{3}{\epsilon_0\lambda_0} \right )^{4+4\log_{\lambda_1 } \lambda }\\
           ~&=~ m\left (  \frac{3}{\epsilon_0\lambda_0} \right )^4 |S| \cdot \lambda^{ \frac{ -4\log\left ( \frac{3}{\epsilon_0\lambda_0} \right )  }{\log (1/\lambda_1) }  }\\
 ~&\leq~ O_{\epsilon_0\lambda_0}(|S|) \cdot \lambda^{-4 \left(1 + \frac{ \log\left ( \frac{3 +6\epsilon_0}{\epsilon_0} \right ) }{\log (1/\lambda_0)}  \right )}\mper
  \end{align*}
%\edits{So here, $\delta(\lambda_0, \epsilon_0) = \frac{ \log\left ( \frac{3 +6\epsilon_0}{\epsilon_0} \right ) }{\log (1/\lambda_0)}$.} \todo{The delta definition}

  Let $\rho$ be any irreducible representation of $G$. From ~\cref{eqn:homo} and \cref{cor:ta-shma_simp_matrix},  we get,
  \begin{align*}
     \opnorm{\Ex{s_0\cdots s_t \in S'}{\rho(s_t \cdots s_0)}} \leq \left ( 2\lambda(X) + \bias(S)\right )^{ t/2 -1 } \leq \left ( \expan_1 \right )^{t/2-1} \leq  \lambda \mper \qedhere
  \end{align*}
\end{proof}

Using the amplification above, we now derive our first simplified
explicit construction.
\expwalk*
\begin{proof}  Pick a \emph{constant} $\expan' < \min \left (\frac{1}{2}, \left
  (\frac{3}{4} \right )^{4\beta} \right ) $. This choice ensures that $\delta(\expan', 1/2 ) \leq \beta $.  
    Use \cref{lem:simple_amplification} with target expansion $
  \expan'$ to obtain a set $S_1$ with size $|S_1| =
  O_{\lambda_0,\beta}(|S|)$ as $\expan'$ is a constant.
  Now use
  \cref{lem:simple_amplification} again with $S_1$ as the initial set, {$\epsilon_0 = \frac{1}{2}$ }and the final expansion as $\lambda$ to obtain $S'$.
  Thus, the final size is $|S'| \leq O_{\lambda'} \left(
  \frac{|S_1|}{\lambda^{4+\delta(\lambda',1/2)}} \right) \leq
  O_{\lambda_0, \beta} \left( \frac{|S|}{\lambda^{4+\beta }} \right)$.
\end{proof}

\subsection{Explicit Expanders close to any desired size}

As an application of \cref{theo:exp_walk}, we demonstrate an
explicit construction of Cayley expanders of size $n +o(n)$ vertices for every (large enough) $n$.  

Such a construction will be crucial for us to prove~\cref{theo:main_1}. We cannot use existing results like the recent work of Alon \cite{Alon21} or the construction in~\cite{Ta-Shma17}. This is because Alon's construction does not have a Cayley graph structure (which our proof utilizes).  On the other hand, the construction in~\cite{Ta-Shma17} is a Cayley graph based on~\cite{LPS88}, but it only guarantees a graph of size $O(n)$ rather than $n + o(n)$. 
%is very large\todo{Need to explain this largeness}.

Recall that $\slin_2(p)$ is the group of $2 \times 2$ invertible matrices over $\F_p$ with determinant $1$.  We obtain a base generating set for $\slin_2(p)$ via the following result.

\begin{theorem}[{\cite{Lub11}}]\label{thm:sl2}
  {There exists an absolute
  constant $\lambda_0 < 1$ such that for every $p > 17$, there exists an explicit generating set} $S$ ({of constant size independent of $p$}) for
  $\slin_2(p)$, such that $\lambda\left
  (Cay(\slin_2(p), S) \right) \leq \lambda_0$ .
\end{theorem}

\begin{theorem}[{\cite{Cheng10}}]\label{lem:primes}
  For every $n \geq 2^{3\cdot 2^{15}}$, there exists a prime in $[n, n+4n^{2/3}]$.
\end{theorem}

\begin{corollary}\label{cor:small_expander}
  For any $n > 2^{9\cdot 2^{15}},  \lambda > 0$, there is a deterministic polynomial time algorithm to construct an
  $(n',d,\lambda)$-graph  $Cay(\slin_2(p),S)$, where $n' = n + O(n^{8/9})$ and $d = O(\lambda^{-4.1})$. 
\end{corollary}

\begin{proof}
  Find a prime $p \in [n^{1/3}+1 , n^{1/3}+ O(n^{2/9})]$, which exists due to \cref{lem:primes}, via brute-force search.
  Since, $\slin_2(p)$ is a group of order $(p^2-1)p$, we have $n \leq |\slin_2(p)| \leq n + O(n^{8/9}) $.
  We use the constant-sized generating set $S$ from \cref{thm:sl2} and amplify using \cref{theo:exp_walk}.
\end{proof}

%% file: adv_amp.tex
\section{Operator Bias Reduction via the $s$-wide Replacement Walk}\label{sec:adv_ampl}

We have seen in \cref{sec:simple_amp} that bias reduction via random
walks on an expander $X$ is sub-optimal (by a factor of $2$ in the
exponent).  We will derandomize this random walk construction to
achieve an almost optimal bias reduction. The idea is to introduce a
new graph $Y$ which has a much smaller degree, and to ``simulate" a
random walk on $X$ via a walk on $Y$.  This is realized by a
higher-order version of the zig-zag product~\cite{RVW00} called the
\emph{$s$-wide replacement product} defined by Ben-Aroya and
Ta-Shma~\cite{BT08} (see~\cref{def:s_wide_replacement}).

This section establishes our key technical result which states that
given any initial \emph{operator valued function} of constant bias $<
1$, we amplify the bias in an almost optimal way. This generalizes the
analysis of Ta-Shma~\cite{Ta-Shma17} from \emph{scalar} valued
functions to \emph{operator} valued functions.

\begin{restatable}[Operator Generalization of Theorem 24~\cite{Ta-Shma17}]{theorem}{OpAmp}\label{theo:ta-shma_main}
  Fix integers $t \ge s \ge 1$.
  Let $X$ be any $d_1$-regular graph {(with $d_1$ a power of $2$)}, and $Y$ be any $d_2$-regular Cayley graph on $\mathbb{F}_2^{s\log d_1}$.
  Let $\sW_t$ be the collection of length $t$ walks on the $s$-wide replacement product of $X$ and $Y$.
  Let $\calH$ be a Hilbert space.
  For any operator valued function $\matr f \colon V_X \to \mathcal{L}(\cH)$, satisfying $\max_{x \in V_X} \opnorm{\matr f (x)\,} \le 1$
  and $ \opnorm{\Ex{x \in V_X}{ \matr f (x)\, } } := \lambda_0 \le~ \lambda(Y)^2 - 2 \lambda(X)$, we have
  \begin{align*}
    \opnorm{ \Ex{(x_0,\cdots, x_{t}) \in \sW_t}{\matr \matr f(x_t)\cdots \matr f(x_0)} } \le \left(\lambda(Y)^s + s \cdot \lambda(Y)^{s-1} + s^2 \cdot \lambda(Y)^{s-3} \right)^{ \lfloor t/s \rfloor } \le O_s\left(\lambda(Y) \right)^{(1-o_s(1))t}\mper
  \end{align*}
  Furthermore, the size of the collection is $|\sW_t| = |X| \cdot d_1^s \cdot d_2^{t}$.
\end{restatable}

\begin{remark}
  Note that there is an inherent trade-off between the spectral bound amplification (on the operator norm), and the degree bound (on the number of walks),
  which causes the suboptimality in how close this technique lets us approach the Ramanujan bound.
  As in~\cite{Ta-Shma17}, the $o_\lambda(1)$ term we obtain from the bound above is $(1/\log(1/\lambda))^c$ for some $c >0$
  (see~\cref{lemma:almost_ram_exp_1} for the precise computation).
\end{remark}

\paragraph{Section Outline} In
\cref{sec:s_wide_replacement_prod}, we recall the $s$-wide replacement
product and describe random walks on it. Then,
in~\cref{sec:simulation}, we formalize the distributions we work with
and reprove the result that if $Y$ is a Cayley graph over any product group of appropriate size, then it enables the transfer of pseudorandomness from $Y$ to $X$. The key generalization to operator valued functions
is established in \cref{lem:operator_walk} which is identical in
spirit to \cref{eqn:connect}. In~\cref{subsec:norm_decay}, we then
finish the amplification analysis in a manner similar
to~\cite{Ta-Shma17}.  In \cref{sec:parameters}, we provide details
about instantiating the setup by explicitly constructing the graphs we
need.
%\footnote{Any product group of the form $G^s$
%  with $\abs{G}=d_1$ can be used in the $s$-wide construction and it
%  satisfies this \emph{compatible} property. Note that
%  in~\cref{theo:ta-shma_main} we used $G=\F_2^{\log_2(d_1)}$.} 

\subsection{The $s$-wide Replacement Product and its Walks}\label{sec:s_wide_replacement_prod}

Let $X$ be a $d_1$-regular graph. For each $x \in V_X$ and $j \in [d_1]$, let $x[j]$ be the $j$-th neighbor of $x$.

\begin{definition}[Locally Invertible Rotation Map]
  $X$ admits a locally invertible rotation map if there exists a bijection $\phi \colon [d_1] \to [d_1]$
  such that for every $(x,j) \in V_X \times [d_1]$,
  \begin{align*}
    \text{if }\; x' = x[j],\;  \text{ then, }\; x'[\phi(j)] = x\mper
  \end{align*}
\end{definition}

\begin{example}[Cayley Graphs are Locally Invertible]\label{exam:cayleylocal}
  Let $G$ be a group and $A \subseteq G$ where the set $A$ is closed under inversion.
  Label the neighbors of vertices in  $\textup{Cay}(G,A)$, by elements of $A$ such that $g[a] = a \cdot g $.
  Then,  $\textup{Cay}(G,A)$ is locally invertible as the map $\phi \colon A \to A$ defined as
  $\phi(a) = a^{-1}$ clearly satisfies the criteria,   
  \begin{align*}
    \text{if }\; g' = g[a] = a \cdot g,\;  \text{ then, }\; g'[\phi(a)] = a^{-1} \cdot  g'= g\mcom
  \end{align*}
  for every $g \in G$, $a \in A$.
\end{example}

We now define the $s$-wide replacement product which is a generalization of standard
replacement product of graphs which can be seen as the special case when $s=1$. 

\begin{definition}[$s$-wide Replacement Product]\label{def:s_wide_replacement}
  Suppose we are given the following:
  \begin{itemize}
    \item A $d_1$-regular graph $X$ with a bijection
          $\phi:  [d_1] \to [d_1]$ which defines a locally invertible rotation map.
    \item A $d_2$-regular graph $Y$ on the vertex set $[d_1]^s$.
  \end{itemize}
  And we define:
  \begin{itemize}
    \item For $i \in \set{0,1,\dots,s-1}$, we define $\textup{Rot}_i \colon V_X \times V_Y \to V_X \times V_Y$ such that,
          \begin{align*}
            \textup{Rot}_i((x, (a_0,\dots,a_{s-1}))) \coloneqq (x[a_i], (a_0,\dots,a_{i-1}, \phi(a_i),a_{i+1},\dots,a_{s-1}))\mcom
          \end{align*}
          for every $x \in V_X $ and $(a_0,\dots,a_{s-1}) \in  V_Y = [d_1]^s$.         
    (Note that the $Y$ component of the rotation map depends only on a vertex's $Y$ component, not its $X$ component.)
          
    \item Denote by $\matr X_i$, the operator on $\C[V_X \times V_Y]$ which acts on the natural basis via the permutation
          $\textup{Rot}_i$, and let $\matr A_Y$ be the normalized random walk operator of $Y$.
  \end{itemize}

  Then $t$ steps of the $s$-wide replacement product are given by the operator
  \begin{align*}
    \matr X_{t-1 \bmod s} \AA_Y\, \cdots \, \matr X_{1 \bmod s} \AA_Y \matr X_{0 \bmod s} \AA_Y\mper
  \end{align*}
\end{definition}

Observe that a walk on the $s$-wide replacement product yields a walk
on the outer graph $X$ by recording the $X$ component after each step
of the walk.  Since a walk is completely determined by its intra-cloud
steps, the number of $t$-step walks on the $s$-wide replacement
product is,
\begin{align*}
  \abs{V_X} \cdot \abs{V_Y} \cdot d_2^{t} = n \cdot d_1^s \cdot d_2^{t} \ll n \cdot d_1^{t}\mcom
\end{align*}
which therefore gives us a very sparse subset of all $t$-walks on $X$.
Thus the $s$-wide replacement product will be used to simulate random
walks on $X$ while requiring a reduced amount of randomness (as we
shall see this simulation is only possible under special conditions,
namely, when we are uniformly distributed on each cloud).

\subsection{The Collection of Derandomized Walks}\label{sec:simulation}

We now describe the distribution obtained by the walks on the $s$-wide
replacement product using the language of operators.

%Recall that, in the expander walk case discussed
%in~\cref{sec:simple_amp}, we first relate the set of walks $W_t$ to
%the action of the t-step walk operator (\cref{eqn:connect}) and then
%obtain that the task of bounding the bias reduces to bounding the
%operator norm (\cref{eqn:product_pi}). Similarly for $s$-wide case,
%we express a $t$-step walk in terms of a $s$-wide operator that act on
%the extended space $\C[V_X] \otimes \C[V_Y]$. Then we prove a core
%lemma that intuitively says: the action of $t$-step $s$-wide operator
%is same as the action of the $t$-step random walk operator in an
%appropriate sense, whenever $t \leq s$. The scalar version of this
%lemma is present (Lemma 26) in~\cite{Ta-Shma17} and we generalize it
%for operator valued functions. This generalization requires some care
%and appropriate notational setup. Finally, we use this lemma to show
%bias decay for any number of steps ($t$).

\begin{definition}[Operators and Distributions]\label{def:op_dist}
  Given a tuple of random walk operators\footnote{Markov chain operators on $V_X \times V_Y$.}
  $\matr B = (\matr B_0, \cdots, \matr B_{t-1} )$ on $\C[V_X]\otimes\, \C[V_Y]$ and a starting
  vertex $x_0 \in V_X$, we can define a distribution induced by the walk using
  these operators.
  More precisely, $\cD(\matr B, x_0)$ is the distribution on
  $(V_X \times V_Y)^{t+1}$ such that for every $1 \leq \ell \leq
  t$,
  \begin{equation}\label{eqn:distrib}
    \left (  \matr B_{\ell -1} \cdots \matr B_0 \right ) \left (x_0\otimes \vecn{1}\right ) = \E_{ (\vec{x}, \vec{y} ) \sim \cD(\matr B,x_0) } x_\ell \otimes y_\ell .
   \end{equation}
   We typically suppress $x_0$ as it will not matter and denote $\cD_{X}(\matr B)$, and $\cD_{Y}(\matr B) $ to specify the projections to $V_X$, and  $V_Y$ respectively.
\end{definition}

The next lemma is a generalization of~\cref{eqn:connect} which we need
for the $s$-wide replacement walk. This can also be specialized to
prove~\cref{eqn:connect} by letting $Y$ be a graph with one vertex
(and thus $\vsX \cong \vsXY$). Recall that $\pifw \,(v\otimes x\otimes y)
= \matr f (x)\, v\otimes x\otimes y$.

\begin{lemma}[Operator Generalization]\label{lem:operator_walk}
  For any tuple of random walk operators $\matr B$, any operator valued $\matr f$, and any $v \in \cH$,  $x_0 \in V_X$, we have
  \begin{align*}
    \left ( \BB_{t-1} \pifw  \cdots \BB_0 \pifw \right ) \left (v \otimes x_0\otimes \vecn{1} \right )  = \Ex{ (\vec{x}, \vec{y} ) \sim \cD(\matr B) }{  \matr f (x_{t-1}) \cdots \matr f (x_0)\,  v  \otimes x_t \otimes y_t}\mper
  \end{align*}
\end{lemma}

%\edits{Delete? \textit{Since we now work with the tensor
%products of three spaces (one for the graph $X$, one for the graph
%$Y$, and one for the operator valued function $f$), we formalize the
%computation more explicitly}}.

\begin{proof}
  We prove the computation via induction on $t$. The base case is when $t =1$
  \begin{align*}
    \left ( \BB_{0}\pifw \right ) \left (v \otimes x_0\otimes \vecn{1} \right )  ~&=~  \BB_{0} \left ( \matr f (x_0)\, v \otimes x_0\otimes \vecn{1} \right )\\
        ~&=~  \Ex{(\vec{x}, \vec{y} ) \sim \cD(\matr B)}{\matr f (x_0) \,v \otimes x_1\otimes y_1} && (\text{Using  \cref{eqn:distrib} for } \ell = 1)
  \end{align*}
 
  Let $y_0 =  \vecn{1}$ and assume the statement holds for $t-1$. Then,
  \begin{align*}
   \left ( \BB_{t-1} \pifw  \cdots \BB_0 \pifw \right )\left (v \otimes x_0\otimes  y_0 \right )  ~&=~
           \BB_{t-1} \pifw\cdot \prod_ {i= t-2}^{0} \left (\BB_{i}\pifw \right )\left (v \otimes x_0\otimes y_0 \right )  \\
         ~&=~\BB_{t-1}\pifw \Ex{ (\vec{x}, \vec{y} ) \sim \cD(\matr B) }{ \matr f (x_{t-2}) \cdots \matr f (x_0)\,v  \otimes x_{t-1} \otimes y_{t-1}}  \\
         ~&=~\BB_{t-1} \Ex{ (\vec{x}, \vec{y} ) \sim \cD(\matr B) }{ \matr f (x_{t-1}) \matr f (x_{t-2}) \cdots \matr f (x_0)\, v  \otimes x_{t-1} \otimes y_{t-1}}\\
         ~&=~ \Ex{(\vec{x}, \vec{y} ) \sim \cD(\matr B)}{\matr f(x_{t-1}) \cdots \matr f (x_0)\,v  \otimes x_t \otimes y_t}\mper
  \end{align*}
  The second equality uses the inductive hypothesis, and the last two
  equalities use \cref{eqn:distrib} for $\ell = t-1$ and $\ell = t$ respectively.
\end{proof}

Using~\cref{def:op_dist}, we further define the operators for the distributions we wish to study.
\paragraph*{Uniform Distribution} Let us first capture, using this notation, the uniform distribution on walks on $X$ starting from $x_0 \in V_x$.
We define $\matr B_U$ where for each $i$,  $\matr B_i = \matr A_X\otimes I_Y$ for every $i$.
Then,  for any $\ell,\; (\matr A_X\otimes  I_Y)^\ell = \matr A_X^\ell\otimes  I_Y$. Therefore, we obtain that $\cD_X(\matr B_U)$
is the $t$-step random walk distribution on $X$ \ie $x_i \sim A_X^ix_0$.  
\paragraph*{The $s$-wide Distribution} This is the distribution obtained by the $s$-wide walks as described in the earlier section.
For $0\leq a\leq b < s$, we define
\begin{align*}
  \matr B [a,b] = \left(\matr X_{a } \AA_Y,\matr X_{a+1 } \AA_Y, \cdots,  X_{b} \AA_Y \right) \mper
\end{align*}
We can view this random walk as occurring in two steps. The first
being picking an initial vertex $y_0 \in Y$ and then, picking the
sequence of neighbors according to which we will perform the walk in
$Y$.
To formalize this, let $\matr A_Y = (1/d_2) \sum_{j = 1}^{d_2}
\matr P_j$ where $\matr P_j$ are permutation matrices and let $\permc
= (j_0, \cdots, j_{b-a}) \in [d_2]^{b-a+1}$. For a fixed sequence of permutations $J$, the conditional
distribution (conditioned on $J$) is defined by,
\begin{align*}
  \matr B [a,b,\permc] = \left(\matr X_{a} \PP_{j_0}, \,\matr X_{a +1} \PP_{j_1}, \cdots, \matr X_{b } \PP_{j_{b-a}} \right)\mper
\end{align*}
We would like these two distributions,{ \ie the uniform distribution and the $s$-wide distribution} , to be the same and a graph $Y$
is said to be \emph{compatible} with respect to $(X,\phi)$, if for any
fixed sequence, $\permc$, of a walk of length $\ell \leq s$, the
distribution obtained on $X$ via the uniform sampling of $y_0$, is the
same as the usual $\ell$-length walk on $X$ from any fixed initial
vertex, $x_0$. Thus, the randomness of sampling a vertex from $Y$ is
effectively \emph{transferred} to a random walk on $X$.

\begin{definition}[Compatible]
  A graph $Y$ is \emph{compatible} with respect to $(X,\phi)$
  if for every $0\leq a\leq b < s$, $\permc \in [d_2]^{b-a+1}$ and  $x_0 \in V_X$, we have\footnote{It is important to note that $\cD_Y(\matr B[a,b,\permc]) \neq \cD_Y(\matr B_U)$.}
  \begin{align*}
    \cD_X(\matr B[a,b,\permc], x_0)= \cD_X(\matr B_U, x_0) = \matr A_X^{b-a+1}x_0 \mper
  \end{align*}
\end{definition}

  This \textit{compatible} property is the same as the $0$-pseudorandom
  property in~\cite{Ta-Shma17}. We rename it as it is more of a
  structural compatibility property than a pseudorandomness one. We now prove, for the sake of completeness,  that Cayley graphs are compatible with every locally invertible graph.

\begin{lemma}[{\cite[Lemma 29]{Ta-Shma17}}]\label{lem:pseudorandom_cayley}
  Let $Y = \Cay(G^s,T)$ where $\abs{G} = d_1$.
  Then, $Y$ is compatible with respect to any $X, \phi$.
\end{lemma}
\begin{proof} Let $x_0\in V_X$ be arbitrary and let $y_0 = (r_1, \cdots, r_s) \sim G^s$ be sampled uniformly. Since $Y$ is a Cayley graph, the sequence of permutations, $J$, is equivalent to a sequence of generators $(t,^1, \cdots, t^s) \in T^s$, and the permutation is group multiplication, $y \mapsto t^k\cdot y$.

For each $1\leq i \leq s$, let $t^i \cdot y_{i-1} = (r_1^i,\cdots, r_s^i)$. The walk evolves as follows,
\begin{align*}
x_i &= X_i(x_{i-1}, t^i \cdot y_{i-1}) = x_{i-1}[r_i^i] \\
y_i &=   \phi(t^i \cdot y_{i-1}) =  \left(r_1^i,\cdots, r_{i-1}^i,  \phi(r_i^i), r_{i+1}^i,  \cdots, r_{s}^i\right)\mper
\end{align*}

Compatibility requires that $ x_i \sim A_X^i x_0$, which is inductively implied if for every $i$, $r_i^i$ is uniform over $G$ and independent of $r_j^j$ for $j \neq i$.

This clam is true initially as $y_0$ is uniform over $G^s$. Assume it is true for $y_{i-1}$. Since, $t^i, \phi$ are fixed permutations, they do not affect the uniformity of the distribution of $r^i_k$ for any $k$. Since, $\phi$ acts only on the $i^{th}$ component, the independence of $r_k^i \text{ and } r_i^i$, guaranteed by inductive claim, implies the independence of  $r_k^i \text{ and } \phi(r_i^i)$. 
\end{proof}

%--fix a sequence $\vec{z} = (z_1,\cdots, z_k)$, where $z_i \in V_X$. The number of $y_0$ such that $(z_1,\cdots, z_k) = (x_1,\cdots, x_k)$ is $d_2^{s-k}$.  We note it below,
%\tnote{Added the observation below. Should I}
For a fixed $x_0$ and $\permc$, we say that $y_0$ \textit{gives rise} to a sequence $\vec{z} = (z_1,\cdots, z_t) $ if $z_i = x_i$ where $x_i$ is as defined in the above proof.

\begin{observation}\label{obs:give_rise}
	Fix $x_0\in V_X$ and a sequence $\permc$. For any $\vec{z} = ( z_1,\cdots, z_k)$, the number of $y_0$ that \textit{gives rise} to $\vec{z}\,$ is $d_2^{s-k}$.
\end{observation}
\begin{proof}
	From the proof of~\cref{lem:pseudorandom_cayley}, we see that enforcing $z_i = x_i$ for each $i$ starting from $i=1$, forces exactly $r_i$ to be fixed. Thus, the remaining $(r_{k+1},\cdots, r_s)$ are free.
\end{proof}

\subsection{The $s$-wide Operator Norm Decay}\label{subsec:norm_decay}

We are now ready to establish the key technical lemma in the analysis of the $s$-wide replacement which is an operator-valued generalization of the scalar version of present in~\cite{Ta-Shma17}.

\begin{lemma}[Simulation Lemma (generalization of Lemma 26 from~\cite{Ta-Shma17})]\label{lemma:simulation_matrix}
  Let $0 \le s_1 \le s_2 < s$. For every pair of vectors $ z,   z' \in \vsX$, we have,
  \begin{align*}
     \ip{\prod_{i=s_1}^{s_2} \left( \XX_i \AA_{Y}\pifw\right) \left (z\otimes \vecn{1}\right )}{z'\otimes \vec{1} } ~=~  \ip{\left(\AA_X \pif \right)^{s_2-s_1+1} z}{z'}.
  \end{align*}
\end{lemma}
\begin{proof}
Let $z = \sum_x v_x\otimes x$ and $z' = \sum_x w_x\otimes x $.  Since the expression is bilinear,  it suffices to prove the
equation for $v\otimes x_0, \,  w\otimes x'$ for an arbitrary pair $(x_0,x')$.
Let $t =  s_2-s_1 +1$.
\begin{align*}\label{eqn:conditional}
  \prod_{i=s_1}^{s_2} \left(\XX_i \AA_{Y}\pifw\right) ~=~ \Ex{(j_{s_1}, \cdots, j_{s_2}) \sim [d_2]^{t} }{\prod_{i=s_1}^{s_2} \left( \XX_i\PP_{j_i} \pifw\right)}
\end{align*}

Therefore, we can fix $\permc = (j_{s_1}, \cdots, j_{s_2}) \in [d_2]^t$ and prove it for that. Recall the notation $ \cW_t$ that denotes the set of $t$-length walks on the graph $X$. Applying \cref{lem:operator_walk} to $\matr B[s_1, s_2, \permc]$, we get,
%
%For each given $\vec{x} = (x_0, \cdots, x_t)$, there are exactly
%$d_1^{s-t}$ starting vertices $y_0 = (r_1, \cdots, r_s)$ that give
%rise \todo{Explain the "give rise" part} to $\vec{x}$.
%%
%This is because,  the only requirement is that each of the $t$ constraints $x_i = x_{i-1}[\tau_{s_1+i-1,i}(r_{a + i-1})]$
%is satisfied where $\tau_{s_1+i-1,i}$ is a fixed permutation (for a given $\permc$).
%%
%Each of these equations determine one of the $r_i$'s and therefore we have $d_1^{s-t}$ free choices.
\begin{align*}
\prod_{i=s_1}^{s_2} \left( \XX_i\PP_{j_i} \pifw\right) \left (v \otimes x_0 \otimes \vecn{1}\right ) &~=~ \Ex{ (\vec{x}, \vec{y} ) \sim D(\matr B[s_1,s_2, \permc]) }{ \matr f(x_{t-1}) \cdots \matr f (x_0)\,  v  \otimes x_t \otimes y_t} \\
&~=~ \sum_{  \vec{x} \in  \cW_t } \Ex{y_0 \sim V_Y }{ \matr f (\vec{x}) \, v  \otimes x_t \otimes y_t }\mathbb{I}[y_0 \text{ gives rise to } \vec{x}],\\
&~=~ \frac{d_1^{s-t}}{d_1^s} \sum_{  \vec{x} \in  \cW_t }  \matr f (\vec{x}) \, v  \otimes x_t \otimes y_t , 
\end{align*}
where $\matr f (\vec{x}) = \matr f(x_{t-1}) \cdots \matr f (x_0)$. The last equality uses~\cref{obs:give_rise}.
Therefore, the conditioning on $y$ does not change the distribution $\cD_X$ and when we take inner products, we obtain,
\begin{align*}
  \ip{\prod_{i=s_1}^{s_2} \left(\XX_i \PP_{j_i} \pifw\right) \left (v \otimes x_0 \otimes \vecn{1}\right )}{w \otimes x'\otimes  \vec{1} }
               &~=~ \frac{d_1^{s-t}}{d_1^s} \sum_{  \vec{x} \in \cW_t }  { \ip{x_t}{x'} \ip{\matr f(x_{t-1}) \cdots \matr f (x_0)\,  v }{w}}\\
&~=~  \Ex{ \vec{x} \sim D_X(\matr B[s_1,s_2, J]) }{ \ip{x_t}{x'} \ip{\matr f(x_{t-1}) \cdots \matr f (x_0)\,v }{w}}\mper
\end{align*}

We now use\footnote{ As we only want to work with the space $\vsX$ here,  we can assume in the application of the lemma that $\abs{V_Y} = 1$.
  Else, one could directly apply \cref{eqn:connect} and use the observation that $\cD_X(\matr B_U)$ is the same as the random walk distribution on $X$.}
\cref{lem:operator_walk} for $\matr B_U$ and take inner product to get,
\begin{align*}
 \ip{\left(\AA_X \pif \right)^{s_2-s_1+1} \left (v\otimes x_0 \right )}{w\otimes x' } &~=~ \Ex{ \vec{x} \sim D_X(\matr B_U) }{ \ip{x_t}{x'} \ip{\matr f(x_{t-1}) \cdots \matr f (x_0)\,v }{w}} 
\mper
\end{align*}
From \cref{lem:pseudorandom_cayley}, we know that $Y$ is compatible
and thus, $\cD_X(\matr B[s_1,s_2,\permc]) = \cD_X(\matr B_U) $. Thus,
the right hand side (and therefore left hand sides) of these two equations above are equal.
\end{proof}

\paragraph*{The s-step Decay} Just like the amplification in \cref{sec:simple_amp} was analyzed by studying the norm decay
obtained in every two steps (\cf \cref{lem:simple_amplification}),
this amplification via the $s$-wide walks will be analyzed by bounding
the norm decay for steps of length $s$
using~\cref{lemma:simulation_matrix} similarly
to~\cite{BT08,Ta-Shma17}.
We will use the shorthand $\matr L_i \coloneqq \XX_i \pifw  \AA_{Y}$.

The goal is to bound $ \opnorm{ \matr L_{s-1}\cdots \matr L_0}$ which
controls the bias of the set obtained by $s$-long, $s$-wide walks (\cf
proof of \cref{eqn:product_pi}).
Equivalently, we will bound $\ip{(\prod_i \matr L_i)  \bv_0}{w_s}$ for any unit vectors\footnote{Here we deviate from our
  notation and use $v,w$ for vectors in $\vsXY$. } ${v}_0, {w}_s \in
\vsXY$.
We will use the orthogonal decomposition,
\begin{align*}
  \vsXY :=  \vsX \otimes \C[V_Y] =  \vsXY^\parallel \oplus \vsXY^\perp \;\text{ where }\;
  \vsXY^{\parallel} \coloneqq \textup{span}\set{z \otimes \vec{1} \mid z \in\vsX}\mper
\end{align*}

For $i \geq 1$, we inductively define the vectors $\bv_i, \bw_i, \bz_i$ and bound their norms\footnote{By definition $\norm{v_i} \leq  \norm{\AA_{Y}v_{i-1}^\perp} \le \expan(Y) \norm{v_{i-1}}. $The computation is similar for $w \text{ and }z$.},
\begin{align}
&\bv_i = \matr L_{i-1} \bv_{i-1}^\perp, \; \;\;\;\; &&\bz_{s-i} =  \left(\XX_{s-i} \pifw \right)^*\bw_{s-i+1},  \;\;\;\;
&&&\bw_{s-i}  = \left (\AA_Y\right )^* \bz_{s-i}^\perp\\
\label{eqn:normbounds}&\norm{\bv_i} \leq \expan(Y)^i , \; \;\;\;\;  &&\norm{\bz_{s-i}} \leq  \expan(Y)^{i-1}, \; \;\;\; &&&\norm{\bw_{s-i}} \leq \expan(Y)^i \mper
\end{align}

The following lemma follows readily from a calculation and we omit its proof.
\begin{lemma}
For any $\bv_0,\bw_s$ and $0 \leq r \leq s -2$ we have,
\begin{align*}
\matr L_{s-1} \cdots \matr L_{0} \bv_0 &~=~ \bv_s ~+~ \sum_{i=0}^{s-1} \matr L_{s-1} \cdots \matr L_{i}  \bv_{i}^\parallel \\
\matr L^*_{s-1} \bw_s &= \bw_{s-1} + \bz_{s-1}^{\parallel} \\
\matr L^*_{r} \cdots \matr L^*_{s-1} \bw_s &~=~ \bw_{r} ~+~ \bz_{r}^{\parallel} ~+~  \sum_{i=r+1}^{s-1} \matr L^*_{r} \cdots \matr L^*_{i-1}  \bz_{i}^\parallel
\end{align*}
\end{lemma}
%\begin{proof}
%	\todo{Add a proof?}
%\end{proof}

\begin{restatable}[Operator Generalization of Theorem 24~\cite{Ta-Shma17}]{theorem}{OpAmp2}
\label{theo:ta-shma_main}
  Let $X$ be any $d_1$-regular graph and $Y$ be a Cayley graph on $\mathbb{F}_2^{s\log d_1}$.
  Let $W_t$ be the collection of $t$-length $s$-wide walks,  on the s-wide replacement product on $X$ and $Y$.
  For any operator valued function $\matr f$ on $V_X$, such that $\max_{x \in V_X} \opnorm{\matr f (x)\,} \le 1$ and $ \opnorm{\Ex{x \in V_X}{ \matr f (x)\, } } := \lambda_0 \le~ \lambda(Y)^2 - 2 \lambda(X)$,  \begin{align*}
    \opnorm{ \Ex{(x_0,\cdots, x_{t}) \in W_t}{\matr \matr f(x_t)\cdots \matr f(x_0)} } ~\le~ \left( \lambda(Y)^s + s \cdot \lambda(Y)^{s-1} + s^2 \cdot \lambda(Y)^{s-3} \right)^{ \lfloor t/s \rfloor }\mper
  \end{align*}
\end{restatable}
\begin{proof}
Using \cref{lem:operator_walk}, we can repeat the proof
of \cref{eqn:product_pi} to see that,
\begin{align*}
  \opnorm{\Ex{(x_0, \cdots, x_t) \in W_t}{ \matr f (x_t)  \cdots \matr f (x_0)} \,} \le \opnorm{\matr L_{t} \cdots \matr L_{0}} \le \opnorm{\matr L_{s-1} \cdots \matr L_{0}}^{\lfloor t/s\rfloor}\mper
\end{align*}
\begin{align*}
 \ip{\matr L_{s-1} \cdots \matr L_{0} \bv_0}{\bw_s} ~&=~ \ip{\bv_s}{\bw_0} ~+~ \sum_{r=0}^{s-1} \ip{ \matr L_{s-1} \cdots \matr L_{r}  \bv_{r}^\parallel }{ \bw_s} \\
      ~&=~ \ip{\bv_s}{\bw_s} ~+~ \sum_{r=0}^{s-1} \ip{\bv_{r}^\parallel}{ \matr L^*_{r} \cdots \matr L^*_{s-1} \bw_s} \\
      ~&=~ \ip{\bv_s}{\bw_s} ~+~\sum_{i=0}^{s-1}  \ip{\bv_{r}^\parallel}{\bw_{r} + \bz_{r}^\parallel}  ~+~ \sum_{r=0}^{s-2} \sum_{ i = r+1}^{s-1}  \ip{\bv_{r}^\parallel}{ \matr L^*_{r} \cdots \matr L^*_{i-1} \bz_i^\parallel} \\
                              ~&=~ \ip{\bv_s}{\bw_s} ~+~ \sum_{i=0}^{s-1}  \ip{\bv_{r}^\parallel}{ \bz_{r}^\parallel}  ~+~ \sum_{r=0}^{s-2} \sum_{ i = r+1}^{s-1}  \ip{\bv_{r}^\parallel}{ \matr L^*_{r} \cdots \matr L^*_{i-1} \bz_i^\parallel} \mper
\end{align*}

\noindent The last step uses  $\ip{\bv_r^\parallel}{\bw_r} = \ip{\AA_Y\bv_r^\parallel}{\bz_r^\perp} = 0$.
Using \cref{eqn:normbounds}, we get $ \ip{\bv_{r}^\parallel}{
  \bz_{r}^\parallel} \leq \expan(Y)^{s-1}$. To bound the last term, we
finally use~\cref{lemma:simulation_matrix}.  Let $\bv_r^\parallel =
v'_r \otimes \vec{1}$, and $\bz_i^\parallel = z'_i \otimes \vecn{1}$.
Then,
\begin{align*}
  \ip{\bv_{r}^\parallel}{ \matr L^*_{r} \cdots \matr L^*_{i-1} \bz_i^\parallel} &~=~  \ip{\bv'_{r}}{ \left( \AA_X \pif \right)^{i-r} \bz'_i} && \text{(Using \cref{lemma:simulation_matrix})}\\
               &~\leq~  \opnorm{\left( \AA_X \pif \right)^{i-r} }\norm{ z'_i} \norm{v'_{r}}    \\                        
   &~\leq~ \expan(Y)^{ 2 \lfloor \frac{i-r}{2} \rfloor  }  \expan(Y)^{r+s-i-1} \leq  \expan(Y)^{s-3} ,
  \end{align*}
   where the penultimate inequality
   uses~\cref{cor:ta-shma_simp_matrix} and plugs in the assumption
   that $2 \lambda(X) + \opnorm{\Ex{x \in V_X}{ \matr f (x)\, } } \le
   \lambda(Y)^2$.  Substituting this back in our expression above
   gives us the result.
\end{proof}

\input{cons_param.tex}

%% file: cons_param.tex
\subsection{Instantiating the $s$-wide Replacement Product }\label{sec:parameters}

The goal of this section is to prove the following result that implies \cref{theo:main_2}.

\begin{theorem}[Almost Ramanujan Expanders I]\label{lemma:almost_ram_exp_1}
  Let $\Cay(G,S)$ be $\expan_0$-expander with constant $\expan_0 \in (0,1)$.
  For every function $\beta (\lambda) > 0$,
  and for any  $\expan > 0$, sufficiently small such that
  \begin{align*}
    \frac{32}{\beta (\lambda) } ~\le~ \left ( \frac{\log(1/\expan)}{4\log\log(1/\expan)}\right )^{1/3}\mcom
  \end{align*}
  there exists a deterministic polynomial time algorithm to construct $S'$ such that $\Cay(G,S')$ is a $\expan$-expander with
  degree $\abs{S'} = O_{\expan_0}(\abs{S}/\expan^{2+\beta})$.
  
  Furthermore, each element in $S'$ is the product of $O(\log(1/\lambda))$ elements of $S$.
\end{theorem}

 \subsubsection*{Overview} {We will explicitly construct the graphs $X$ and
$Y$ as needed for the $s$-wide product}.  Once we
obtain the graphs, we identify the vertices of $X$, \ie $V_X$ with the
initial generating set $S$, or perhaps a slightly modified set $S'$, obtained by duplicating and adding identities. The
final set is obtained by multiplying elements along each
$(t-1)$-length walk on the $s$-wide replacement product of $X$ and
$Y$.  The way we choose parameters and objects for it borrows heavily from
Ta-Shma's arguments in~\cite{Ta-Shma17}.  The analysis follows an
analogous structure of~\cite{JQST20} for binary codes, which in turn
builds on the original analysis of Ta-Shma~\cite{Ta-Shma17}. We will also use the following result from that work,

\begin{lemma}[Based on Lemma 6~\cite{Ta-Shma17}]\label{lemma:aghp}
  For every $m \in \mathbb{N}^+$ and $d = 2^{2k} \leq 2^m$, there exists a fully explicit set 
  $A \subseteq \mathbb{Z}_2^m$ such that the graph $\textup{Cay}(\mathbb{Z}_2^m,A)$ is a
  $(2^m,d,\lambda= \frac{m}{\sqrt{d}} )$-expander graph.
\end{lemma}

%  Let $s$ be the smallest power of $2$ greater than $\frac{32}{\beta}$ and let $d_2 = s^{4s}$.
%We will only summarize the construction here and show that the choice of the parameters does in fact yield our main result. Detailed computation and verification is present in \cref{sec:param_app}.
\paragraph{The construction}Given as input $n := |S|, \expan$ and a slowly growing function $\beta(\lambda)$, we construct the graphs $X,Y$ as described below with the parameters $(s,d_1, d_2, \lambda_1, \lambda_2)$ which are all functions of $\lambda$ and $\beta(\lambda)$. These are summarized in a table below for reference\footnote{The choice of parameters is similar but not identical to Ta-Shma's choice.}. Recall that a $(n,d,\lambda)$-graph has $n$ vertices, is $d$-regular,
and has the second largest singular value of its normalized adjacency matrix at most $\lambda$. 
  \begin{itemize}
    \item \textbf{Outer Graph} --- The outer graph $X$ will be an $(n',d_1,\lambda_1)$-graph which is a Cayley graph on $\slin_2(p)$ constructed using \cref{cor:small_expander} with $(n, \lambda_1)$ as input. By \cref{exam:cayleylocal}, we obtain a locally invertible graph on $n' \approx n$. The condition on the size is satisfied as $n = 2|S|d_2^{5} \geq d_2^{5} \geq 2^{2^{17}}$ by the assumption that $s\geq 2^{10}$.  Moreover, the degree is $\frac{c}{\lambda_1^{2*4.1}} \le \frac{c d_2^{4.1}}{b_2^{8.2} } \le d_2^5$.  We increase its degree to $d_2^5$ by taking multiple copies of the generating set which does not change bias\footnote{This is wasteful, but we do it to ensure that $V(Y) =  d_1^s$ and that $d_1^s$ is a power of $2$. }. Thus, we obtain a $(n', d_1, \lambda_1)$-graph where $n' = n + O(n^{8/9})$.

    \item \textbf{Inner Graph} --- The inner graph $Y$ will be a $(d_1^s,d_2,\lambda_2)$-graph which is a Cayley graph on $\Z_2^m$ and therefore by \cref{lem:pseudorandom_cayley},
          it is compatible. For this, we use the construction of Alon et al.\@~\cite{AGHP92} and the analysis of Ta-Shma (\cref{lemma:aghp}).
  \end{itemize}
  
%\paragraph*{The outer graph $X$.} We use our construction of expander from \cref{cor:small_expander} to construct a graph on $n' \approx n$ vertices with expansion $\lambda_1 = \frac{\lambda_2^2}{10}$.  
%\paragraph*{The inner graph $Y$.} We obtain a Cayley graph $Y = \textup{Cay}(\mathbb{Z}_2^{\log(n_2)},A)$
%such that $Y$ is an $(n_2=d_2^{5s}, d_2, \lambda_2)$
%graph\footnote{Notice that since $s$ (and therefore $d_2, n_2$) is
%  chosen to be a power of $2$, the conditions of \cref{lemma:aghp} are
%  satisfied.}. The set $A$ of generators comes from a small bias code
%derived from a construction of Alon et al.~\cite{AGHP92}, but we will
%rely on Ta-Shma's analysis.

%  
We summarize the construction and the choice of parameters $n', d_1,d_2, \lambda_1,\lambda_2$ and $s$, which are chosen as follows for a fixed $\beta(\lambda)$ here - 
\begin{center}
\fbox{\begin{minipage}[c]{32em}	
$ s \text{ is the smallest power of 2 such that } \frac{32}{\beta } \le s \le \left ( \frac{\log(1/\expan)}{4\log\log(1/\expan)}\right )^{1/3} $
\vskip 0.3cm
Every other parameter is a function of $s$.
\vskip 0.3cm
$Y: (n_2,d_2,\lambda_2),\quad n_2=d_2^{5s},\quad d_2=s^{4s},\quad \lambda_2\le\frac{b_2}{\sqrt{d_2}},\quad b_2 = 5s \log d_2$
\vskip 0.3cm
$X: (n',d_1,\lambda_1),\quad n' \approx n = O(\abs{S}d_2^{5}),\quad d_1=d_2^{5},\quad \lambda_1 = \frac{\lambda_2^2}{10}$
\vskip 0.3cm
$t: \text{ smallest integer such that } (\lambda_2)^{(1-5\alpha)(1-\alpha)(t-1)} \leq \expan,\,;\; \text{ where }\alpha = 1/s$ 
\end{minipage}}
\end{center}
\textbf{Note: } We assume that $s \ge 2^{10}$ since otherwise
$\expan$ is a constant and we can use \cref{theo:exp_walk}.

\noindent Now, we mention the central claim that we need from our choice of parameters.

%\tBound*
 \begin{restatable}{claim}{tBound}\label{claim:t_bound}
The selection of the parameters above implies the following bounds on $t$, 
\begin{enumerate}
\item[\textup{(i)}.]  $t-1 \ge 2s^2$ 
\item[\textup{(ii)}.]   $(d_2)^{(t-1)}  \le  \expan^{-2(1+10\alpha)}$,
\end{enumerate}
\end{restatable}
\begin{proof}
 Proof of (i) : Using $d_2 = s^{4s}$ and the upper bound on $s$, we have
 \begin{align*}
  \left(\frac{1}{\lambda_2}\right)^{(1-5\alpha)(1-\alpha)2s^2} & ~\le~ \left(\frac{1}{\lambda_2}\right)^{2s^2} ~=~ \left(\frac{d_2}{b_2^2}\right)^{s^2} ~\le~ \left(d_2\right)^{s^2} = s^{4s^3}\\
                                 & ~=~ 2^{4s^3 \log_2(s)}  ~\le~  2^{\log_2(1/\expan)} ~=~ \frac{1}{\expan}\mper
  \end{align*}
  Hence, $(\lambda_2)^{(1-5\alpha)(1-\alpha)s/\alpha} \ge \expan$ and thus
  $t-1$ must be at least $2s^2$.  Also observe that,
  \begin{align}
    \lambda_2^{(1-5\alpha)(1-\alpha)^2 (t-1)} &~=~ \lambda_2^{(1-5\alpha)(1-\alpha) (t-2) \left(\frac{(1-\alpha)}{1-1/(t-1)} \right )  } \\
                                       &~\ge~  \lambda_2^{(1-5\alpha)(1-\alpha) (t-2) }  && \text{($t-1 \ge s =  1/\alpha$)}\\
               \label{eq:t_upperbound} &~\ge~ \expan && \text{(From  the choice of minimal $t$)} 
  \end{align} 
 
  Since $b_2 = 5 s \log_2(d_2) = 20 s^2 \log_2(s) \le s^4$ (recall that $s=1/\alpha \ge 2^{10}$), 
  \begin{align*}
    d_2^{1-2\alpha} ~=~ \frac{d_2}{d_2^{2\alpha}} ~=~ \frac{d_2}{s^8} ~\le~ \frac{d_2}{b_2^2} ~=~ \frac{1}{\lambda_2}\mper
    \end{align*}
  We obtain claim (ii) by the following computation,
  \begin{align*}
    (d_2)^{(t-1)} &~\le~ {\lambda_2}^{\frac{-(t-1)}{1-2\alpha}} \\
                &~\le~ \expan^{\frac{-2}{(1-2\alpha)(1-5\alpha)(1-\alpha)^2}} && \text{(Using~\cref{eq:t_upperbound})}\\
                &~\le~ \expan^{-2(1+10\alpha)}\mper  \hfill \qedhere 
  \end{align*}
\end{proof}

%  The parameters $n', d_1,d_2, \lambda_1,\lambda_2$ and $s$ are chosen as follows for a fixed $\beta(\lambda)$. \footnote{\textbf{Note: }While we let $\beta$ be a
%    function of $\lambda$, it might be instructive to make the
%    simplifying assumption that it is an arbitrarily small constant. We will denote it simply as $\beta$ from now on.}
%  
%  
%  \begin{center}
%\fbox{\begin{minipage}[c]{32em}	
%$ s \text{ is the smallest power of 2 such that } \frac{32}{\beta } \le 2^{10} \le s \le \left ( \frac{\log(1/\expan)}{4\log\log(1/\expan)}\right )^{1/3} $
%\vskip 0.3cm
%Every other parameter is a function of $s$.
%\vskip 0.3cm
%$Y: (n_2,d_2,\lambda_2),\quad n_2=d_2^{5s},\quad d_2=s^{4s},\quad \lambda_2\le\frac{b_2}{\sqrt{d_2}},\quad b_2 = 5s \log d_2$
%\vskip 0.3cm
%$X: (n',d_1,\lambda_1),\quad n' \approx n = \Theta(\abs{S}d_2^{5}),\quad d_1=d_2^{5},\quad \lambda_1 = \frac{\lambda_2^2}{10}$
%\vskip 0.3cm
%$t: \text{ smallest integer such that } (\lambda_2)^{(1-5\alpha)(1-\alpha)(t-1)} \leq \expan,\,;\; \text{ where }\alpha = 1/s$ 
%\end{minipage}}
%\end{center}

\begin{lemma}
  The number of walks of length $t-1$ on the $s$-wide replacement product of $X$ and $Y$ is $O(\abs{S}/\expan^{2+ \beta})$.
\end{lemma}

\begin{proof}
  Since each step of the walk has $d_2$ options, the number of walks is 
  \begin{align*}
    \abs{V(X)} \abs{V(Y)} \cdot d_2^{(t-1)} ~=~ n' \cdot d_1^s \cdot d_2^{(t-1)} &= n' \cdot  d_2^{(t-1)+5s} \\
                                          &~=~ \Theta\left(\abs{S} \cdot d_2^{(t-1)+5s+5}\right)\\
                                          &~=~ O\left(\abs{S} \cdot d_2^{(1+5\alpha)(t-1)}\right)\mper
  \end{align*}

  which from \cref{claim:t_bound} (ii),  implies a size of
  \begin{align*}
     O\left(\abs{S} \cdot d_2^{(1+5\alpha)(t-1)}\right) ~=~ O\left(\frac{\abs{S}}{\expan^{2(1+10\alpha)(1+5\alpha)}}\right) ~=~ O\left(\frac{\abs{S}}{\expan^{2+32\alpha}}\right) ~=~ O\left(\frac{\abs{S}}{\expan^{2+\beta}}\right) \mper \qedhere
  \end{align*}
\end{proof}

Before we prove the main result, we need the following simple observation that will be used to construct a modified $(\varepsilon + o(1))$-biased set starting from an $\epsilon$-biased set, $S$. This is needed because the graph obtained from~\cref{cor:small_expander} does not have size exactly $|S|$ but is only guaranteed to be at most $(1+o(1))|S|$.

\begin{lemma}\label{lem:remove_gen}
  Let $S$ be an $\epsilon$-biased set of a group $G$.
  And let $S'$ be obtained by adding $\theta\abs{S}$ many identity elements.
  Then, $S'$ is an $(\epsilon +\theta)$-biased set.
\end{lemma}

\begin{proof}
  Denote by $e$ the identity element of $G$.
  Let $\rho$ be any non-trivial irreducible representation of a group $G$.
From the computation we have,
  \begin{align*}
    \norm{\E_{s\in S'} \rho(s) }_\textup{op} &~=~  \frac{1}{1+\theta} \opnorm{ \E_{s\in S} \rho(s)  + \theta\cdot \E_{s\in S\setminus S'} \rho(1)  } \\
    &~\leq ~  \opnorm{ \E_{s\in S} \rho(s)} + \theta   && \text{($\opnorm{\rho(1)}  = 1$)} \\
&~\leq~  \epsilon+ \theta  && \text{($S \text{ is }\;  \epsilon\text{- biased}$)} \mper \qedhere
  \end{align*}
\end{proof}

\begin{theorem}[Almost Ramanujan Expanders I]\label{lemma:almost_ram_exp_1}
  Let $\Cay(G,S)$ be $\expan_0$-expander with constant $\expan_0 \in (0,1)$.
  For every function $\beta (\lambda) > 0$,
  and for any  $\expan > 0$, sufficiently small such that
  \begin{align*}
    \frac{32}{\beta (\lambda) } ~\le~ \left ( \frac{\log(1/\expan)}{4\log\log(1/\expan)}\right )^{1/3}\mcom
  \end{align*}
  there exists a deterministic polynomial time algorithm to construct $S'$ such that $\Cay(G,S')$ is a $\expan$-expander with
  degree $\abs{S'} = O_{\expan_0}(\abs{S}/\expan^{2+\beta})$.
  
  Furthermore, each element in $S'$ is the product of $O(\log(1/\lambda))$ elements of $S$.
\end{theorem}

\begin{proof}
We can assume that $s \ge 2^{10}$ since otherwise $\expan$ is a constant and we can just use \cref{theo:exp_walk}.

\noindent \textbf{Initial Boost} \enspace We first boost the expansion from $\expan_0$ to $1/d_2 \le \lambda_2^2/3$.
Using~\cref{theo:exp_walk} (with its parameter $\beta$ equal to $1$), we can find a new set of generators,
$S_1$, such that $\Cay(G,S_1)$ is $1/d_2$-spectral expander and $\abs{S_1} = O(\abs{S} d_2^{5})$.
Moreover, we also know that each element in $S_1$ is a multiple of at
most $\log\left(d_2^{5} \right)$ elements in $S$. We add multiple
copies of the entire set to make the size $\abs{S} d_2^{5}$.

\paragraph*{The $s$-wide walk} Obtain an $(n',d_1,\lambda_1)$ Cayley graph $X$ from~\cref{cor:small_expander}
 as explained before. We add $n'-n = O(n^{8/9})$ copies of the identity to $S_1$ to obtain
 $S_2$.  By \cref{lem:remove_gen} and the assumption that $s \geq
 2^{10}$, $S_2$ is a $\lambda_2^2/3 + O(n^{-1/9}) \leq
 2\lambda_2^2/3$-biased set. We denote by $S'$ the final set of
 generators obtained by $t$ steps of the $s$-wide replacement product
 of $X$ and $Y$. By definition, each element in $S'$ is a product of
 $t$ elements in $S_2$ which has the same elements as $S_1$. Thus,
 each element in $S'$ is a product of at most
\begin{align*}
   O(t\log(d_2)) &~\leq~ O((1+10\alpha) \log(1/\lambda))             && \text{(Using \cref{claim:t_bound} [ii])}\\
                 &~\leq~ O(\log(1/\lambda))                          && \text{(By the assumption that $\alpha \leq 1/128$)}
\end{align*}
 elements of $S$.   The only thing that remains is to prove expansion of $\Cay(G,S')$. We pick any irreducible representation $\rho$ and apply~\cref{theo:ta-shma_main} to the function $\rho$ on $S_2 \leftrightarrow V(X)$.
 The condition that $2\lambda(X) + \opnorm{\Ex{g \sim S_2}{\rho(g)} } \leq \lambda(Y)^2 $ translates to 
 $\lambda_1 \leq \lambda_2^2/6 $ which is satisfied by our choice of $\lambda_1$.
 Thus, the final expansion is given by, 
  \begin{align*}
    \opnorm{\Ex{g \in S'}{\rho(g)}} &\coloneqq \left(\lambda_2^{s} + s \cdot \lambda_2^{s-1} + s^2 \cdot \lambda_2^{s-3}\right)^{\lfloor (t-1)/s \rfloor}\\
      &~\le~ \left (3s^2\lambda_2^{s-3}\right )^{((t-1)/s)-1} && \left (\text{Using  $\lambda_2 \le\frac{20s^2 \log s}{s^{2s^2}} \leq \frac{1}{3s^2}$}\right)\\
      &~\le~ \left (\lambda_2^{s-4} \right )^{(t-1-s)/s} \\
      &~\le~ \lambda_2^{(1-5/s)(1-s/(t-1))(t-1)}\\
      &~\le~ \lambda_2^{(1-5\alpha)(1-\alpha)(t-1)}                && \left (\text{Using ~\cref{claim:t_bound} [i]}\right )  \\
      &~=~ \lambda_2^{(1-5\alpha)(1-\alpha)(t-1)} \le \expan,      && \text{(From the choice of $t$)} \mper \qedhere
  \end{align*}
\end{proof}

%% file: applications.tex
\section{Some Applications}\label{sec:appl}

Our operator amplification leads to almost optimal explicit
constructions of many pseudorandom objects (from existing suboptimal
ones): transforming arbitrary {regular} expander graphs into almost-Ramanujan
expanders (\cref{subsec:arb_exp_amp}), quantum expanders
(\cref{subsec:qua_exp}), monotone expanders (\cref{subsec:mon_exp}),
to generating sets with improved (average) Kazhdan constants
(\cref{subsec:kazhdan}) and to dimension expanders
(\cref{subsec:dim_exp}). These pseudorandom objects embody various
notions of expansion.

\paragraph*{Permutation Amplification}

The key to these applications is observing that the adjacency matrix
of an arbitrary graph and that of a monotone expander can be written
as a sum of permutation matrices which can be interpreted as $\matr
P_\sigma = \rho_{\text{def}}(\sigma)$ for the \emph{defining} (or
\emph{natural}) representation $\rho_{\text{def}}$.  We plug in the
collection of these permutations $\set{\sigma}$ in our amplification
machinery to obtain almost optimal spectral expanders and monotone
expanders.

\paragraph*{Almost Ramanujan Expanders for the Symmetric Group}
Constructing constant size expanding generating sets for the symmetric
group was quite challenging (even non-explicitly). In a breakthrough
work~\cite{Kas07}, Kassabov provided the first family of such
expanding generators which was also explicit. However, this family was
not close to the Ramanujan bound and no such generating set was
known. \cref{theo:main_2} lets us amplify Kassabov's generating set to
one close to optimum bound showing that the symmetric group has
explicit almost Ramanujan Cayley expanders. The same obviously holds
for every expanding group.

\paragraph*{Quantum Expanders} A quantum expander is a generalization of
an expander graph having many applications in quantum information
theory \cite{AS04, BAST08, Has07, Hastings07b, HH09,
  AHLNSV14}. Harrow~\cite{Harrow07} showed that Cayley graphs can be
used to construct quantum expanders inheriting the expansion of the
starting Cayley graph. However, the construction is only explicit if
the group admits an efficient quantum Fourier transform (QFT). Since
we can now obtain almost Ramanujan Cayley graphs for the symmetric
group which has a known efficient QFT~\cite{Beals97}, we obtain the
first explicit almost Ramanujan quantum expanders.

\paragraph*{Improving the Kazhdan Constant}
The \emph{Kazhdan constant} $\mathcal{K}(G,S)$ of a finitely generated
group $G$, with respect to a generating set $S$, is a quantitative
version of Property $(\mathrm{T})$ which has been used to construct
explicit expanders (\eg Margulis~\cite{Margulis88}).
We show that this can be amplified by considering a slightly different
version called the \emph{average Kazhdan constant} which directly
relates to the bias of the set $S$.  This is interesting as typically
the bound on the Kazhdan constant is used to construct expanders but
here we construct expanding generating sets to improve the
constant!  The improved constants and the generating sets have
algorithmic implications and we mention two of them.
\begin{enumerate}[label=$\cdot$]
  \item \emph{Dimension expanders} - Lubotzky and Zelmanov~\cite{LZ08} showed that the image of a generating set of a group under an
         irreducible representation gives a dimension expander and its expansion is controlled by its Kazhdan constant.
  \item \emph{Product replacement algorithm} - uses random walks on $k$-tuples of groups elements.
        Lubotzky and Pak~\cite{LP00} showed that the mixing time of the algorithm relates to the Kazhdan constant of certain lattice groups like $\mathrm{SL}_n(\Z)$, {assuming Property $(\textup{T})$}.
        This crucial assumption was proven in complete generality\footnote{\cite{KKN21} prove that $\mathrm{Aut}(\mathrm{F}_n)$, the automorphism
          group of the free group generated by $n$ elements, has Property $(\textup{T})$. This implies Property $(\textup{T})$ for quotients of $\mathrm{Aut}(\mathrm{F}_n)$, which includes $\mathrm{SL}_n(\Z)$.}
        recently by Kaluba, Kielak and Nowak~\cite{KKN21}. In particular, we have a mixing time bound of $\frac{4\log{\abs{G}}}{\mathcal{K}(G,S)^2}$.  
\end{enumerate}
Using our amplified generating set (\cref{theo:kazhdan}), we can improve both these results.

%\paragraph*{Sampling Group elements} 
%Another application of having almost optimal Ramanujan Cayley graphs
%is to sample random group elements efficiently. Given a Cayley graph,
%$\Cay(G,S)$, one can consider a random walk on $G$ which starts at an
%arbitrary vertex $g$ and at each step moves to a random neighbor $g
%\to sg$.  Spectral expansion guarantees that walks mix quickly, \ie in
%at most $ O_{\lambda}(\log |G|)$ steps (See \cite{HooryLW06}). The
%amount of randomness used in each step is $\log d$ and since the
%degree versus expansion trade-off is now almost optimal, we can
%achieve the same convergence guarantee using a smaller degree and thus
%the random walk is more efficient in terms of randomness.\todo{rev3--More details? Or delete this.}
%\snote{}

\subsection{Permutation Amplification}

The \emph{defining representation} - ($\rho_{\text{def}}(\sigma),
\C^n$) for $\mathrm{Sym}_n$ is defined as the representation that maps
a permutation to the matrix defining it. More formally,
$\rho_{\text{def}}(\sigma) e_i = e_{\sigma(i)}$ for every unit basis
vector $e_i$ of $\C^n$. It is a fact that $\mathcal{V}_{\text{def}} =
\mathcal{V}_{\text{triv}} \oplus \mathcal{V}_{\text{standard}}$ where
$\mathcal{V}_{\text{standard}}$ is an irreducible non-trivial
representation. Note that if we are given a set $\{\matr P_1, \cdots,
\matr P_r \}$ of permutation matrices acting on $\C^n$, we can
identify a set $S =\{ \sigma_1, \cdots, \sigma_r\} \subseteq
\mathrm{Sym}_n$ such that $\rho_{\text{def}}\,(\sigma_i) = \matr P_i$.

\begin{corollary}[Permutation Amplification]\label{cor:perm_amp}
  Let $\matr P = \set{\matr P_1, \cdots, \matr P_r }$ be a collection of permutation matrices such that $\lambda(\Ex{i \sim [r]}{ \matr P_i }) \leq \expan_0$.
  Then,  for any $\expan \in (0,1)$,  we can explicitly construct a collection $\matr P'$ such that  
  \begin{enumerate}[topsep=2pt, itemsep=2pt]
    \item $\lambda\left(\Ex{ \matr M \sim \matr P'}{ \matr M }\right) \leq \expan$,
    \item $\abs{\matr P'} \le O\left (\abs{\matr P}/ \expan^{2 + o(1)}\right )$  and
    \item each $\matr P'_i \in \matr P'$ is a product of at most $O_{\expan_0}\left ( \log(1/\expan) \right)$ many matrices from $\matr P$.
  \end{enumerate}
\end{corollary}

\begin{proof}
  Let $\matr P_i = \sigma_i$. Applying \cref{theo:ta-shma_main} to the set $S = \{\sigma_i\}$ we get a larger set of permutations, $S'$
  of the form $\sigma' = \sigma_{i_1} \circ \cdots \circ \sigma_{i_k}$ where $k = O_{\expan_0}\left ( \log(1/\expan) \right)$.  By the decomposition
  of the defining representation, we have that
 \begin{align*}
   \spec\left (\Ex{ \matr M \sim \matr P'}{ \matr M }\right ) &= \spec\left (\Ex{ \sigma' \sim S'}{ \rho_{\text{def}}(\sigma') }\right )\\
                                                              &=  \{1\} \cup \spec\left (\Ex{ \sigma' \sim \matr S'}{ \rho_{\text{standard}}(\sigma') } \right  )\mper
 \end{align*}
 where the $1$ corresponds to the eigenvalue from the trivial
 representation. Since the operator amplification reduces the bias of
 every non-trivial irreducible representation, it also does so for
 $\mathcal{V}_{\text{standard}}$.
\end{proof}

\input{perm_amp}

\input{quan_exp}

\input{mon_exp}

\input{kahzdan}

\input{dim_exp}

\input{diameter}

%% file: perm_amp.tex
\subsection{Arbitrary Expanders via Permutation Amplification}\label{subsec:arb_exp_amp}

We can make any family of bounded degree expander graphs into an
almost Ramanujan family while preserving their adjacency structure. First, we recall
K\"{o}nig's theorem that says that the adjacency matrix of a
$d$-regular graph can be expressed in terms of permutation matrices.

\begin{theorem}[K\"{o}nig]\label{theo:konig}
  Let $\matr A_X$ be normalized adjacency matrix of a $d$-regular $n$-vertex simple graph $X$.
  Then, there exists $d$ permutation matrices $\matr P_1,\ldots,\matr P_d \in \mathbb{R}^{n\times n}$ such that
  \begin{align*}
    \matr A_X = \frac{1}{d} \sum_{j=1}^d \matr P_j.
  \end{align*}
\end{theorem}

It is also efficient to find permutation matrices as above.

\begin{claim}\label{claim:konig_eff_decomp}
  The permutations in~\cref{theo:konig} can be found in time $\poly(n)$.
\end{claim}

\begin{proof}
 We view $A_X$ as encoding the adjacency relation of a bipartite graph
 with vertex bipartition $(A=V(X),B=V(X))$. This bipartite graph is
 $d$-regular so it has at least one perfect matching $M$, which can be
 found in $\poly(n)$ time. We remove this matching $M$ obtaining
 a $(d-1)$-regular graph and we repeat till the resulting graph is empty.
\end{proof}

Our general transformation into an almost Ramanujan
bound follows by using~\cref{claim:konig_eff_decomp} to obtain an initial set of permutation matrices which are amplified using~\cref{cor:perm_amp}.

\begin{restatable}[Main I (Formal version of~\cref{theo:main_1_informal})]{theorem}{MainPermAmpInformal}\label{theo:main_1}
  Let $\set{X_i}_{i \in \mathbb{N}}$ be a family of $d_0$-regular $\lambda_0$-expanders with constant $\lambda_0 < 1$.
  For any $\lambda \in (0,1)$ and any expander $X_i$, we can deterministically compute a $d$-regular
  $\lambda$-expander $X_i'$ with $d=O_{\lambda_0}(d_0/\lambda^{2+o(1)})$ in time $\poly(\abs{V(X_i)})$.
  Moreover, the construction is local in the sense that edges in $X_i'$ correspond to short walks in $X_i$.
  More precisely, if the adjacency matrix of $X_i$ is
  $$
  \matr A_{X_i} = \frac{1}{d_0} \sum_{j=1}^{d_0} \matr P_j,
  $$
  where $\matr P_1,\ldots,\matr P_{d_0}$ are permutation matrices, then the adjacency matrix of $X_i'$ is
  $$
  \matr A_{X_i'} = \frac{1}{d} \sum_{j=1}^d \matr P_j',
  $$
  where each $\matr P_j'$ is the product of at most $k=O_{\lambda_0}(\log(1/\lambda))$
  permutation matrices among $\matr P_1,\ldots,\matr P_{d_0}$.
\end{restatable}

%% file: quan_exp.tex
\subsection{Explicit Almost Ramanujan Quantum Expanders}\label{subsec:qua_exp}

Quantum expanders were defined in \cite{AS04, BAST08, Hastings07b} and
have found many applications in quantum information theory.
For instance, they can be used in the construction of designs and
gates sets~\cite{HH09}, in quantum statistical zero knowledge (QSZK)
\cite{BAST08}, in detecting EPR pairs~\cite{AHLNSV14} and in the study
of \emph{entanglement} \cite{Has07}.

Roughly speaking, a quantum expander is a generalization of an
expander graph (see \cref{def:quantum_exp} for precise details). While
a usual degree-$d$ expander graph $X=(V,E)$ is given by $d$
permutation matrices acting on a vector space $\C[V]$, a quantum
expander is given by $d$ (suitable) linear operators acting on quantum
states (\ie PSD matrices of trace $1$). The normalized adjacency
matrix of a $\lambda$-expander shrinks the $\ell_2$-norm of vectors
orthogonal the all ones function by a factor of $\lambda$. Similarly,
a quantum expander shrinks the Frobenius norm of PSD matrices
orthogonal~\footnote{With respect to the Hilbert--Schmidt inner
  product.} to the identity matrix (the quantum analogue of the all
ones function) by a factor of $\lambda$ (the quantum expansion
parameter).
  
  \begin{definition}[Quantum Expander~\cite{AHLNSV14}]\label{def:quantum_exp}
The (super) operator $\Phi: \C^{N\times N} \to \C^{N\times N}  $ is an $(N,d,\lambda)$ quantum expander if  
  \begin{enumerate}[leftmargin=*,label =$\cdot$]
    \item ``Degree'' -- The operator $\Phi$ can be expressed as a sum of $d$ linear operators as follows, $\Phi(\rho) = \sum_{i = 1}^d B_i\rho B_i^\dagger$
          where\footnote{A useful special case is when each $B_i$ is a (normalized) unitary.} $\sum_{i=1}^d B_i^\dagger B_i = \matr I_N$. 
    \item ``Expansion'' -- The second largest eigenvalue\footnote{If $\rho$ satisfies $\Tr(\rho) = 0$,
          then $\norm{\Phi(\rho)}_2 \le \lambda \norm{\rho}_2$, where $\norm{\rho}_2 \coloneqq \sqrt{\Tr(\rho^\dag \rho)}$.} of $\Phi$ as a linear map is $\leq \lambda$.
  \end{enumerate}
\end{definition}

In~\cite{Hastings07}, Hastings showed that the Ramanujan bound also
applies to quantum expanders and that $d$ random unitaries get
arbitrarily close to the bound.  However, such a construction cannot
be efficiently implemented and thus used in applications like
\cite{AHLNSV14} which rely on existing explicit constructions
(\eg~\cite{BAST08,Harrow07}) that are far from the Ramanujan bound and
thus give sub-optimal results. 

{Harrow~\cite{Harrow07} proved that one can construct a quantum expander using an expander Cayley graph over a group for which efficient Quantum Fourier Transform (QFT) is known \cite{Beals97}. 
}
%We deduce the existence of explicit families of almost Ramanujan
%quantum expanders 
%together with a result of Harrow~\cite{Harrow07}. For this, it is
%important that we can efficiently construct almost Ramanujan Cayley
%expanders on the symmetric group $\mathrm{Sym}_n$, . \todo{Explain why Harrow's result is not explicit}

\begin{theorem}[Harrow~\cite{Harrow07}]
  Let $G$ be a group and $S \subseteq G$ be a multiset such that $\Cay(G,S)$ is a $\lambda$-spectral expander.
  Let $V^\mu$ be an irreducible representation of $G$ of dimension $N$. Then, there exists an $(N, \abs{S},\lambda)$-quantum
  expander. Furthermore, if the group $G$ admits an efficient QFT and $\log N = \Omega(\log \abs{G})$,
  then the quantum expander is explicit.
\end{theorem}

{Until now, we did not have almost-Ramanujan expanders over such a group. Since the symmetric group admits such a QFT algorithm, we deduce the existence of explicit families of almost Ramanujan
quantum expanders by applying our amplification to the Cayley graphs over the symmetric group due to Kassabov~\cite{Kas07}. }
%As a corollary of Harrow's result and our explicit family of almost
%Ramanujan Cayley expanders over the symmetric group obtained from the
%expanding family of Kassabov~\cite{Kas07}, we deduce the following
%corollary.

\TheoQuanExp

%% file: mon_exp.tex
\subsection{Explicit Almost Ramanujan Monotone Expander}\label{subsec:mon_exp}

We now show how to obtain almost Ramanujan monotone expanders starting
from the explicit construction in Bourgain and
Yehudayoff~\cite{BY13}. First, we recall the definition of a
monotone graph. All graphs we consider are undirected.

\begin{definition}[Monotone partial map]
A partial map $f :[n] \to [n]$ is monotone if for every pair $\{x, y\}$ for which $f$ is defined, if $x <y$, we have $f(x) < f(y)$.  	
\end{definition}

\begin{definition}[Monotone Graph]
  A bipartite graph $X=([n]_A\sqcup [n]_B,E)$ is a $d$-monotone graph if there are $d$  monotone partial
  maps $f_1,\ldots,f_d : [n] \to [n]$, such that the edges set $E$ is the following disjoint union,
  \[  E = \bigsqcup_{i=1}^d \set{ (v_A,f_i(v)_B) \mid v \in \textup{Domain}(f_i) }.
  \]
\end{definition}

We observe that there are two notions of degree of a monotone graph:
the usual vertex degree and the number of monotone functions.
Clearly, if a graph is $d$-monotone, all vertex degrees are at most
$d$. The converse is not necessarily true; for example, the complete bipartite graph on $2$ vertices on each side, $K_{2,2}$, has vertex degree $2$, but the graph is not $2$-monotone. We stress that our almost Ramanujan bound is with
respect to the usual notion of vertex degree (and keeps the number of
monotone maps polynomial in the vertex degree).

\begin{definition}[Monotone Vertex Expander]
  We say that $X=(A=[n]_A\sqcup B=[n]_B,E)$ is a $d$-monotone expander if it is a $d$-monotone graph
  and there exists $\delta > 0$ such that for all $A' \subseteq A$ with $\abs{A}\le n/2$, we have
  $\abs{\partial(A')} \ge (1+\delta) \abs{A'}$, where $\partial(A')$ is the set of vertices in $B$ adjacent to $A'$.
\end{definition}

\begin{theorem}[Bourgain and Yehudayoff~\cite{BY13}]\label{theo:by_mon_family}
  There is an explicit family $\set{X_n}_{n\in \mathbb{N}}$ of $d$-monotone vertex expanders with $d=\Theta(1)$.
\end{theorem}

We will work with a spectral definition of monotone expander. For a bipartite graph $X$, we define its \textit{biadjacency matrix}, $B_X$ such that the adjacency matrix  $A_X = \left(\begin{matrix}
	0 & B_X\\
	B_X^T & 0
\end{matrix}\right)$. Precisely, for a monotone graph $ X = ( [n]_A \sqcup [n]_B, E)$, we have $(B_X)_{ij} = \mathbb{I}[(i_A,j_B) \in E]$. Note that if $X$ is $d$-regular, then $B_X$ is $d$-regular. We will define the graph $X$ via $B_X$ throughout.

\begin{definition}[Spectral Monotone Expander]
 We say that a $d$-monotone graph, $X$, is a $\lambda$-spectral monotone expander if  
  $\lambda(X) = \max \set{ \abs{ \lambda_2(\matr B_X)},   \abs{ \lambda_n(\matr B_X)} } < \lambda$.	
\end{definition}

It is well-known that starting from a monotone expander (not
necessarily a vertex regular graph), we can add monotone partial
functions to obtain a monotone graph of regular (vertex) degree that
is still expanding.  We use this to establish the following,     

\begin{corollary}\label{cor:spec_base_mon_exp}
  There is explicit family $\set{X_n}_{n\in \mathbb{N}}$ of $d_0$-regular $2d_0$-monotone
  expanders with $\lambda(X_n) \le \expan_0 < 1$ and $d_0=\Theta(1)$. Furthermore, the unnormalized adjacency matrix
  of $X_n$ can be written as a sum of $d_0$ permutation matrices each corresponding to two monotone maps.
\end{corollary}

\begin{proofsketch}
  Let $\set{X_n'}_{n\in \mathbb{N}}$ be the family
  in~\cref{theo:by_mon_family}.  Let $X=X_n'$ be a fixed $d_0$-regfular graph that is also $d_0$-monotone expander with the maps $\{f_i\}$.
  
For each monotone function $f_i$, we define its ``complement'', $\overline{f}_i$, as the (unique) 
monotone \edits{partial} function $\overline{f}_i$ such that $f_i \cup \overline{f}_i$ is a total function. Let $Y$ be the $2d_0$-monotone
graph corresponding to the maps $\{f_i, \overline{f_i}\}$. Then, we have
\begin{align*}
 \matr B_Y &=  \sum_{i=1}^{d_0} \matr P_i\mcom
\end{align*}
where $\matr P_i =  \matr M_{f_i} + \matr M_{\overline{f_i}}$ and $(\matr M_{f_i})_{x,y} = \mathbb{1}\left[f_i(x)=y\right]$.

Each matrix $\matr P_i$ is a permutation matrix as $f_i \cup
\overline{f}_i$ is a total function. Adding more maps preserves the
constant vertex expansion parameter which (together with having
constant vertex degree) implies constant spectral expansion bounded
away from $1$ (see~\cite[Theorem 4.9]{Vadhan12}). Thus, $\set{Y_n}_{n
  \in \N}$ is the required family.
\end{proofsketch}

In the amplification process, we will be multiplying permutation
matrices rather than just composing monotone maps since the latter
operation can result in a map with empty domain. We now establish the
derandomized spectral amplification of monotone expanders.

\TheoMonExp

\begin{proof}
  Let $\set{X_n'}_{n\in \mathbb{N}}$ be the family
  in~\cref{cor:spec_base_mon_exp}. Fix $X=X_n'$ and let $\matr
  P_1,\ldots,\matr P_{d_0} \in \mathbb{R}^{n\times n}$ be the
  permutation matrices guaranteed by~\cref{cor:spec_base_mon_exp},
  where each $\matr P_i = \matr M_{f_i} + \matr M_{\overline{f}_i}$.
  Use~\cref{cor:perm_amp} to obtain a collection $P'$ of size $|P'| = d \coloneqq
  O(1/\expan^{2+\beta}) $ such that,
  \[P' = \big\{ \sigma \mid \sigma = P_{i_1}\cdots P_{i_k} \text{ for some } i_i,\cdots, i_k \in [d_0]\big\}.\]
  
Our final bipartite monotone graph will be $Y$ given by $B_Y = \sum_{\sigma \in P'} \sigma$. 
    To compute its monotone degree, we observe that,
%   where  , and each  permutation $\sigma$ of $P'$ is a product and so we obtain
  \begin{align*}
   \matr P_{i_1}\cdots \matr P_{i_k}~=&~ \sum_{g_i \in \set{f_i,\overline{f_i}}} \matr M_{g_{i_1}} \cdots \matr M_{g_{i_k}}\\
       ~=&~ \sum_{g_i \in \set{f_i,\overline{f_i}}}  \matr M_{g_{i_1}\circ g_{i_2}\circ \,\cdots \;\circ g_{i_k}}\mcom
  \end{align*}
  where $g_{i_1}\circ g_{i_2}\circ \cdots g_{i_k}$ is the composed map
  which is monotone (possibly with an empty domain). This means that we
  can have at most $2^k$ monotone maps (and at least one since $\matr
  P_{i_1}\cdots \matr P_{i_k} \ne 0$). Therefore, the total number of
  maps is at most $d\cdot 2^k = d^{O(1)} $ as
  $k=O_{\expan_0}(\log(1/\expan))$. 
%  This can be made undirected by  adding $f^{-1}$ for each $f$ and thereby doubling the degree.
\end{proof}

%% file: kahzdan.tex
\subsection{Amplifying the Average Kazhdan Constant}\label{subsec:kazhdan}

The \emph{Kazhdan constant} is a notion of ``spectral gap'' (and so it
is related to bias) for discrete groups which predates, and was central
to the study of expansion in finite groups and graphs. {In particular, we can work with finitely generated} groups that
can have infinitely many irreducible representations on more general
Hilbert spaces, possibly of infinite dimension. Nonetheless, we can
still apply our operator version of Ta-Shma's amplification procedure
as it is independent of dimension and works for any unitary
representation $\rho$. Therefore, we amplify the average Kazhdan
constant which also amplifies the Kazhdan constant. We now define
these two parameters formally.

\noindent Let $G$ be a group generated by a finite set $S$ of
generators. The Kazhdan constant of $G$ with respect to generators $S$
is defined as
\begin{align*}
  \mathcal{K}(G,S) &\coloneqq \inf \set{ \mathcal{K}(G,S,\rho) \mid (\rho,\cH) \textup{ irreducible and non-trivial {unitary representation}}}\mcom
\end{align*}
where $\mathcal{K}(G,S,\rho) \coloneqq \inf_{v \in \cH \colon \norm{v}_2=1} \max_{g \in S} \norm{\rho(g) \,v - v}_2^2$.

Analogously, an average version of the Kazhdan constant, as in the
work of Pak and Zuk~\cite{PZ01}, can be defined as
\begin{align*}
\overline{\mathcal{K}}(G,S) &\coloneqq \inf \set{ \overline{\mathcal{K}}(G,S,\rho) \mid (\rho,\cH) \textup{ irreducible and non-trivial {unitary representation}}}\\
  \overline{\mathcal{K}}(G,S,\rho) &\coloneqq \inf_{v \in \cH \colon \norm{v}_2=1} \frac{1}{\abs{S}} \sum_{g \in S} \norm{\rho(g)\, v - v}_2^2 \\
  &= \inf_{v \in \cH \colon \norm{v}_2=1} \frac{1}{\abs{S}} \sum_{g \in S} 2 - 2\ip{\rho(g) \,v}{ v}\\
   &= \inf_{v \in \cH \colon \norm{v}_2=1} 2 - 2 \;\Big\langle{ \Ex{g \sim S}{\rho(g)} v}, \,{ v}\Big\rangle \\
   &= 2 \left( 1 - \Big\Vert{\Ex{g \sim S}{\rho(g)}}\Big\Vert_{\mathrm{op}} \right)\mper
\end{align*}

\cref{theo:main_2} gives an improved generating set in this more general setting.

\TheoKahzdan

\begin{remark}
  Note that the above amplification for $\overline{\mathcal{K}}$ immediately implies
  the same amplification for $\mathcal{K}$ (since the maximum is at least the average, {$\overline{K}(G,S)\leq K(G,S)$)} .
  Moreover, we remark that the above amplification can also similarly improve
  the constant of Lubotzky's property $(\tau)$ (the latter being a weaker version
  of property $(\textup{T})$), so it is more general and applies to expansion
  in many more discrete groups~\cite{RL10}.
\end{remark}

In~\cref{subsec:dim_exp}, we will apply this corollary to a specific
family of representations which will give a simple improvement to the
bounds on the dimension expander constructed in~\cite{LZ08}.

%% file: dim_exp.tex
\subsection{Explicit Almost Ramanujan Dimension Expanders}\label{subsec:dim_exp}

Dimension expanders were defined in~\cite{BISW01} motivated by
applications in theoretical computer science. A conjectured
construction based on irreducible representations was suggested by
Wigderson to hold over every field. The conjecture was subsequently
established by Lubotzky and Zelmanov~\cite{LZ08} for fields of
characteristic zero. We now define dimension expanders, explain
the~\cite{LZ08} proof, and our amplification in this setting.

\begin{definition}[$(\epsilon, \gamma)$ Dimension Expander]
  Let $\F$ be a field, $d \in \mathbb{N}$, $\epsilon > 0$, ${V}$ be a vector space
  of dimension $n$ and $T_1,\ldots,T_d \colon V \to V$ be linear transformations.
  We say that $({V}, \set{T_i}_{i \in [d]})$ is an $(\epsilon, \gamma)$-dimension expander if
  for every subspace ${W} \subseteq {V}$ of dimension at most $\gamma n$,
  we have $\dim({W} + \sum_{i=1}^d T_i({W})) \ge (1+\epsilon) \cdot \dim({W})$.
\end{definition}

\begin{remark}
  Observe that if the maps $T_i$ are restricted to being permutation matrices, and the expansion condition is restricted only
  to subspaces $W$ generated by elementary basis vectors, then one obtains the usual definition of vertex expansion of graphs.
  Thus dimension expanders may be viewed as a linear-algebraic extension of expander graphs.
\end{remark}

\noindent For an irreducible unitary representation $\rho$, there
exists an associated representation\footnote{Let
  $\fksl_n(\C)= \{ \tr(\matr A) = 0\;\vert \; \matr A \in \matr
    M_n(\C)\}$.  Equip the space with the Frobenius inner product
  defined as $\ip{\matr A}{\matr B} = \tr(\matr A^\dagger \matr B)$
  where $\matr A^\dagger$ is the conjugate transpose.
 For any finite-dimensional unitary representation $\rho:G
\to \mathbb{U}_n$, we have an adjoint representation
$(\mathrm{adj}_\rho, \fksl_n)$ where the action is by conjugation
$\mathrm{adj}_\rho (g) \cdot \matr A = \rho(g)\cdot \matr A \cdot
\rho(g)^{-1}$. Since conjugation by unitary matrices preserves the
trace, $\fksl_n$ is closed under the representation. Moreover, it is
unitary as \[\ip{ \mathrm{adj}_\rho (g) \matr A}{\mathrm{adj}_\rho (g)
  \matr B} = \mathrm{tr} \left( \rho(g) \matr A^\dagger \rho(g)^\dagger
\rho(g) \matr B \rho(g)^{-1} \right) = \ip{\matr A}{\matr B}. \]}
$\mathrm{adj}_\rho$.  The construction in \cite{LZ08} relates
dimension expansion with the Kazhdan constant. Their result gives a dimension expander, which gives expansion for all subspaces $W$, such that $\dim(W) \leq n/2$, but their expansion guarantee is significantly stronger when $\dim(W)$ is smaller. To obtain this, we first state a simple improvement to a computation in~\cite{LZ08}.

\begin{claim}\label{clm:project}
  Let ${W},{W}' \subseteq \mathbb{C}^d$ be two vector spaces.
  Let $\matr P,  \matr P'$ be orthogonal projectors onto ${W},{W}'$, respectively.
  Then,  \[
    \mathrm{Re}\left( \Tr(\matr P \matr P')\right) = \Tr(\matr P \matr P') \ge \dim({W} \cap {W}').  \]
\end{claim}
\begin{proof}
  Let $\mathcal{U}_0 = \mathcal{W} \cap \mathcal{W}'$, $\mathcal{U}_1 = \mathcal{W} \cap \mathcal{U}_0^{\perp}$
  and $\mathcal{U}_2 = \mathcal{W}' \cap \mathcal{U}_0^{\perp}$, with orthonormal bases $\set{u_1,\ldots, u_k}$,
  $\set{a_1,\ldots,{a_\ell}}$ and $\set{{b_1},\ldots,{b_m}}$, respectively. We can write $P$ and $P'$ as
  \begin{align*}
    P = \sum_{i=1}^k u_i u_i^T + \sum_{i=1}^{\ell} a_i a_i^T \textup { and } P' = \sum_{i=1}^k u_i u_i^T + \sum_{i=1}^{m} b_i b_i^T.
  \end{align*}
   Using linearity and orthogonality, we obtain
   \begin{align*}
     \Tr(PP') &= \sum_{i=1}^k \norm{u_i}^2\Tr(u_i u_i^T) + \sum_{i=1}^{\ell} \sum_{j=1}^m \ip{a_i}{b_j} \Tr(a_i b_j^T)\\
              &= k + \sum_{i=1}^{\ell} \sum_{j=1}^m\ip{a_i}{b_j}^2   \geq \dim(\mathcal{U}_0),
   \end{align*}
here in the last step we used that $\norm{u_i}^2 = \Tr(u_iu_i^T) = 1$.
 \end{proof}
 
 The above claim is a variant of the one used in ~\cite{LZ08} to prove their main result. By plugging in~\cref{clm:project} in their proof we obtain,

\begin{proposition}[Adapted from~\cite{LZ08} using~\cref{clm:project}]\label{prop:dimension}
  Let $\rho \colon G \to \mathbb{U}_{\mathbb{C}^n}$ be a unitary irreducible representation.
  Then $(\mathbb{C}^n,\set{\rho(g)}_{g \in S})$ is
  $\left( 1-\lambda - o_n(1), 1/2 - O(\lambda) \right)$-dimension expander, where $2(1-\lambda) :=
  \mathcal{K}(G,S,\text{adj}_\rho)$.
%   $\dim(W) \leq n \left(1 - \frac{1}{\mathcal{K}(G,S,\text{adj}_\rho)}\right)$.
 \end{proposition}
 
% Let $G$ be a discrete group and $S$ a finite set of generators such that the average Kazhdan constant $\overline{\mathcal{K}}(G,S)$ is
%  equal to $2\cdot (1-\expan_0)$ for some constant $\expan_0 \in (0,1)$
  
 \begin{corollary}\label{cor:dimension}
 	Let $\lambda > 0$ be any fixed constant. Then, there exists an explicit infinite family of
  $\left( 1-\lambda - o_n(1), \tfrac{1}{2} - O(\lambda) \right)$-dimension expanders.
 \end{corollary}
\begin{proof}
Pick a family of groups $\{G_n\}_n$ such that each $G_n$ satisfies the
condition of~\cref{theo:kazhdan}; for example, one can take any non-abelian finite simple group. By definition, for any such $G$, we have  $\mathcal{K}(G,S, \mathrm{adj}_\rho) \ge
\mathcal{K}(G,S) $ and therefore we obtain a set $S'$ such that
$\mathcal{K}(G,S', \mathrm{adj}_\rho) \geq 2(1-\lambda)$ for the given
$\lambda$. We can now apply~\cref{prop:dimension}.
\end{proof}
%\textbf{Need:} $\mathcal{K}(G,S,\text{adj}_\rho)*(1-\dim({W})/n)\geq 1$

% When $\epsilon$ is arbitrarily close to $0$, we get a dimension expander that is arbitrarily close to being a $(1,1/2)$-dimension expander. 

%
%Also, need to state the restriction on $\dim(\cW)$?}. In fact, we need another
%simple improvement to a computation in~\cite{LZ08}\edits{ yielding
%a stronger lower bound of the trace of the product of two projectors.}
%With the above claim and the analysis in~\cite{LZ08}, we obtain
%stronger dimension expansion for small dimensional spaces. 

\begin{remark}
  Forbes and Guruswami~\cite{FG15} point out that the quantum expander
  construction of Harrow~\cite{Harrow07} also yields a dimension
  expander (with a similar construction of the dimension expanders
  from~\cite{LZ08}). As mentioned earlier, monotone expanders are
  dimension expanders over any field~\cite{DS09,DW10}. Moreover, the
  Bourgain and Yehudayoff~\cite{BY13} construction of monotone
  expanders with constant generating set yields such dimension
  expanders with constant generating set!
\end{remark}

%% file: diameter.tex
\subsection{Diameter of Finite Groups}

The study of the diameter of Cayley graphs can take many forms, \eg it
can be with respect to every generating set (as in the celebrated
Babai--Seress conjecture~\cite{BS88}) or with respect to some constant
size generating set as in~\cite{BabaiKL89}. Here, we explore the
latter case.

First, recall that any $n$-vertex degree-$d$ graph has diameter at
least $\log_{d-1}(n)$.
On the other hand, it is well-known that expansion directly implies
diameter at most $C \cdot \log_{d-1}(n)$ for some constant $C \ge 1$
(depending on the expansion). {The best upper bound a spectral proof can provide is $2$ due to the Alon--Bopanna bound for spectral expansion. Using our amplification to almost-optimal spectral expansion, deduce that any expanding group $G$ has a constant degree-$d$ Cayley expander of diameter $2+ o_d(1)$.}

%Using the operator amplification, we  $\approx 2\cdot \log_{d-1}(\abs{G})$.
More precisely, we have the following.

\begin{lemma}
  Suppose $\set{\Cay(G_i,S_i)}_{i \in \mathbb{N}}$ is a family of bounded degree Cayley expanders.
  Then, there exists a family $\set{\Cay(G_i,S_i')}_{i \in \mathbb{N}}$ of constant degree-$d$ Cayley expanders
  with diameter at most $(2 + o_d(1)) \cdot \log_{d-1}(G_i)$.
\end{lemma}

\begin{proof}
  We apply~\cref{theo:main_2} to the family $\set{\Cay(G_i,S_i)}_{i \in \mathbb{N}}$ obtaining
  a new family of $\set{\Cay(G_i,S_i')}_{i \in \mathbb{N}}$ of $(d,\lambda)$-expanders with
  $d = 1/\lambda^{2+\beta}$ for some sufficiently small constants $\lambda, \beta > 0$.
  Let $\matr A_i$ be the normalized adjacency matrix of $\Cay(G_i,S_i')$ and $n_i = \abs{G_i}$.
  Let $e_g$ be the indicator vector of some fixed $g \in G_i$. 
  Note that \begin{align*}
  \norm{(\matr A_i - \matr J/n_i)^t e_g}_2 \le \lambda^t = d^{- t/(2+\beta)} < 1/\abs{G_i}, \\
\text{for }  t = (2+ 2\beta) \cdot \log_{d}(\abs{G_i}) = (2+ o_{d,\beta}(1)) \cdot \log_{d-1}(\abs{G_i}).
  \end{align*}
  This implies that $\matr A_i^t e_g$ is supported on all elements of $G_i$, and thus the diameter of $G_i$
  is at most $t$.
\end{proof}

%% file: op_eml.tex
\section{Operator Expander Mixing Lemma}\label{sec:eml}

In \cref{sec:simple_amp}, we showed an operator amplification based on
walks on an auxiliary expander. The same construction can be analyzed differently by using the expander mixing lemma (EML) iteratively.

For the scalar case, this is a classic result due to Rozenman--Wigderson (see also~\cite{Bog12}). An operator version of this approach was proved by Chen, Moore, and Russell~\cite{CMR13} but they obtain a much larger degree, $\lambda^{-11}$, rather than the $\lambda^{-4}$ similar to the expander walk approach (\cref{theo:exp_walk}). Our proof obtains the bound of $\lambda^{-4}$.

% proves an operator version of  the
% and applies it in an iterated way (using
%different auxiliary graphs) for bias amplification. They obtain a
%dependence factor We show that this
%approach~\cite{CMR13} can achieve a dependence factor of
%$1/\lambda^{4+o(1)}$ which is similar to the expander walk
%approach (\cref{theo:exp_walk}) 
%We formally
%prove the following result,
%
\begin{restatable}[Iterated Operator EML]{theorem}{ItOpEML}\label{theo:mat_iter_eml}
  Let $S \subseteq G$. Suppose $\lambda(\Cay(G,S))=\expan_0 < 1$, where $\expan_0 \in (0,1)$.
  For every $\expan \in (0,1)$, we can find $S' \subseteq G$ such that,
  \begin{enumerate}[topsep=4pt,itemsep=1pt]
    \item $\lambda(\Cay(G,S')) \le \expan$ and $\abs{S'} = O_{\expan_0}(\abs{S}/\expan^{4+o(1)})$, and
    \item the running time is $\poly(\abs{S},1/\expan_0,1/\expan)$.
  \end{enumerate}  
\end{restatable}

We now show an operator version of the expander mixing lemma for
completeness. As we mentioned above, a similar result was first
derived in~\cite{CMR13}. While a simple generalization of EML, it is
of the same nature of the generalizations of this paper and is of
independent interest.

\begin{restatable}[Matrix EML~\cite{CMR13}]{lemma}{MatrixEML}\label{lemma:matrix_eml}
  Let $X=(V,E)$ be a $\lambda(X)$-spectral expander and let $\matr f \colon V \to \cL(\cH)$. Then,
  \begin{align*}
    \norm{\Ex{(x',x) \in E}{\matr f (x')\, \cdot \matr f (x)\,} - \left(\Ex{x \in V_X}{\matr f (x)\,}\right)^2}_{\textup{op}} ~\le~ \lambda(X) \cdot \max_{x \in V_X} \norm{\matr f (x)\,}_{\textup{op}}^2.
  \end{align*}
\end{restatable}
We start with a simple claim describing an operator form the process
of sampling according to the edges of an expander and sampling
according to pairs of vertices. Recall the following maps from
\cref{sec:simple_amp}, $\projh : \vsX \to
\cH $ and $\lifth
: \cH \to \vsX$, 
\[ \projh(w \otimes x) = w,\; \; \lifth (v) = \Ex{x\in V_X}{ v \otimes x }  \mper \]
We will need again that $\opnorm{\projh}
\opnorm{\lifth} = 1$.

\begin{claim}\label{claim:exp_over_edges}
  Let $\matr A_{X}$ be the normalized adjacency matrix of a $d$-regular graph $X$ and let $ \matr J_X$
  be the normalized  $\abs{V_X} \times \abs{V_X} $ all-ones matrix. 
  \begin{align*}
    \Ex{(x, x') \in E}{\matr f (x')\, \cdot \matr f (x)\,}  &~=~  \projh \matr \Pi_z \AA_{X} \matr \Pi_z \lifth.\\
     \Ex{x,x' \in V}{\matr f (x')\, \cdot \matr f (x)\,} &~=~ \projh \matr \Pi_z \JJ_{X} \matr \Pi_z \lifth.
  \end{align*}  
\end{claim}

\begin{proof}
The proof is identical for both so we prove just the first one.  For any $w \in \cH$, we have
  \begin{align*}
      \projh \matr \Pi_z \AA_{X} \matr \Pi_z \lifth  w  
     &~=~ \frac{1}{|V_X|} \projh \matr \Pi_z \AA_{X} \matr \Pi_z \left(\sum_{x \in V_X} x \otimes \matr  w \right)  \\
     &~=~
 \frac{1}{|V_X|}    \projh \matr \Pi_z \AA_{X} \left(\sum_{x \in V_X} x \otimes \matr f (x)\, w \right).\\
     &~=~
        \frac{1}{d|V_X|} \projh \matr \Pi_z \left(\sum_{x \in V_X}   \sum_{x' \sim x}x' \otimes \matr f (x)\, w \right).\\
     &~=~
    \frac{1}{|E|} \projh \left(\sum_{x \in V_X}   \sum_{x' \sim x}x' \otimes \matr f (x')\, \matr f (x)\, w \right).\\
    &~=~  \frac{1}{|E|} \sum_{x\sim x'}\matr f (x')\, \matr f (x)\, w  =  \Ex{(x', x) \in E}{\matr f (x')\, \cdot \matr f (x)\,} w.
  \end{align*}
  as claimed.
\end{proof}

We now prove the operator mixing lemma above.

\begin{proof}[Proof of~\cref{lemma:matrix_eml}]
 By~\cref{claim:exp_over_edges}, it is enough to bound the operator norm \begin{align*}
  \opnorm{ \projh \matr \Pi_z  \left( \AA_X - \JJ_{X}  \right)\matr \Pi_z \lifth }  &~\le~ \opnorm{ \projh}  \opnorm {\matr \Pi_z}^2 \opnorm{ \left( \AA_X - \JJ_{X}  \right)} \opnorm{ \lifth } \\
   & ~\le~ \lambda(X) \cdot \opnorm{\matr \Pi_z}^2 = \lambda(X) \cdot \max_{x \in V_X} \opnorm{\matr f (x)\,}^2,
 \end{align*}
  concluding the proof.
\end{proof}

\begin{corollary}[Non-abelian EML]\label{cor:non_abelian_eml}
  Let $X=(V,E)$ be a $\lambda(X)$-spectral expander, $\rho \colon G \to \textup{U}_\cH$ be
  an unitary representation and $(g_v)_{v\in V} \in G^V$. Then
  \begin{align*}
    \norm{\Ex{(u,v) \in E}{\rho(g_u) \cdot \rho(g_v)} - \left(\Ex{u \in V}{\rho(g_u)}\right)^2}_{\textup{op}} ~\le~ \lambda(X).
  \end{align*}
\end{corollary}

\begin{proof}
   Follows immediately from~\cref{lemma:matrix_eml} and the fact that unitary operators have operator norm bounded by $1$.
\end{proof}

We now prove the main result of this section. This iterated
amplification also appears in the derandomized squaring of Rozenman
and Vadhan~\cite{RV05} used to give an alternative proof of the
$\textup{SL}=\textup{L}$ result of Reingold~\cite{R04}.

\begin{proof}[Proof of~\cref{theo:mat_iter_eml}]
  We amplify the expansion in two phases. The first phase amplifies the initial
  expansion of $S$ from $\expan_0$ to a \emph{constant} expansion $\expan_0''= 1/4$.
  This phase increases the size of the generator set by a constant factor.
  
  \noindent \textbf{(First Phase)}
  Let $\epsilon_0, \gamma_0$ be constants such that 
  \[ \epsilon_0 = \expan_0(1-\expan_0)/2 , \;\;\;\; 0 < \gamma_0  \le  (1-\expan_0)/2 < 1 \]  
  Let $X_0 = (V_0,E_0)$ be an explicit expander via~\cref{thm:alon_exp}, 
 with $\lambda(X_0) \le \epsilon_0$,
  degree $O(1/\epsilon_0^2)$ and with the number of vertices $\abs{V_0}=m\abs{S}$ with $m=O(1)$.
  Replicate each element of $S$ $m$ times  and still call the resulting multiset $S$ (observe that expansion remains $\expan_0$).
  For every edge $(u,v) \in E_0$, add $g_ug_v$ to $S_0$.  By~\cref{cor:non_abelian_eml}, 
  \[\lambda(G,S_0) \leq \expan_0^2 + \epsilon_0 \leq \expan_0(1-\gamma_0),\;\;\; \abs{S_0} = 9\abs{S}/\epsilon_0^2 =O(\abs{S}) \]
   Repeat this procedure $\log_{1-\gamma_0}{1/4\lambda_0}$ times which
   ensures that the expansion is $\expan_0'' = 1/4$.  Let $S_0$ be
   this final set.

   \noindent \textbf{(Second Phase)} We will amplify the bias inductively using a stronger (\ie more expanding)
   auxiliary expander graph $X_i$ at each step. As mentioned, this inductive amplification is similar to the derandomized
   squaring of Rozenman and Vadhan~\cite{RV05}.
   We start with $S_0$ and  expansion $\lambda_0'' = 2^{-2}$ as in the first phase.
At each step assume that you have a set $S_{i-1}$ with expansion $\expan_{i-1}$.
Use~\cref{thm:alon_exp}, to construct $X_{i-1}$ to have expansion $\lambda_{i-1}^2$ and degree at most
$9/\lambda_{i-1}^4$. Then,  $S_i$ is obtained via edges of $X_i$ as before and we have $\lambda_i \leq 2\lambda_{i-1}^2$.
It is easy to check that the recurrence yields $\lambda_i \leq 2^{-(2^i)}$ for $i \geq 1$.
Assume for convenience that $\log \lambda = -2^r$. Clearly, then we need to iterate this $r$ times.
In  each iteration, the size grows by a factor of the degree which is $9/\lambda_{i-1}^4$ and
thus the final size of $S'$ can be bounded as,
  \begin{align*}
    \abs{S'} ~=~ \abs{S_0} \prod_{i=0}^{r-1} \frac{9}{\lambda_i^4} ~\le~ \abs{S_0} \cdot 9^r 2^{4 + 4\left ( 2^0 + \cdots + 2^{r-1}\right ) } =  \frac{\abs{S_0}}{\lambda^4} \cdot \left ( \frac{1} {\log \lambda }\right )^{\log 9} \le O_{\lambda_0}\left ( \frac{\abs{S}}{\expan^{4 +o(1)}} \right ) \mper
  \end{align*}
  concluding the proof.
\end{proof}

%% file: main.bbl
\newcommand{\etalchar}[1]{$^{#1}$}
\begin{thebibliography}{BASTS08}

\bibitem[ABN{\etalchar{+}}92]{ABNNR92}
N.~Alon, J.~Bruck, J.~Naor, M.~Naor, and R.~Roth.
\newblock Construction of asymptotically good, low-rate error-correcting codes through pseudo-random graphs.
\newblock {\em {IEEE} Transactions on Information Theory}, 28:509--516, 1992.
\newblock \href {https://doi.org/10.1109/18.119713} {\path{doi:10.1109/18.119713}}.

\bibitem[ACKM19]{ACKM19}
Naman Agarwal, Karthekeyan Chandrasekaran, Alexandra Kolla, and Vivek Madan.
\newblock On the expansion of group-based lifts.
\newblock {\em {SIAM} J. Discret. Math.}, 33(3):1338--1373, 2019.
\newblock \href {https://doi.org/10.1137/17M1141047} {\path{doi:10.1137/17M1141047}}.

\bibitem[AGHP92]{AGHP92}
N.~Alon, O.~Goldreich, J.~H{\aa}stad, and R.~Peralta.
\newblock Simple constructions of almost $k$-wise independent random variables.
\newblock {\em Random Structures and Algorithms}, 3(3):289--304, 1992.

\bibitem[AHL{\etalchar{+}}14]{AHLNSV14}
Dorit Aharonov, Aram~W. Harrow, Zeph Landau, Daniel Nagaj, Mario Szegedy, and Umesh~V. Vazirani.
\newblock Local tests of global entanglement and a counterexample to the generalized area law.
\newblock In {\em Proceedings of the 55th IEEE Symposium on Foundations of Computer Science}, 2014.
\newblock \href {https://arxiv.org/abs/1410.0951} {\path{arXiv:1410.0951}}, \href {https://doi.org/10.1109/FOCS.2014.34} {\path{doi:10.1109/FOCS.2014.34}}.

\bibitem[Alo21]{Alon21}
Noga Alon.
\newblock Explicit expanders of every degree and size.
\newblock {\em Combinatorica}, February 2021.
\newblock \href {https://doi.org/10.1007/s00493-020-4429-x} {\path{doi:10.1007/s00493-020-4429-x}}.

\bibitem[ALW01]{ALW01}
N.~Alon, A.~Lubotzky, and A.~Wigderson.
\newblock Semi-direct product in groups and zig-zag product in graphs: connections and applications.
\newblock In {\em Proceedings of the 42nd IEEE Symposium on Foundations of Computer Science}, 2001.

\bibitem[AR94]{AR94}
Noga Alon and Yuval Roichman.
\newblock Random cayley graphs and expanders.
\newblock {\em Random Struct. Algorithms}, 5(2):271--285, 1994.
\newblock \href {https://doi.org/10.1002/rsa.3240050203} {\path{doi:10.1002/rsa.3240050203}}.

\bibitem[AS04]{AS04}
Andris Ambainis and Adam~D. Smith.
\newblock Small pseudo-random families of matrices: Derandomizing approximate quantum encryption.
\newblock In {\em {APPROX}-{RANDOM} 2004 Proceedings}, volume 3122 of {\em Lecture Notes in Computer Science}, pages 249--260. Springer, 2004.
\newblock \href {https://arxiv.org/abs/0404075} {\path{arXiv:0404075}}, \href {https://doi.org/10.1007/978-3-540-27821-4\_23} {\path{doi:10.1007/978-3-540-27821-4\_23}}.

\bibitem[AW02]{AW02}
R.~Ahlswede and A.~Winter.
\newblock Strong converse for identification via quantum channels.
\newblock {\em IEEE Transactions on Information Theory}, 48(3), 2002.

\bibitem[BASTS08]{BAST08}
Avraham Ben-Aroya, Oded Schwartz, and Amnon Ta-Shma.
\newblock Quantum expanders: Motivation and constructions.
\newblock In {\em Proceedings of the 23rd IEEE Conference on Computational Complexity}, pages 292--303. {IEEE} Computer Society, 2008.
\newblock \href {https://doi.org/10.1109/CCC.2008.23} {\path{doi:10.1109/CCC.2008.23}}.

\bibitem[BATS08]{BT08}
Avraham Ben-Aroya and Amnon Ta-Shma.
\newblock A combinatorial construction of almost-ramanujan graphs using the zig-zag product.
\newblock In {\em Proceedings of the 40th ACM Symposium on Theory of Computing}, pages 325--334, 2008.

\bibitem[Bea97]{Beals97}
Robert Beals.
\newblock Quantum computation of fourier transforms over symmetric groups.
\newblock In {\em Proceedings of the 29th ACM Symposium on Theory of Computing}, STOC '97, pages 48--53, 1997.

\bibitem[BISW01]{BISW01}
B.~Barak, R.~Impagliazzo, A.~Shpilka, and A.~Wigderson.
\newblock Dimension expanders.
\newblock unpublished, 2001.

\bibitem[BKL89]{BabaiKL89}
L{\'{a}}szl{\'{o}} Babai, William~M. Kantor, and A.~Lubotzky.
\newblock Small-diameter cayley graphs for finite simple groups.
\newblock {\em Eur. J. Comb.}, 10, 1989.

\bibitem[BL06]{BL06}
Yonatan Bilu and Nathan Linial.
\newblock Lifts, discrepancy and nearly optimal spectral gap.
\newblock {\em Combinatorica}, 26(5):495--519, October 2006.

\bibitem[BL18]{BL18}
Emmanuel Breuillard and Alexander Lubotzky.
\newblock Expansion in simple groups, 2018.
\newblock \href {https://arxiv.org/abs/1807.03879} {\path{arXiv:1807.03879}}.

\bibitem[Bog12]{Bog12}
Andrej Bogdanov.
\newblock Techniques in the theory of computing: Lecture notes.
\newblock Online, 2012.
\newblock Accessed March 3, 2024.
\newblock URL: \url{https://www.andrejb.net/csc5060/notes/12L12.pdf}.

\bibitem[BS88]{BS88}
L{\'{a}}szl{\'{o}} Babai and Akos Seress.
\newblock On the diameter of cayley graphs of the symmetric group.
\newblock {\em Journal of Combinatorial Theory, Series A}, 49(1), 1988.

\bibitem[BY13]{BY13}
Jean Bourgain and Amir Yehudayoff.
\newblock Expansion in $\mathsf{SL}_2{(\mathbb{R})}$ and monotone expanders.
\newblock {\em Geometric and Functional Analysis}, 23(1), 2013.
\newblock \href {https://doi.org/10.1007/s00039-012-0200-9} {\path{doi:10.1007/s00039-012-0200-9}}.

\bibitem[Che10]{Cheng10}
Yuan-You Fu-Rui Cheng.
\newblock Explicit estimate on primes between consecutive cubes.
\newblock {\em Rocky Mountain Journal of Mathematics}, 40(1), February 2010.
\newblock \href {https://arxiv.org/abs/0810.2113} {\path{arXiv:0810.2113}}, \href {https://doi.org/10.1216/rmj-2010-40-1-117} {\path{doi:10.1216/rmj-2010-40-1-117}}.

\bibitem[CMR13]{CMR13}
Sixia Chen, Cristopher Moore, and Alexander Russell.
\newblock Small-bias sets for nonabelian groups - derandomizations of the {A}lon--{R}oichman theorem.
\newblock In {\em {APPROX-RANDOM}}, volume 8096 of {\em Lecture Notes in Computer Science}, pages 436--451, 2013.

\bibitem[DS09]{DS09}
Zeev Dvir and Amir Shpilka.
\newblock Towards dimension expanders over finite fields.
\newblock {\em Combinatorica}, 31(3), sep 2009.
\newblock \href {https://doi.org/10.1007/s00493-011-2540-8} {\path{doi:10.1007/s00493-011-2540-8}}.

\bibitem[DW10]{DW10}
Zeev Dvir and Avi Wigderson.
\newblock Monotone expanders: Constructions and applications.
\newblock {\em Theory of Computing}, 6(12), 2010.
\newblock \href {https://doi.org/10.4086/toc.2010.v006a012} {\path{doi:10.4086/toc.2010.v006a012}}.

\bibitem[FG15]{FG15}
Michael~A. Forbes and Venkatesan Guruswami.
\newblock {Dimension Expanders via Rank Condensers}.
\newblock In {\em Approximation, Randomization, and Combinatorial Optimization. Algorithms and Techniques (APPROX/RANDOM 2015)}, volume~40, pages 800--814, 2015.
\newblock \href {https://arxiv.org/abs/1411.7455} {\path{arXiv:1411.7455}}, \href {https://doi.org/10.4230/LIPIcs.APPROX-RANDOM.2015.800} {\path{doi:10.4230/LIPIcs.APPROX-RANDOM.2015.800}}.

\bibitem[Fri03]{Friedman03}
Joel Friedman.
\newblock A proof of {A}lon's second eigenvalue conjecture.
\newblock In {\em Proceedings of the 35th ACM Symposium on Theory of Computing}, 2003.
\newblock \href {https://arxiv.org/abs/cs/0405020} {\path{arXiv:cs/0405020}}, \href {https://doi.org/10.1145/780542.780646} {\path{doi:10.1145/780542.780646}}.

\bibitem[Gil52]{G52}
E.N. Gilbert.
\newblock A comparison of signalling alphabets.
\newblock {\em Bell System Technical Journal}, 31:504--522, 1952.
\newblock \href {https://doi.org/10.1002/j.1538-7305.1952.tb01393.x} {\path{doi:10.1002/j.1538-7305.1952.tb01393.x}}.

\bibitem[Gil98]{G98}
D.~Gillman.
\newblock A {Chernoff} bound for random walks on expander graphs.
\newblock {\em SIAM Journal on Computing}, 27(4):1203--1220, 1998.
\newblock \href {https://doi.org/10.1137/S0097539794268765} {\path{doi:10.1137/S0097539794268765}}.

\bibitem[Har07]{Harrow07}
Aram~W. Harrow.
\newblock Quantum expanders from any classical cayley graph expander.
\newblock {\em Quantum Information \& Computation}, 2007.

\bibitem[Has07a]{Hastings07b}
M.~B. Hastings.
\newblock Entropy and entanglement in quantum ground states.
\newblock {\em Physical Review B}, 76(3), jul 2007.
\newblock \href {https://arxiv.org/abs/0701055} {\path{arXiv:0701055}}, \href {https://doi.org/10.1103/physrevb.76.035114} {\path{doi:10.1103/physrevb.76.035114}}.

\bibitem[Has07b]{Has07}
M.~B. Hastings.
\newblock Entropy and entanglement in quantum ground states.
\newblock {\em Phys. Rev. B}, 2007.

\bibitem[Has07c]{Hastings07}
M.~B. Hastings.
\newblock Random unitaries give quantum expanders.
\newblock {\em Phys. Rev. A}, 76:032315, Sep 2007.
\newblock \href {https://arxiv.org/abs/0706.0556} {\path{arXiv:0706.0556}}, \href {https://doi.org/10.1103/PhysRevA.76.032315} {\path{doi:10.1103/PhysRevA.76.032315}}.

\bibitem[HH09]{HH09}
M.~B. Hastings and A.~W. Harrow.
\newblock Classical and quantum tensor product expanders.
\newblock {\em Quantum Info. Comput.}, 2009.

\bibitem[HLW06]{HooryLW06}
Shlomo Hoory, Nathan Linial, and Avi Wigderson.
\newblock Expander graphs and their applications.
\newblock {\em Bull. Amer. Math. Soc.}, 43(04):439–562, August 2006.

\bibitem[JM21]{JM21}
Akhil Jalan and Dana Moshkovitz.
\newblock Near-optimal cayley expanders for abelian groups, 2021.
\newblock \href {https://arxiv.org/abs/2105.01149} {\path{arXiv:2105.01149}}.

\bibitem[JMO{\etalchar{+}}22]{JMOPT22}
Fernando~Granha Jeronimo, Tushant Mittal, Ryan O'Donnell, Pedro Paredes, and Madhur Tulsiani.
\newblock Explicit abelian lifts and quantum ldpc codes.
\newblock In {\em ITCS 2022}, 2022.

\bibitem[JQST20]{JQST20}
Fernando~Granha Jeronimo, Dylan Quintana, Shashank Srivastava, and Madhur Tulsiani.
\newblock Unique decoding of explicit $\epsilon$-balanced codes near the {G}ilbert--{V}arshamov bound.
\newblock In {\em Proceedings of the 61st IEEE Symposium on Foundations of Computer Science}, 2020.

\bibitem[Kas07]{Kas07}
Martin Kassabov.
\newblock Symmetric groups and expander graphs.
\newblock {\em Inventiones mathematicae}, 170(2):327--354, August 2007.
\newblock \href {https://arxiv.org/abs/0505624} {\path{arXiv:0505624}}, \href {https://doi.org/10.1007/s00222-007-0065-y} {\path{doi:10.1007/s00222-007-0065-y}}.

\bibitem[KKN21]{KKN21}
Marek Kaluba, Dawid Kielak, and Piotr~W. Nowak.
\newblock {On property (T) for $\mathrm{Aut}(F_n)$ and $\mathrm{SL}_n(\mathbb{Z})$}.
\newblock {\em Annals of Mathematics}, 193(2):539 -- 562, 2021.
\newblock \href {https://doi.org/10.4007/annals.2021.193.2.3} {\path{doi:10.4007/annals.2021.193.2.3}}.

\bibitem[LP00]{LP00}
Alexander Lubotzky and Igor Pak.
\newblock The product replacement algorithm and kazhdan's property (t).
\newblock {\em Journal of the American Mathematical Society}, 14(2):347--363, October 2000.
\newblock \href {https://doi.org/10.1090/s0894-0347-00-00356-8} {\path{doi:10.1090/s0894-0347-00-00356-8}}.

\bibitem[LPS88]{LPS88}
Alexander Lubotzky, R.~Phillips, and Peter Sarnak.
\newblock Ramanujan graphs.
\newblock {\em Combinatorica}, 8:261--277, 1988.

\bibitem[Lub11]{Lub11}
Alexander Lubotzky.
\newblock Finite simple groups of {{Lie}} type as expanders.
\newblock {\em Journal of the European Mathematical Society}, pages 1331--1341, 2011.
\newblock \href {https://arxiv.org/abs/0904.3411} {\path{arXiv:0904.3411}}, \href {https://doi.org/10.4171/JEMS/282} {\path{doi:10.4171/JEMS/282}}.

\bibitem[Lub12]{L12}
Alexander Lubotzky.
\newblock Expander graphs in pure and applied mathematics.
\newblock {\em Bull. Amer. Math. Soc.}, 2012.

\bibitem[LZ08]{LZ08}
Alexander Lubotzky and Efim Zelmanov.
\newblock Dimension expanders.
\newblock {\em Journal of Algebra}, 319(2):730--738, 2008.

\bibitem[Mar73]{Mar73}
G.~A. Margulis.
\newblock Explicit constructions of concentrators.
\newblock {\em Probl. Peredachi Inf.}, 9, 1973.

\bibitem[Mar88]{Margulis88}
G.~A. Margulis.
\newblock Explicit group-theoretical constructions of combinatorial schemes and their application to the design of expanders and concentrators.
\newblock 1988.

\bibitem[MOP20]{MOP20}
Sidhanth Mohanty, Ryan O'Donnell, and Pedro Paredes.
\newblock Explicit near-ramanujan graphs of every degree.
\newblock In {\em Proceedings of the 52nd ACM Symposium on Theory of Computing}, pages 510--523. {ACM}, 2020.

\bibitem[MSS14]{MSS14}
Adam Marcus, Daniel Spielman, and Nikhil Srivastava.
\newblock Interlacing families ii: Mixed characteristic polynomials and the kadison–singer problem.
\newblock {\em Annals of Mathematics}, 2014.

\bibitem[MSS15]{MSS15}
Adam Marcus, Daniel Spielman, and Nikhil Srivastava.
\newblock Interlacing families i: Bipartite {Ramanujan} graphs of all degrees.
\newblock {\em Annals of Mathematics}, 2015.

\bibitem[MW04]{MW04}
Roy Meshulam and Avi Wigderson.
\newblock Expanders in group algebras.
\newblock {\em Combinatorica}, 24, 2004.

\bibitem[Nil91]{Nil91}
Alon Nilli.
\newblock On the second eigenvalue of a graph.
\newblock {\em Discrete Mathematics}, 91(2):207--210, 1991.
\newblock \href {https://doi.org/10.1016/0012-365X(91)90112-F} {\path{doi:10.1016/0012-365X(91)90112-F}}.

\bibitem[NN90]{NN90}
J.~Naor and M.~Naor.
\newblock Small-bias probability spaces: efficient constructions and applications.
\newblock In {\em Proceedings of the 22nd ACM Symposium on Theory of Computing}, pages 213--223, 1990.

\bibitem[Pin73]{Pinsker73}
Mark~S. Pinsker.
\newblock On the complexity of a concentrator.
\newblock In {\em 7th International Teletraffic Conference}, 1973.

\bibitem[PZ01]{PZ01}
Igor Pak and Andrzej Zuk.
\newblock Two {K}azhdan constants and mixing of random walks.
\newblock Technical report, Int. Math. Res. Not. 2002, 2001.

\bibitem[Rei04]{R04}
Omer Reingold.
\newblock Undirected st-connectivity in log-space.
\newblock Technical Report TR04-094, Electronic Colloquium on Computational Complexity, 2004.

\bibitem[Rei05]{R05}
Omer Reingold.
\newblock Undirected {ST}-connectivity in log-space.
\newblock In {\em Proceedings of the 37th ACM Symposium on Theory of Computing}, pages 376--385, 2005.

\bibitem[RL10]{RL10}
J.D. Rogawski and A.~Lubotzky.
\newblock {\em Discrete Groups, Expanding Graphs and Invariant Measures}.
\newblock Modern Birkh{\"a}user Classics. Birkh{\"a}user Basel, 2010.

\bibitem[RSW06]{RSW06}
Eyal Rozenman, Aner Shalev, and Avi Wigderson.
\newblock Iterative construction of cayley expander graphs.
\newblock {\em Theory of Computing}, 2(5):91--120, 2006.

\bibitem[RV05]{RV05}
Eyal Rozenman and Salil Vadhan.
\newblock Derandomized squaring of graphs.
\newblock In {\em Proceedings of RANDOM'05}, pages 436--447. Springer-Verlag, 2005.

\bibitem[RVW00]{RVW00}
O.~Reingold, S.~Vadhan, and A.~Wigderson.
\newblock Entropy waves, the zig-zag graph product, and new constant-degree expanders and extractors.
\newblock In {\em Proceedings of the 41st IEEE Symposium on Foundations of Computer Science}, 2000.

\bibitem[SS96]{serre96}
L.~L. Scott and J.~P. Serre.
\newblock {\em Linear Representations of Finite Groups}.
\newblock Graduate Texts in Mathematics. Springer New York, 1996.

\bibitem[Tro15]{T15}
Joel~A. Tropp.
\newblock An introduction to matrix concentration inequalities.
\newblock {\em Found. Trends Mach. Learn.}, 2015.

\bibitem[TS17]{Ta-Shma17}
Amnon Ta-Shma.
\newblock Explicit, almost optimal, epsilon-balanced codes.
\newblock In {\em Proceedings of the 49th ACM Symposium on Theory of Computing}, STOC 2017, pages 238--251, New York, NY, USA, 2017. ACM.

\bibitem[Vad12]{Vadhan12}
Salil~P. Vadhan.
\newblock {\em Pseudorandomness}.
\newblock Now Publishers Inc., 2012.

\bibitem[Var57]{V57}
R.R. Varshamov.
\newblock Estimate of the number of signals in error correcting codes.
\newblock {\em Doklady Akademii Nauk SSSR}, 117:739--741, 1957.

\bibitem[Wig18]{Wigderson18}
Avi Wigderson.
\newblock Mathematics and computation.
\newblock Book draft at \url{https://www.math.ias.edu/files/mathandcomp.pdf}, 2018.

\end{thebibliography}
